\PassOptionsToPackage{hyperfootnotes=false}{hyperref}
\documentclass[final]{colt2026}
\jmlrpages{}

\usepackage{booktabs}
\usepackage{array}
\usepackage{multirow}
\usepackage{colortbl}
\usepackage{tcolorbox}
\usepackage{microtype}
\usepackage{mathtools}
\usepackage{enumitem}
\usepackage{float}
\usepackage{tikz}
\usepackage{times}
\hypersetup{hypertexnames=false}
\ifdefined\standaloneappendix
  \usepackage{xr-hyper}
\fi
\newcolumntype{L}[1]{>{\raggedright\arraybackslash}p{#1}}
\definecolor{oursrowblue}{RGB}{234,243,252}
\newtcolorbox{researchquestion}{
  colback=blue!3,
  colframe=blue!45!black,
  boxrule=.55pt,
  arc=1.5pt,
  left=5pt,
  right=5pt,
  top=4pt,
  bottom=4pt,
  before skip=7pt,
  after skip=7pt
}

\RestyleAlgo{ruled}
\SetAlgoLined
\LinesNumbered
\DontPrintSemicolon
\SetAlgoNlRelativeSize{-1}
\SetKwInput{KwIn}{Input}
\SetKwInput{KwInit}{Initialize}
\SetKw{KwRet}{return}

\makeatletter
\newenvironment{coltalgorithm}[1][t]{%
  \setcounter{AlgoLine}{0}%
  \begin{algocf@algorithm}[#1]\small
}{%
  \end{algocf@algorithm}%
}
\makeatother

\ifdefined\standaloneappendix
  \title[Poisson Exchange Beyond Submodularity]{Poisson Exchange Beyond Submodularity: Effective Approximation Algorithms for Offline and Online Subset Selection over Matroids --- Appendix}
\else
  \title[Poisson Exchange Beyond Submodularity]{Poisson Exchange Beyond Submodularity: Effective Approximation Algorithms for Offline and Online Subset Selection over Matroids}
\fi

\coltauthor{
 \Name{Shi Fu}$^{1}$
 \and
 \Name{Youming Qiao}$^{2}$
 \and
 \Name{Dacheng Tao}$^{1}$
 \and
 \Name{Zongqi Wan}$^{3}$
 \and
 \Name{Qixin Zhang}$^{1,\dagger}$\\
 \addr $^{1}$College of Computing and Data Science, Nanyang Technological University, Singapore\\
 \addr $^{2}$Centre for Quantum Software and Information, University of Technology Sydney, Australia\\
 \addr $^{3}$School of Computing and Information Technology, Great Bay University, China%
 \footnotetext{Authors are listed in alphabetical order.\\$^{\dagger}$Corresponding author: qixin.zhang2026@gmail.com.}
}
\newcommand{\x}{\boldsymbol{x}}
\newcommand{\E}{\mathbb{E}}
\newcommand{\R}{\mathbb{R}}
\newcommand{\1}{\mathbf{1}}
\newcommand{\cI}{\mathcal{I}}
\newcommand{\cB}{\mathcal{B}}
\newcommand{\cM}{\mathcal{M}}

\newcommand{\OPT}{\mathrm{OPT}}
\newcommand{\MGPE}{\textnormal{\texttt{MGPE}}}

\newcommand{\Reg}{\operatorname{Reg}}

\newcommand{\odotprod}{\mathbin{\odot}}

\DeclareMathOperator*{\argmax}{arg\,max}
\DeclareMathOperator*{\argmin}{arg\,min}

\newtheorem{assumption}{Assumption}
\begin{document}
\maketitle
\ifdefined\standaloneappendix
  \makeatletter
  \let\appendix@originalref\ref
  \renewcommand{\ref}[1]{%
    \ifcsname r@#1\endcsname
      \appendix@originalref{#1}%
    \else
      \appendix@originalref{main-#1}%
    \fi}
  \makeatother
\else

\begin{abstract}
Over the past decade, a growing body of research has shown that
$\gamma$-weak submodularity broadly arises in numerous subset selection tasks,
including feature selection, neural network pruning, and video summarization. Despite its prevalence, maximizing a $\gamma$-weakly submodular function subject to a general matroid constraint
remains challenging. To
date, the only known approximation guarantee is the conservative
$(1+1/\gamma)^{-2}$ factor established by \citet{chen2018weakly}. To improve upon this result, this paper proposes a novel algorithm called \MGPE, which repeatedly
performs maximum-gain local exchanges through careful control of a
non-homogeneous Poisson clock, and proves that this \MGPE\ can attain an approximation
ratio arbitrarily close to $\rho_\gamma=1-\left(\gamma/(2-\gamma)\right)^{
\frac{\gamma^2}{2(1-\gamma)}
}$. In sharp contrast to the
previous guarantee,  our obtained factor
$\rho_\gamma$ not only strictly improves upon $(1+1/\gamma)^{-2}$ for every
$\gamma\in(0,1]$, but also can asymptotically approach the optimal
$(1-1/e)$-approximation for submodular maximization as $\gamma\to1$. Furthermore, we surprisingly find that when the matroid
constraint reduces to a cardinality or the objective satisfies
the stronger notion of $\alpha$-weak DR-submodularity, \MGPE\ can automatically recover the
tight approximation ratios of $1-e^{-\gamma}$ and $1-e^{-\alpha}$,
respectively. Here, $\alpha\in(0,1]$ denotes the DR ratio. Finally, by combining the balanced exchange matrices of
\citet{wan2026banditsubmodularmaximizationmatroid} with two specially designed online oracles, we extend \MGPE\ to the online settings
while preserving the corresponding approximation factors of its offline
counterpart. To the
best of our knowledge, this is the first constant-factor approximation
guarantee for online non-submodular maximization under a
general matroid constraint.
\end{abstract}

\section{Introduction}
Subset selection aims to identify a small set of representative items from a large ground set, which finds
 numerous applications in machine learning, statistics and data mining, including feature selection~\citep{das2018approximate,elenberg2018restricted}, sensor placement~\citep{krause2008near,hashemi2020randomized}, data summarization~\citep{wei2015submodularity,mirzasoleiman2016distributed} and product marketing~\citep{kempe2003maximizing,borgs2014maximizing}. Beyond
the aforementioned representational capability, many of these applications also require the selected items to obey certain structural constraints, such as cardinality  and diversity, which can be naturally modeled by a matroid. Motivated by these findings, this paper considers the problem of maximizing a set function subject to a matroid constraint. Formally, given a finite ground set $U$, a monotone utility function $f:2^U\to\mathbb{R}$, and a matroid $\mathcal{M}\triangleq(U,\mathcal{I})$ where $\mathcal{I}$ denotes the family of feasible subsets, our goal is to find a set $S\in\mathcal{I}$ that maximizes the corresponding utility $f(S)$, i.e., $\max_{S\in\mathcal{I}} f(S)$.

In general, such problem is \textbf{NP}-hard
~\citep{natarajan1995sparse,feige1998threshold} and a large body of work has therefore focused on developing efficient approximation algorithms. In particular, when $f$ is \emph{submodular}, matroid-constrained maximization has been extensively studied and admits the optimal $(1-1/e)$-approximation guarantee~\citep{calinescu2011maximizing,filmus2012tight,filmus2014monotone}. espite this elegant theory, exact submodularity may be overly restrictive for many practical applications. Indeed, many real-world scenarios give rise to utility functions that are \emph{close to submodular}. Examples include variable selection for regression~\citep{das2018approximate,elenberg2018restricted}, video summarization~\citep{chen2018weakly}, neural network pruning~\citep{el2022data}, experimental design~\citep{hashemi2019submodular,harshaw2019submodular} and sparse optimal transport~\citep{manupriya2024submodular}.

One prominent notion for characterizing \emph{close-to-submodular} objectives is \emph{weak submodularity}~\citep{das2011submodular,das2018approximate}. Specifically, a monotone set function $f:2^U\to\mathbb{R}$ is said to be $\gamma$-weakly submodular if, for any $A\subseteq B\subseteq U$ and $\gamma\in(0,1]$, \[\sum_{e\in B\setminus A} f(e\mid A)
\ge
\gamma\bigl(f(B)-f(A)\bigr),\] where $f(B\mid A)\triangleq f(A\cup B)-f(A)$. Notably, when $\gamma=1$, $f$ will reduce to a submodular function.  
Compared to the extensive literature on submodular maximization, there is only a limited amount of research exploring the maximization of $\gamma$-weakly submodular objectives over matroids. In particular, \citet{chen2018weakly} showed that the Residual Random Greedy algorithm~\citep{buchbinder2014submodular} can achieve an approximation guarantee of $(1+1/\gamma)^{-2}$. However, as $\gamma\to1$,  this ratio will approach $1/4$, which is far below the optimal $(1-1/e)$-approximation guarantee for monotone submodular maximization.  To improve this, \citet{thiery2022two} introduced the notion of \emph{upper submodularity ratio} $\beta$, which complements the standard $\gamma$-weak submodularity by characterizing approximate submodularity from the perspective of removing elements. Specifically, a monotone function $f$ is said to be \emph{$\beta$-weakly submodular from above} if, for any $A\subseteq B\subseteq U$, 
\[
\sum_{e\in B\setminus A} f(e\mid B\setminus\{e\})
\le
\beta\bigl(f(B)-f(A)\bigr),
\] where $\beta\ge 1$. Building upon this two-sided submodularity, \citet{thiery2022two} developed a distorted local-search algorithm that achieves an approximation ratio of
$\gamma^2\frac{(1-e^{-\phi})}{\phi}$ for $(\gamma,\beta)$-weakly submodular maximization, where $\phi\triangleq \gamma^2+\beta(1-\gamma)$. Note that this guarantee will recover the optimal  $(1-1/e)$ ratio in the submodular case $(\gamma=\beta=1)$.

However, this stronger approximation result still has several limitations. First, $(\gamma,\beta)$-weakly submodular functions only constitute \emph{a restricted subclass of} the $\gamma$-weakly submodular objectives we seek to address. Second, as shown by \citet{thiery2022two}, even for a fixed $\gamma$, the corresponding upper ratio $\beta$ can become \emph{very large}, which will substantially deteriorate the resulting $(\frac{\gamma^2(1-e^{-\phi})}{\phi})$-approximation. In view of all this, the following question comes to our mind:
\begin{tcolorbox}[
    colback=blue!4,
    colframe=blue!45!black,
    boxrule=0.8pt,
    arc=1.5pt,
    left=6pt,
    right=6pt,
    top=6pt,
    bottom=6pt
]
\textbf{Q1.} Is it possible to design a \emph{constant-factor} approximation algorithm for $\gamma$-weakly submodular maximization over matroids, which can \emph{recover the optimal $(1-1/e)$-guarantee as $\gamma\rightarrow1$}? 
\end{tcolorbox}

In addition, recent years have witnessed growing attention to \emph{online} subset selection over matroids, where the utility function evolves over time and the decision maker must repeatedly select a feasible subset before observing the current objective. Such settings capture numerous applications including online sensor placement~\citep{zhang2025effective,zhang2025nearoptimal}, sequential resource allocation~\citep{streeter2009online,patton2026online} and dynamic batch selection~\citep{pmlr-v162-mindermann22a}. Unfortunately, most existing studies on \emph{online} subset selection focus on submodular objectives and heavily rely on the multilinear extension~\citep{chen2018projection,chen2018online,zhang2019online,zhang2022boosting,wan2023bandit}. Extending these approaches to $\gamma$-weakly submodular functions is challenging, since it is unclear how to \emph{losslessly} round a fractional solution of the multilinear relaxation of $\gamma$-weakly submodular objectives to a feasible set~\citep{thiery2022two}. This naturally motivates us to ask:

\begin{tcolorbox}[
    colback=blue!4,
    colframe=blue!45!black,
    boxrule=0.8pt,
    arc=1.5pt,
    left=6pt,
    right=6pt,
    top=6pt,
    bottom=6pt
]
\textbf{Q2.} Is it possible to design an \emph{online} approximation algorithm with \emph{provable  guarantees} for $\gamma$-weakly submodular maximization over matroids?
\end{tcolorbox}
In this paper, we first answer \textbf{Q1} affirmatively by presenting a novel
maximum-gain Poisson exchange algorithm, called \MGPE, for $\gamma$-weakly submodular maximization over matroids. At a high level, our \MGPE\ is inspired by the
recent GS-Poisson algorithm of \citet{rozenman2026poisson} for monotone
submodular maximization. Its algorithmic core, however, is a fundamentally
different exchange criterion designed to address the analytical difficulties
arising from the lack of submodularity. Specifically, \MGPE\ selects exchanges
using a newly designed maximum-gain rule, which enables us to establish
constant-factor approximation guarantees for the non-submodular objectives
considered in this paper. Furthermore, by combining two specially designed
online linear oracles with the balanced exchange matrices of
\citet{wan2026banditsubmodularmaximizationmatroid}, we extend \MGPE\ to the
online setting. The resulting algorithm preserves the corresponding
approximation factors of its offline counterpart while achieving sublinear
approximation regret, thereby providing an answer to \textbf{Q2}.
Specifically, our main results are summarized as follows.

\subsection{Our Results}
Our first result gives an approximation guarantee for $\gamma$-weakly submodular maximization problems.
\begin{theorem}\label{thm:intro-offline}
Let $O$ be an optimal solution and $\varepsilon\in(0,1)$.  For every
monotone $\gamma$-weakly submodular objective $f$, our proposed \MGPE\ algorithm will output a 
randomized feasible base $S_{\mathrm{out}}$ satisfying:
\begin{enumerate}[label=(\roman*),leftmargin=2.2em]
\item under a general matroid $\cM=(U,\cI)$,
\begin{equation}\label{intro_add}
   \E[f(S_{\mathrm{out}})]
 \ge (\rho_\gamma-\varepsilon)f(O), 
\end{equation}where where $\rho_\gamma\triangleq1-\left(\gamma/(2-\gamma)\right)^{\gamma^2/(2(1-\gamma))}$ when $\gamma\in(0,1)$ and $\rho_1=1-1/e$;
\item under a $k$-size cardinality constraint, $\E[f(S_{\mathrm{out}})]
 \ge \bigl(1-e^{-\gamma}-\epsilon\bigr)f(O)$.
\end{enumerate}
\end{theorem}


At a high level, the key technical step in our analysis is that the
maximum-gain rule employed by \MGPE\ ensures that the total exchange gain
dominates the self-insertion contribution, namely,
$\sum_{i\in S}\nabla_iF(\tau\1_{S-i})$, where $F$ denotes the multilinear
extension of $f$. Consequently, this contribution can be absorbed by the
total exchange gain and removed from the final drift bound, yielding a closed
one-dimensional differential inequality for the expected value process
$Q(\tau)=\E[F(\tau\1_{A(\tau)})]$, where $A(\tau)$ denotes the random base
maintained by \MGPE\ at time $\tau$, that is to say,
\[
Q^{'}_{+}(\tau)
 \ge \frac{\gamma^2}{\gamma+2(1-\gamma)\tau}
       \bigl(f(O)-Q(\tau)\bigr).
\]
Solving this differential inequality gives the approximation guarantee stated
in Theorem~\ref{thm:intro-offline}. Our second result gives stronger guarantees when the objective has additional
structure. 

\begin{theorem}\label{thm:intro-stronger}
Let $O$ be an optimal base of a general matroid $\cM=(U,\cI)$ and fix
$\varepsilon\in(0,1)$.  There exists a randomized algorithm that returns a
base $S_{\mathrm{out}}\in\cI$ satisfying:
\begin{enumerate}[label=(\roman*),leftmargin=2.2em]
\item for every normalized monotone $(\gamma,\beta)$-weakly submodular
objective $f$,
\[
 \E[f(S_{\mathrm{out}})]
 \ge (R_{\gamma,\beta}-\varepsilon)f(O),
\]
where $R_{\gamma,\beta}\triangleq1-\left(\zeta/(2-\zeta)\right)^{\gamma\zeta/(2(1-\zeta))}$ and $\zeta=\max\{\gamma,2-\beta\}$;
\item for every normalized monotone $\alpha$-weakly DR-submodular objective $f$,
\[
 \E[f(S_{\mathrm{out}})]
 \ge (1-e^{-\alpha}-\varepsilon)f(O).
\]
\end{enumerate}
\end{theorem}

Table~\ref{tab:results} compares our guarantees with representative offline
results for non-submodular maximization. Our third result extends all the previous guarantees to the online settings.

\begin{theorem}\label{thm:intro-online}
Fix $\varepsilon\in(0,1)$, let $k$ be the rank of a general matroid $\cM=(U,\cI)$  and
$n=|U|$, if we consider an oblivious sequence of normalized monotone
objectives, namely, $\{f_1,\dots, f_T\}$,
then there exists a randomized online algorithm satisfies
\begin{align*}
 \Reg_{\rho_\gamma-\varepsilon}(T)
 &=O\bigl(k\sqrt{\log\left(n/k\right)T}\bigr)
 &&\text{for $\gamma$-weakly submodular maximization under a matroid},\\
 \Reg_{1-e^{-\gamma(1-\varepsilon)}}(T)
 &=O\bigl(k\sqrt{\log\left(n/k\right)T}\bigr)
 &&\text{for $\gamma$-weakly submodular maximization under cardinality},\\
 \Reg_{R_{\gamma,\beta}-\varepsilon}(T)
 &=O\bigl(k\sqrt{\log\left(n/k\right)T}\bigr)
 &&\text{for $(\gamma,\beta)$-weakly submodular maximization},\\
 \Reg_{1-e^{-\alpha}-\varepsilon}(T)
 &=O\bigl(k\sqrt{\log\left(n/k\right)T}\bigr)
 &&\text{for $\alpha$-weakly DR-submodular maximization},
\end{align*} where the symbol $ \Reg_\rho(T)$ denotes the $\rho$-regret of \citet{chen2018projection,chen2018online} and $\rho\in(0,1]$. 
\end{theorem}

As far as we know, this is the first polynomial-time algorithm to achieve a
constant-factor approximation guarantee for online $\gamma$-weakly submodular maximization under a matroid constraint.

\subsection{Related Work}

Table~\ref{tab:results} compares our guarantees with prior offline results.

\begin{table}[!t]
\centering
\caption{Offline approximation guarantees for non-submodular maximization.
Here, $\OPT=f(O)$,
$\phi\triangleq\gamma^2+\beta(1-\gamma)$,
$\rho_\gamma\triangleq
1-\left(\gamma/(2-\gamma)\right)^{\gamma^2/(2(1-\gamma))}$,
$\zeta\triangleq\max\{\gamma,2-\beta\}$, and
$R_{\gamma,\beta}\triangleq
1-\left(\zeta/(2-\zeta)\right)^{\gamma\zeta/(2(1-\zeta))}$.
The abbreviation ``Sub'' stands for ``submodular''.}
\label{tab:results}
\scriptsize
\setlength{\tabcolsep}{3.2pt}
\renewcommand{\arraystretch}{1.10}
\resizebox{\textwidth}{!}{%
\begin{tabular}{@{}llll@{}}
\toprule
Method & Objective & Constraint type & Guarantee\\
\midrule
Standard Greedy~\citep{das2018approximate}
  & $\gamma$-weakly Sub & Cardinality
  & $(1-e^{-\gamma})\OPT$\\
\midrule
Residual Random Greedy~\citep{chen2018weakly}
  & $\gamma$-weakly Sub & Matroid
  & $(1+1/\gamma)^{-2}\OPT$\\
\midrule
Standard Greedy~\citep{gatmiry2019nonsubmodular}
  & $\gamma$-weakly Sub & Matroid
  & $\dfrac{0.4\gamma^2}{\sqrt{\gamma k}+1}\OPT$\\
\midrule
Continuous Greedy~\citep{gong2019parametric}
  & $\alpha$-weakly DR-Sub & Matroid
  & $\alpha(1-1/e)(1-e^{-\alpha})\OPT$\\
\midrule
Residual Random Greedy~\citep{thiery2022two}
  & $(\gamma,\beta)$-weakly Sub & Matroid
  & $\dfrac{\gamma}{\gamma+\beta}\OPT$\\
\midrule
Distorted local search~\citep{thiery2022two}
  & $(\gamma,\beta)$-weakly Sub & Matroid
  & $\left(\dfrac{\gamma^2(1-e^{-\phi})}{\phi}
    -O(\varepsilon)\right)\OPT$\\
\midrule
Distorted local search~\citep{lu2022maximizing}
  & $\alpha$-weakly DR-Sub & Matroid
  & $(1-e^{-\alpha}-\varepsilon)\OPT$\\
\midrule
\multirow{2}{*}{Multinoulli-SCG~\citep{zhang2026multinoulli}}
  & $(\gamma,\beta)$-weakly Sub & Partition matroid
  & $\dfrac{\gamma^2(1-e^{-\phi})}{\phi}\OPT-\varepsilon$\\
  & $\alpha$-weakly DR-Sub & Partition matroid
  & $(1-e^{-\alpha})\OPT-\varepsilon$\\
\midrule
\rowcolor{oursrowblue}
  & $\gamma$-weakly Sub & Matroid
  & $(\rho_\gamma-\varepsilon)\OPT$\\
\rowcolor{oursrowblue}
  & $\gamma$-weakly Sub & Cardinality
  & $(1-e^{-\gamma}-\varepsilon)\OPT$\\
\rowcolor{oursrowblue}
  & $(\gamma,\beta)$-weakly Sub & Matroid
  & $(R_{\gamma,\beta}-\varepsilon)\OPT$\\
\rowcolor{oursrowblue}
\multirow{-4}{*}{\textbf{Algorithm~\ref{alg:mgpe}~(Ours)}}
  & $\alpha$-weakly DR-Sub & Matroid
  & $(1-e^{-\alpha}-\varepsilon)\OPT$\\
\bottomrule
\end{tabular}%
}
\end{table}

\paragraph{Submodular maximization under a matroid.}For monotone submodular maximization subject to a general matroid constraint,
the classical greedy algorithm achieves a $1/2$-approximation
~\citep{fisher1978analysis}. The optimal $1-1/e$ approximation ratio was later
obtained by continuous greedy combined with pipage or swap rounding
~\citep{vondrak2008optimal,calinescu2011maximizing}. The same guarantee can
also be achieved by purely combinatorial algorithms based on non-oblivious
local search~\citep{filmus2012tight,filmus2014monotone}. An alternative line
of research seeks to combine the continuous and combinatorial viewpoints.
In particular, \citet{rozenman2026poisson} recently proposed the Greedy Swap
Poisson process, which maximizes a monotone submodular function by repeatedly
performing local exchanges governed by a Poisson clock. Beyond the monotone
setting, non-monotone submodular maximization under matroid constraints has
also been extensively studied
~\citep{feldman2011unified,buchbinder2024constrained,buchbinder2025extending,lu2026upperlinearizability,kulik2026spiteful}.

\paragraph{Weakly submodular maximization.}
The $\gamma$-weak  submodularity
were originally introduced by the work~\citep{das2011submodular,das2018approximate}, which also show that the standard greedy algorithm can achieve an approximation ratio of $1-e^{-\gamma}$ for maximizing such functions subject to a cardinality constraint. Moving
beyond cardinality constraints, \citet{chen2018weakly}  subsequently proved that Residual
Random Greedy attains the rank-independent factor
$(1+1/\gamma)^{-2}$ under a general matroid constraint. In contrast,
\citet{gatmiry2019nonsubmodular} analyzed the standard greedy algorithm and
established the rank-dependent approximation guarantee of $\frac{0.4\gamma^{2}}{\sqrt{\gamma k}+1}$ where $k$ is the rank of the matroid.  The limitation of standard greedy was recently
clarified by \citet{ward2026maximizing}, who showed that, for every fixed
$\gamma<1$, its approximation ratio can vanish with the ground-set size $k$ even
under a simple partition matroid. Thus, unlike in the cardinality-constrained
setting, standard greedy admits no constant-factor guarantee under general
matroid constraints. To obtain stronger guarantees, \citet{thiery2022two} introduced the upper submodularity ratio $\beta$ and 
developed a more powerful distorted local search for $(\gamma,\beta)$-weakly submodular maximization, which can guarantee a $\frac{\gamma^{2}(1-e^{-(\beta(1-\gamma)+\gamma^2)})}{\beta(1-\gamma)+\gamma^2}$-approximation for the problem of maximizing a monotone $(\gamma,\beta)$-weakly submodular functions subject to a matroid constraint.  For partition matroids, \citet{zhang2026multinoulli} then introduced
the multinoulli extension and its associated Multinoulli-SCG algorithm can preserve the approximation guarantee of the aforementioned distorted local search with fewer function evaluations.

\paragraph{Weakly DR-submodular maximization.}
The first approximation guarantee for maximizing a monotone
$\alpha$-weakly DR-submodular function under a general matroid constraint was
established by \citet{gatmiry2019nonsubmodular}, whose analysis shows that the
standard greedy algorithm achieves an approximation factor of
$\alpha/(1+\alpha)$. 
A different line of work pursued continuous-relaxation approaches.
In particular, \citet{gong2019parametric} combined continuous greedy with a
contention-resolution scheme~\citep{chekuri2014submodular} to obtain an
asymptotic $\alpha(1-1/e)(1-e^{-\alpha})$-approximation under a general matroid constraint. To achieve the tight $(1-e^{-\alpha})$-approximation guarantee,
\citet{lu2022maximizing} next developed a distorted local-search algorithm that
achieves a $(1-e^{-\alpha}-\varepsilon)$-approximation under an arbitrary
matroid. For partition matroids, \citet{zhang2026multinoulli} subsequently introduced the Multinoulli Extension, a lossless continuous-relaxation framework, and
developed the associated Multinoulli-SCG algorithm, which achieves an approximation ratio of $(1-e^{-\alpha})$ with fewer function evaluations than distorted
local search.

\section{Preliminaries}\label{sec:prelim}

\subsection{Notations and Approximate Submodularity}\label{sec:notations_and_weak}

\paragraph{Notations.}
For a positive integer $m$, we let $[m]=\{1,\ldots,m\}$ and suppose that the finite ground set
is $U$, with $n=|U|$.  For any $A\subseteq U$ and $e\in U$, we write
$A+e=A\cup\{e\}$ and let $\1_A\in\{0,1\}^U$ be
the indicator vector of $A$.  Moreover, the symbols $\langle x,y\rangle$ and $x\odotprod y$ denote the
Euclidean inner product and coordinate-wise product, respectively.  For a set function $f:2^U\to\R_+$, we define
$f(B\mid A)=f(A\cup B)-f(A)$. Unless stated otherwise, in this paper, $f$ is \emph{normalized} with  $f(\varnothing)=0$, and \emph{monotone}, namely,
$f(A)\le f(B)$ for $A\subseteq B\subseteq U$.  

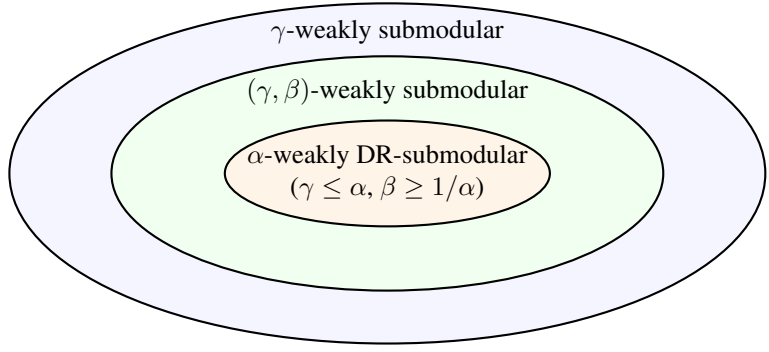
\begin{figure}[t]
\centering
\begin{tikzpicture}[>=stealth]
  \draw[thick,fill=blue!4] (0,0) ellipse (5.0cm and 2.25cm);
  \draw[thick,fill=green!6] (0,0) ellipse (3.65cm and 1.55cm);
  \draw[thick,fill=orange!9] (0,0) ellipse (2.15cm and .70cm);
  \node[font=\small] at (0,1.88) {$\gamma$-weakly submodular};
  \node[font=\small,align=center] at (0,1.08)
    {$(\gamma,\beta)$-weakly submodular};
  \node[font=\small,align=center] at (0,0)
    {$\alpha$-weakly DR-submodular\\
     ($\gamma\le\alpha$, $\beta\ge1/\alpha$)};
\end{tikzpicture}
\caption{Venn diagram of three classes of
approximate submodularity.}
\label{fig:function-classes}
\end{figure}

\begin{definition}[Weak submodularity]
A monotone set function $f$ is $\gamma$-weakly submodular, if for any $A\subseteq B\subseteq U$,
\begin{equation}\label{eq:weak-submod}
 \sum_{e\in B\setminus A}f(e\mid A)\ge \gamma\bigl(f(B)-f(A)\bigr).
\end{equation}
The parameter $\gamma\in(0,1]$ is called submodularity ratio~\citep{das2011submodular,das2018approximate}.
\end{definition}

\begin{definition}[Upper weak submodularity]
A monotone set function $f$ is \emph{$\beta$-weakly submodular from above}, if for any $A\subseteq B\subseteq U$,
\begin{equation}\label{eq:weak-upper}
 \sum_{e\in B\setminus A}f(e\mid B-e)
 \le \beta\bigl(f(B)-f(A)\bigr).
\end{equation}
The parameter $\beta\ge1$ is called upper submodularity ratio. Moreover, we say that a function is \((\gamma,\beta)\)-weakly submodular if it satisfies both \eqref{eq:weak-submod} and \eqref{eq:weak-upper}~\citep{thiery2022two,zhang2026multinoulli}.
\end{definition}

\begin{definition}[Weak DR-submodularity]\label{def:weak-dr}
A monotone set function $f$ is $\alpha$-weakly DR-submodular, if for any $A\subseteq B\subseteq U$,
\begin{equation}\label{eq:weak-dr}
 f(e\mid A)\ge \alpha f(e\mid B).
\end{equation}
The parameter $\alpha\in(0,1]$ is the diminishing-returns(DR) ratio~\citep{lu2022maximizing,zhang2026multinoulli}.
\end{definition}

It is worth noting that the definitions in
Eq.\eqref{eq:weak-submod} and Eq.\eqref{eq:weak-upper} imply, when $\gamma\le\alpha$ and $\beta\ge1/\alpha$, the class
of $\alpha$-weakly DR-submodular functions is contained in the
$(\gamma,\beta)$-weakly submodular class, while the latter is, in turn, a
subclass of the $\gamma$-weakly submodular class, as illustrated in
Figure~\ref{fig:function-classes}.



\subsection{Matroids and Brualdi exchange maps} A \emph{matroid} is a pair $\cM=(U,\cI)$, where $U$ is a finite ground set
and $\cI\subseteq 2^U$ is a collection of sets satisfying the following
axioms: (i)~$\varnothing\in\cI$; (ii)~if $A\subseteq C$ and $C\in\cI$, then
$A\in\cI$; and (iii)~if $A,C\in\cI$ and $|A|<|C|$, then there exists an
element $e\in C\setminus A$ such that $A+e\in\cI$.  The members of $\cI$
are called \emph{independent sets}, and its maximal members are called
\emph{bases}.  Every base has the same cardinality, called the
\emph{rank} of $\cM$, which we denote by $k$. Thus, we can denote the family of
bases as $\cB=\{B\in\cI:|B|=k\}$ and  set its corresponding polytope via $P(\cM)
 \triangleq\operatorname{conv}\{\1_S:S\in\cI\}$.

Furthermore, in our analysis, we will repeatedly use the following Brualdi’s strong basis-exchange
theorem~\citep{brualdi1969comments} about matroids, that is to say,

\begin{lemma}[Brualdi exchange map]\label{lem:brualdi-map}
For every pair of bases $S,O\in\cB$, there exists a bijection
$h_{S,O}:S\to O$ such that
\begin{equation*}
    S-i+h_{S,O}(i)\in\cB,\quad\text{for any}\ \ i\in S.
\end{equation*}Moreover, the bijection can be chosen so that
$h_{S,O}(i)=i$ for every $i\in S\cap O$.
\end{lemma}

\subsection{Multilinear Extension}
 For a set function $f:2^U\to\R_+$, we define its multilinear extension as follows:
\begin{equation}\label{eq:multilinear}
 F(\x)=\sum_{A\subseteq U}f(A)\prod_{e\in A}x_e\prod_{e\notin A}(1-x_e)
     =\E_{R\sim \x}[f(R)],
\end{equation}
where $\x\triangleq(x_1,\dots,x_{|U|})$  and the symbol $R\sim \x$ denotes that each element $e\in U$  is included independently with probability $x_e$.

\subsection{Non-Homogeneous Poisson Process over \texorpdfstring{$[\varepsilon,\infty)$}{[epsilon, infinity)}}\label{sec:poisson-prelim}
 A non-homogeneous Poisson process with rate $\lambda(t)$ is a counting process \(\{N(t)\}_{t\ge\varepsilon}\) characterized by the following properties: (i) \(N(\varepsilon)=0\); (ii) its counts over pairwise disjoint intervals are independent, that is to say, if we denote write
$N_{(a,b]}=N(b)-N(a)$, for any two disjoint intervals \(I\) and \(J\) contained in \([\varepsilon,\infty)\), the random variables
\(N_I\) and \(N_J\) are independent; and (iii) at every time point \(t\ge\varepsilon\), its increments satisfy
\begin{equation}\label{eq:poisson-infinitesimal}
 \lim_{\delta\downarrow0}
 \frac{\Pr\!\left[N_{(t,t+\delta]}=1\right]}{\delta}
 =\lambda(t),
 \qquad
 \lim_{\delta\downarrow0}
 \frac{\Pr\!\left[N_{(t,t+\delta]}\ge2\right]}{\delta}
 =0.
\end{equation}



\section{Maximum-Gain Poisson Exchange}\label{sec:offline}
This subsection is devoted to presenting the Maximum-Gain Poisson Exchange~(\MGPE) algorithm for maximizing a utility set function $f:2^U\to\R_+$ over a matroid \(\cM\triangleq(U,\cI)\). Let \(F\) denote the multilinear extension of \(f\). Throughout this section, we assume exact access to \(F\) and its derivatives.

 As shown in Algorithm~\ref{alg:mgpe}, our proposed \MGPE\ starts from an arbitrary base and evolves over the time interval \([\varepsilon,1]\) via a non-homogeneous Poisson process with rate \(\lambda(t)=k/t\). At each event time \(t\in[\varepsilon,1]\), the \MGPE\ algorithm first draws an element \(i\) uniformly from the current base \(A\). It then searches for a maximum-gain partner \(j_i^{*}\) by maximizing the \emph{leave-one-out multilinear marginal} \(\nabla_jF(t\1_{A-i})\) over all \(j\in\mathcal C_A(i)\), and finally performs the exchange \(A\leftarrow A-i+j_i^{*}\). Here, \(\mathcal C_A(i)\) denotes the feasible exchange neighbors of \(i\), namely, $\mathcal C_A(i)\triangleq\{j\in U: A-i+j\in\cB\}$.

Compared with the earlier exchange framework~\citep{rozenman2026poisson}, the main difference of our proposed \MGPE\ algorithm lies in how an exchange is selected. Under submodularity, the
earlier rule constructs each exchange $(i,j)$ using only the derivatives
evaluated at \(t\1_A\), namely, $\nabla_{j}F(t\1_A)$. In contrast, \MGPE\ explicitly accounts for the
effect of removing \(i\) from the current base \(A\) and selects the
insertion element according to $j_i^*\in
\argmax_{j\in\mathcal C_A(i)}
\nabla_jF(t\1_{A-i})$. This leave-one-out evaluation makes our \MGPE\ applicable to broader classes of
subset selection problems.

\begin{coltalgorithm}[t]
\caption{Maximum-Gain Poisson Exchange~(\MGPE)}\label{alg:mgpe}
\KwIn{Matroid $\cM=(U,\cI)$, start time $\varepsilon\in(0,1)$, and multilinear extension $F$}
\KwInit{Find a base $A_0\in\cI$ and set $A\leftarrow A_0$ and $\tau\leftarrow\varepsilon$}
\While{$t\le1$}{
  Sample a $i$ uniformly from $A$\;
  Compute a  maximum gain $j_i^*$ via maximizing
  $\nabla_jF(\tau\1_{A-i})$ over $j\in\mathcal C_A(i)$\;
  $A\leftarrow A-i+j_i^*$\;
  Sample next random event $\tau\in[\varepsilon,\infty)$ from the Poisson process with rate \(\lambda(\tau)=k/\tau\)\;
  Set $\tau\leftarrow\tau$\;
}
\KwRet $S_{\mathrm{out}}\leftarrow A$\;
\end{coltalgorithm}

\subsection{State-Conditioned Analysis of Maximum-Gain Poisson Exchange}\label{sec:conditioned-analysis}

In the remainder of this section, we focus on analyzing our proposed
\MGPE\ algorithm for the close-to-submodular objectives introduced in Section~\ref{sec:notations_and_weak}. Let \(A(\tau)\) denote the random base maintained by
Algorithm~\ref{alg:mgpe} at time \(\tau\). Following the same proof
strategy of \citep{rozenman2026poisson,kulik2026spiteful}, our analysis mainly tracks the potential function
\(Q(r)\triangleq\E[F(r\1_{A(r)})]\) through its state-conditioned counterpart. Specifically, for any \(S\in\cB\) and \(\varepsilon\le \tau\le r\le1\), we define the conditioned value process of $Q(r)$ as $V_{S,\tau}(r)
:=\E\left[F(r\1_{A(r)})\mid A(\tau)=S\right]$.

For every fixed $\tau\le r$, the law of conditional expectation gives the exact bridge, i.e., 
\begin{equation}\label{eq:Q-V-bridge}
 Q(r)=\sum_{S\in\cB}\big(\Pr[A(t)=S]\cdot V_{S,t}(r)\big).
\end{equation}Consequently, a lower bound on $V_{S,\tau}$ yields a corresponding lower bound on \(Q\). Moreover, according to the procedure of \MGPE, we can show 

\begin{lemma}
\label{lem:conditioned-generator}
For every $\tau\in[\varepsilon,1)$ and $S\in\cB$, the following identity holds:
\begin{align}
  \left.\frac{\partial^+V_{S,\tau}(r)}{\partial r}\right|_{r=\tau}
 =\sum_{i\in S}\nabla_{j_i^*(S,\tau)}F(\tau\1_{S-i}),
\end{align} where $j_i^*(S,\tau)\in\argmax_{j\in\mathcal C_S(i)}
\nabla_jF(\tau\1_{S-i})$. Thus, the unconditional potential function satisfies
\begin{equation}\label{eq:Q-derivative-average}
 \left.\frac{\partial^+Q(r)}{\partial r}\right|_{r=\tau}=\sum_{S\in\cB}\Pr[A(\tau)=S]
 \cdot\left.\frac{\partial^+V_{S,t}(r)}{\partial r}\right|_{r=\tau}=\E_{A(\tau)}\left[\sum_{i\in A(\tau)}\nabla_{j_i^*(A(\tau),\tau)}F(\tau\1_{A(\tau)-i})\right],
\end{equation} where $A(\tau)$ is the random base at time $\tau$  of Algorithm~\ref{alg:mgpe}.
\end{lemma} Note that this Lemma~\ref{lem:conditioned-generator} provides a core analytical foundation for our subsequent analysis.
\subsection{\texorpdfstring{$\gamma$}{gamma}-Weakly Submodular Objectives}\label{sec:offline-weakly-submodular}
Building on Lemma~\ref{lem:conditioned-generator}, this subsection establishes the approximation guarantee of \MGPE\ algorithm for \(\gamma\)-weakly submodular maximization. For notational convenience, let $ C_\gamma(\tau)=\gamma+(1-\gamma)\tau$. At first, we derive an inequality about $\nabla_{j} F(\tau\1_{S-i})$.
\begin{lemma}\label{lem:repair}When $f$ is monotone $\gamma$-weakly submodular, for any $i\in S$ and $j\in U$, we can show that
\begin{equation}\label{eq:repair}
 \nabla_{j} F(\tau\1_{S-i})\ge \frac{\gamma}{C_\gamma(\tau)}\E _{R\sim \tau\1_{S}}\left[f(j\mid R)\right]
 -\frac{(1-\gamma)\tau}{C_\gamma(\tau)}
   \nabla_iF(\tau\1_{S-i}),
\end{equation} where $R$ is the random subset drawn from the point $\tau\1_{S}$.
\end{lemma}
Merging the result of Eq~\eqref{eq:repair} into Lemma~\ref{lem:conditioned-generator}, we can show that  
\begin{lemma}[Weakly submodular drift]\label{lem:gamma-drift}
For every $S\in\cB$ and time $\tau\in[\varepsilon,1]$, when $f$ is $\gamma$-weakly submodular, the conditioned value process  $V_{S,\tau}(r)$ and potential function $Q(r)$ of \MGPE\ satisfies that
\begin{align}
   \left.\frac{\partial^+V_{S,\tau}(r)}{\partial \tau}\right|_{r=\tau}
 &\ge \frac{\gamma^2}{\gamma+2(1-\gamma)\tau}
 \bigl(f(O)-F(\tau\1_S)\bigr) \label{eq:gamma-drift1}\\
 \left.\frac{\partial^+Q(r)}{\partial r}\right|_{r=\tau}&\ge\frac{\gamma^2}{\gamma+2(1-\gamma)\tau}\label{eq:gamma-drift2}
 \bigl(f(O)-Q(\tau)\bigr)
\end{align}
\end{lemma}Finally, by solving the ordinary differential equation~\eqref{eq:gamma-drift2}, we can prove that\begin{theorem}\label{thm:offline-gamma}
For $\varepsilon\in(0,1),\gamma\in(0,1]$ and a matroid $\cM=(U,\cI)$, if $f$ is monotone $\gamma$-weakly submodular, Algorithm~\ref{alg:mgpe} returns a base satisfying
\begin{equation}\label{eq:finite-gamma-ratio}
 \E[f(S_{\mathrm{out}})] \ge (\rho_\gamma-\varepsilon)\max_{S\in\cI}f(S),
\end{equation}where $\rho_\gamma\triangleq1-\left(\gamma/(2-\gamma)\right)^{\gamma^2/(2(1-\gamma))}$ when $\gamma\in(0,1)$ and $\rho_1=1-1/e$.
\end{theorem}
Notably, prior to our work, the best known approximation factor for monotone \(\gamma\)-weakly submodular maximization over a general matroid was \((1+1/\gamma)^{-2}\), achieved by Residual Random Greedy~\citep{chen2018weakly}. Compared with this result, our approximation factor \(\rho_\gamma\) is strictly better, namely, \(\rho_\gamma>(1+1/\gamma)^{-2}\). More importantly, as \(\gamma\) approaches $1$, \(\rho_\gamma\) will converge to the optimal \(1-1/e\) guarantee for monotone submodular maximization~\citep{calinescu2011maximizing}, thereby providing an affirmative answer to \textbf{Q1}. In addition, when the matroid reduces to a rank-\(k\) cardinality constraint, we also find our proposed \MGPE\ can recover the tight
 \(1-e^{-\gamma}\) guarantee established for standard greedy~\citep{das2018approximate,harshaw2019submodular}. Specifically, we can get the following theorem:

\begin{theorem}\label{thm:offline-cardinality}
Suppose that $\cM$ is the rank-$k$ uniform matroid and that $f$ is monotone $\gamma$-weakly submodular.  For every
$\varepsilon\in(0,1)$, Algorithm~\ref{alg:mgpe} returns a base satisfying
\begin{equation}\label{eq:cardinality-finite-start}
 \E[f(S_{\mathrm{out}})]
 \ge \bigl(1-e^{-\gamma}-\varepsilon\bigr)\max_{|S|\le k}f(S).
\end{equation}
\end{theorem} Theorem~\ref{thm:offline-cardinality} shows that, when the matroid constraint reduces to a cardinality, \MGPE\ can achieve  a stronger approximation guarantee of \(1-e^{-\gamma}\). This naturally raises an interesting question:\emph{under general matroid constraints, does the factor \(\rho_\gamma\) established in Theorem~\ref{lem:gamma-drift} reflect an intrinsic limitation of our \MGPE\ algorithm?} Equivalently, beyond the cardinality settings, can a more refined theoretical analysis yield a guarantee strictly larger than \(\rho_\gamma\), or even reach the tight \(1-e^{-\gamma}\)? In the remainder of this subsection, we answer these questions.

\subsubsection{Limitations of \MGPE\ for Weakly Submodular Maximization over Matroids}
\label{sec:offline-limit}

This section focuses on the theoretical limitations of our proposed \MGPE\ for monotone \(\gamma\)-weakly submodular maximization. We first give a negative answer to the question of whether \MGPE\ can attain the \(1-e^{-\gamma}\) guarantee under general matroids by constructing a hard instance, namely,
\begin{proposition}[$1-e^{-\gamma}$ barrier]
\label{prop:strict-barrier}
For $\gamma\in(0,1]$, let
$s_\gamma\triangleq\gamma(2+3\gamma-\gamma^2)/(2(\gamma+2))$. Then, for any 
$0<\delta<(1-s_\gamma)/2$, there exist
 a monotone $\gamma$-weakly submodular and a rank-two partition matroid such
that the return subset of \MGPE\  satisfies
\[ \frac{\E[f(S_{\mathrm{out}})]}{f(O)}=s_\gamma+2\delta,\]where \(O\) is an optimal base. Particularly when  $0<\gamma<0.8576$, $s_\gamma<1-e^{-\gamma}$. Consequently, if we choose $0<\delta<\frac{1-e^{-\gamma}-s_\gamma}{2}$, then $\E[f(S_{\mathrm{out}})]/f(O)<1-e^{-\gamma}$.
\end{proposition}

The construction underlying Proposition~\ref{prop:strict-barrier} uses a ground set consisting of two bad elements \(\{b_1,b_2\}\) and two good elements \(\{g_1,g_2\}\). Moreover, we impose a partition matroid with  two bad-good blocks \(\{b_i,g_i\}\), each of capacity one, and consider a set function whose value only depends on the numbers of blocks in the bad-only, good-only, and double-occupied states. Our high-level idea is to design such objective set function that all-bad base \(B=\{b_1,b_2\}\) forms a strict trap for \MGPE, even though the all-good base \(O=\{g_1,g_2\}\) has a higher objective value.

We next examine how closely the lower bound \(\rho_\gamma\) reflects the true worst-case performance of our proposed \MGPE\ algorithm. To this end, we extend the rank-two construction to \(k\) bad-good blocks \(\{b_i,g_i\}\) and encode the worst-performing instance  as a factor-revealing linear programming. By solving this LP numerically, we finally obtain the Table~\ref{tab:factor-grid} and the following proposition.
\begin{proposition}\label{prop:factor-scope}
Let $R_{\MGPE}(\gamma)$ denote the worst-case approximation ratio of
\MGPE\ for matroid-constrained $\gamma$-weakly submodular maximization problems.
Combining our approximation guarantee with the LP
upper bound in Table~\ref{tab:factor-grid}, we can infer that
$R_{\MGPE}(\gamma)\le1.01\rho_\gamma$ for
$\gamma\in\{0.5,0.7,0.8,0.9\}$,
$R_{\MGPE}(\gamma)\le1.05\rho_\gamma$ for
$\gamma\in\{0.3,0.4\}$, and
$R_{\MGPE}(\gamma)\le1.14\rho_\gamma$ for
$\gamma=0.2$. Consequently, for six of the seven parameter values
reported above, the guarantee $\rho_\gamma$ is tight up to a
multiplicative gap of at most $5\%$; for four of these values, the gap
is below $1\%$.
\end{proposition}

\begin{table}[t]
\centering
\caption{Factor-revealing LP upper bounds on the worst-case approximation
ratio of \MGPE\ for matroid-constrained $\gamma$-weakly submodular
maximization.}
\label{tab:factor-grid}
\small
\setlength{\tabcolsep}{7pt}
\renewcommand{\arraystretch}{1.08}
\begin{tabular}{@{}rrrrr@{}}
\toprule
$\gamma$ & Best tested rank $k$ & $\rho_\gamma$ & LP upper bound & Relative gap\\
\midrule
$0.9$ & $9$  & $0.556349$ & $0.556349$ & $<0.0001\%$\\
$0.8$ & $9$  & $0.477298$ & $0.477300$ & $0.0004\%$\\
$0.7$ & $9$  & $0.396825$ & $0.396841$ & $0.0038\%$\\
$0.5$ & $10$ & $0.240164$ & $0.241344$ & $0.4913\%$\\
$0.4$ & $12$ & $0.168762$ & $0.172847$ & $2.4204\%$\\
$0.3$ & $18$ & $0.105518$ & $0.110495$ & $4.7172\%$\\
$0.2$ & $22$ & $0.053449$ & $0.060742$ & $13.6452\%$\\
$0.1$ & $25$ & $0.016225$ & $0.023332$ & $43.8004\%$\\
\bottomrule
\end{tabular}
\end{table}

In Table~\ref{tab:factor-grid}, the last column reports the relative gap
$(u_{\gamma,k}-\rho_\gamma)/\rho_\gamma$, where $u_{\gamma,k}$ denotes
the certified LP upper bound obtained at rank $k$. Note that all computations were performed in Python 3.12.14 using \texttt{scipy.optimize.linprog} with the HiGHS backend (SciPy 1.18.1 and HiGHS 1.12.0), on a MacBook Pro equipped with an
eight-core Apple M1 Pro processor and 16 GB of memory. Furthermore, the numerical
procedure terminates only when the maximum constraint violation is at most $10^{-9}$. Further details on the factor-revealing LP are provided in Appendix~\ref{appendix:limitation_mgpe}.

The results in Table~\ref{tab:factor-grid} indicate that, for most of the tested parameter $\gamma\in\{0.5,0.7,0.8,0.9\}$, the lower bound $\rho_\gamma$ essentially captures the
worst-case approximation ratio of \MGPE, leaving little room for an improved theoretical analysis. Moreover, the gap remains below $5\%$
for $\gamma\in\{0.3,0.4\}$. However, as $\gamma$ approaches $0$, the gap between LP upper bound and $\rho_{\gamma}$ will grow from $13.65\%$ at $\gamma=0.2$ to $43.8\%$ at $\gamma=0.1$. We attempted
to reduce this discrepancy by increasing the rank $k$ of the factor-revealing
LP, but its size grows rapidly with $k$. At rank $25$, each round requires inspecting $13{,}880{,}880$
nontrivial  constraints, which is prohibitively expensive on our hardware. We therefore leave a sharper
characterization of the small-$\gamma$ regime for future work.
\subsection{Beyond the \texorpdfstring{$\rho_\gamma$}{rho-gamma} Barrier: Poisson Exchange under Refined Weak Submodularity}
\label{sec:beyond-barrier}
In this subsection, we analyze the performance of our proposed poisson exchange on two more structured classes of objectives, namely, $(\gamma,\beta)$-weakly submodular functions and $\alpha$-weakly DR-submodular functions. As illustrated in Figure~1, when \(\gamma \leq \alpha\) and \(\beta \geq 1/\alpha\),  each of these two objectives forms a proper subclass of \(\gamma\)-weakly submodular functions considered above. Consequently, Theorem~1 immediately yields a \(\rho_\gamma\)-approximation guarantee for \MGPE\ over both aforementioned subclasses.  
However,  in the reminder of this subsection, we will show their additional structure allows us to derive sharper approximation guarantees for \MGPE\ than the generic \(\rho_\gamma\) bound.

As in the previous section, we begin by examining the conditioned value function
$V_{S,\tau}(r)$. The following lemma establishes the key lower bounds on its
right derivative at $r=\tau$.
\begin{lemma}\label{lem:repair1}
Let $S$ be a base and $\tau\in[\varepsilon,1]$. Then, when $f$ is monotone $(\gamma,\beta)$-weakly submodular, we can show that
\begin{align}
   \left.\frac{\partial^+V_{S,\tau}(r)}{\partial r}\right|_{r=\tau}
 &\ge \frac{\gamma\zeta}{\zeta+2(1-\zeta)\tau}
    \bigl(f(O)-F(\tau\1_S)\bigr)\label{eq:gamma-drift1+}\\
 \left.\frac{\partial^+Q(r)}{\partial r}\right|_{r=\tau}&\ge\frac{\gamma\zeta}{\zeta+2(1-\zeta)\tau}\label{eq:gamma-drift2+}
 \bigl(f(O)-Q(\tau)\bigr),
\end{align} where  $\zeta \coloneqq \max\{\gamma,2-\beta\}$. Furthermore, if $f$ is monotone $\alpha$-weakly DR-submodular, we also have
\begin{align}
   \left.\frac{\partial^+V_{S,t}(r)}{\partial r}\right|_{r=\tau}
 &\ge \alpha
    \bigl(f(O)-F(\tau\1_S)\bigr)\label{eq:gamma-drift1++}\\
 \left.\frac{\partial^+Q(r)}{\partial r}\right|_{r=\tau}&\ge\alpha\label{eq:gamma-drift2++}
 \bigl(f(O)-Q(\tau)\bigr),
\end{align} 
\end{lemma}
With the help of this Lemma~\ref{lem:repair1}, we can finally prove that
\begin{theorem}\label{thm:offline3}
For $\varepsilon\in(0,1),\gamma\in(0,1]$ and a matroid $\cM=(U,\cI)$, if $f$ is monotone $(\gamma,\beta)$-weakly submodular, Algorithm~\ref{alg:mgpe} returns a base satisfying
\begin{equation}\label{eq:finite-ratio2}
 \E[f(S_{\mathrm{out}})] \ge (R_{\gamma,\beta}-\varepsilon)\max_{S\in\cI}f(S),
\end{equation}where $R_{\gamma,\beta}\triangleq1-\left(\zeta/(2-\zeta)\right)^{\gamma\zeta/(2(1-\zeta))}$ and $\zeta=\max\{\gamma,2-\beta\}$. Moreover, if $f$ is monotone $\alpha$-weakly DR-submodular, we also have
\begin{equation}\label{eq:finite-ratio3}
\E[f(S_{\mathrm{out}})]\ge (1-e^{-\alpha}-\varepsilon)\max_{S\in\cI}f(S).
\end{equation}
\end{theorem}
\begin{remark}Theorem~\ref{thm:offline3} implies that the additional structure encoded by
both weak DR-submodularity and upper weak submodularity can yield stronger guarantees than the generic
$\rho_\gamma$ bound. Specifically, for $\alpha$-weakly DR-submodular
objectives, the resulting factor
$1-e^{-\alpha}$ strictly exceeds $\rho_\gamma$ when $\gamma\leq\alpha$. Similarly, for $(\gamma,\beta)$-weakly submodular functions,
$R_{\gamma,\beta}$ also strictly improves upon $\rho_\gamma$ when
$\beta<2-\gamma$, while recovering $\rho_\gamma$ when
$\beta\geq2-\gamma$.
\end{remark}

\section{Online Subset Selection over Matroids}\label{sec:online}
In this section, we  explore how to extend our proposed \MGPE\ to the online settings. Online subset selection is commonly formulated as a repeated game between a
learner and an adversary. At each round $t\in[T]$, the learner first commits
to  a
feasible subset $S_t\in\cI$ and then receives reward $f_t(S_t)$.  After that, the adversary reveals
the full information of $f_t$ to the learner. The objective of  the learner is to maximize its cumulative reward over
the $T$ rounds. Since the corresponding offline subset selection problem is generally \textbf{NP}-hard~\citep{natarajan1995sparse,feige1998threshold},
we  generally compare the obtained reward with a
\(\rho\)-fraction of the best fixed feasible set in hindsight, that is to say,\begin{equation*}
    \Reg_\rho(T)=
\rho\max_{S\in\cI}\sum_{t=1}^T f_t(S)-
    \sum_{t=1}^T \E[f_t(S_t)],
\end{equation*} where \(\rho\in(0,1]\) and the expectation is taken over the learner's randomness.

Note that the core of our proposed \MGPE\ is the selection of an effective exchange pair $(i,j)$.  Therefore, a natural idea for extending \MGPE\ to the online settings is to maintain a distribution $p_t(i,j)$ for feasible exchange pairs at every time $t$ and then select an exchange according to these probabilities.  Once the objective function $f_t$ is
revealed, the resulting feedback can be used to update these
probabilities. However,  the main difficulty with this approach is that the exchange feasibility will evolve over time. More specifically, when the current
solution changes from $S_t$ to $S_{t+\delta}$, there is no guarantee that the updated
distribution $p_{t+\delta}$ remains feasible for $S_{t+\delta}$. 

To overcome this challenge, we next introduce a novel construction of exchange distributions proposed by \citet{wan2026banditsubmodularmaximizationmatroid}. A key advantage of this method is that it can parameterize feasible exchange
distributions by a fractional point $\x_t$ in the fixed matroid polytope $P(\cM)$, thereby effectively addressing the difficulty arising from the changing feasibility of exchange pairs.

\subsection{Fractional Exchange Matrix}\label{sec:fractional-exchange}
This subsection mainly presents a novel construction of feasible exchange distributions for every pair $(\x,S)$ where $S\in\cB$ and $x\in P(\cM)$. Before going into the details, we introduce the notion of fractional exchange matrix, namely,
\begin{definition}[Fractional exchange matrix]\label{def:balanced}
Fix a base $S\in\cB$ and $\x\in P(\cM)$.  A fractional exchange matrix from $S$ with insertion marginal $\x$ is a nonnegative matrix $Q_S^{\x}\in\R_+^{S\times U}$ satisfying
\begin{enumerate}
\item[(i)]\emph{Feasible support:} $ Q_S^{\x}(i,j)>0\Longrightarrow S-i+j\in\cB$;
\item[(ii)]\emph{Row capacity:} $\sum_{j\in U}Q_S^{\x}(i,j)\le 1$, for every $i\in S$;
   
 \item[(iii)]\emph{Insertion marginals:} $\sum_{i\in S}Q_S^{\x}(i,j)=x_j$,  for every $j\in U$.
\end{enumerate}
\end{definition}\begin{remark}
Definition~\ref{def:balanced} differs slightly in form from the
fractional exchange matrix introduced by
\citet{wan2026banditsubmodularmaximizationmatroid}. Specifically, condition~\textnormal{(ii)} imposes only a row-capacity constraint,
whereas \citet{wan2026banditsubmodularmaximizationmatroid} require every
row sum to be exactly one. This relaxation allows us to define a
fractional exchange matrix for every point $\x\in P(\cM)$, including
points satisfying $\sum_{i}x_i<k$, while their definition is restricted to
full-rank points, i.e., $\sum_{i}x_i=k$. In particular, when $\sum_{i}x_i\equiv k$, we will show that our definition recovers that of
\citet{wan2026banditsubmodularmaximizationmatroid}.
\end{remark}Note that, when $\x\neq\textbf{0}$, the normalized entries $q_S^{\x}(i,j)\triangleq\frac{Q_S^{\x}(i,j)}{\sum_{i\in U} x_{i}}$ naturally define a probability distribution supported on the exchange pairs feasible for $S$. We next establish the existence of such a fractional exchange
matrix $Q_S^{\x}$ and present an efficient algorithm for constructing it.\begin{theorem}[Existence and construction]\label{thm:balanced}
For every $S\in\cB$ and $\x\in P(\cM)$, a fractional exchange matrix $Q_S^{\x}$ exists.  Given a $\x\in P(\cM)$, one can compute $Q_S^{\x}$ via a feasible transportation flow with $k+n$ nonterminal vertices and at most $kn$ exchange arcs.
\end{theorem}Finally, we show two frequently used properties about the fractional exchange
matrices.
\begin{theorem}[Properties of fractional exchange matrices]
\label{thm:fractional-exchange-properties}
Fix a base $S\in\cB$ and a fractional point $\x\triangleq(x_1,\dots,x_{|U|})\in P(\cM)$, the following properties hold:
\begin{enumerate}
    \item[(i)] If $\x$ is full-rank, namely, $\sum_{i\in U}x_i=k$, $\sum_{j\in U}Q_S^{\x}(i,j)=1$ for every $i\in S$;
   \item[(ii)] If $\x=\1_S$, then  $Q_S^{\1_S}(i,j)=\1\{i=j\}$ for every $i\in S$ and $j\in U$.
\end{enumerate}
\end{theorem}\subsection{Online \texorpdfstring{$\gamma$}{gamma}-Weakly Submodular Maximization}\label{sec:online-gamma}With the fractional exchange matrix in hand, we now describe how to extend \MGPE\ to the online settings. Throughout the remainder of this
section, we suppose that each reward function satisfies $f_t\in[0,1]$. Note that the theoretical analysis of \MGPE\ highly relies on an identity about its
conditioned value $V_{S,\tau}(r)$,  i.e., Lemma~\ref{lem:conditioned-generator}. Thus, before presenting the online variant
of our \MGPE\ , we first derive an analogous identity for the exchange
distribution induced by a fractional exchange matrix.

\begin{lemma}[Fractional-exchange generator]
\label{lem:conditioned-generator1}
Fix an online round $t$, a base $S\in\cB$, and a full-rank fractional point
$\x\in P(\cM)$. Let $f_t:2^U\to[0,1]$ be a set function with
multilinear extension $F_t$. Consider the Poisson exchange process
$\{A_t(r)\}_{r\in[\varepsilon,1]}$ with intensity $ \lambda(\tau)=\frac{k}{\tau}$, where each exchange pair is
drawn according to $(i,j)\sim\frac{Q_S^{\x}(i,j)}{k}$.Then, for every $\tau\in[\varepsilon,1)$,
\begin{equation}\label{eq:fractional-generator1}
    \left.
    \frac{\partial^+V^t_{S,\tau}(r)}{\partial r}
    \right|_{r=\tau}
    =\sum_{i\in S}\sum_{j\in U}
    Q_S^{\x}(i,j)
    \nabla_jF_t\bigl(\tau\1_{S-i}\bigr),
\end{equation} where $V^t_{S,\tau}(r)=
    \E\left[
        F_t\bigl(r\1_{A_t(r)}\bigr)
        \,\middle|\,
        A_t(\tau)=S
    \right]$ is the induced conditioned value process.
\end{lemma} To ensure that our online algorithm attains the same approximation factor
as \MGPE, we follow the similar analysis in
Section~\ref{sec:offline-weakly-submodular} and examine the right-hand side
of Eq.\eqref{eq:fractional-generator1}. Specifically, we have
\begin{lemma}\label{add_important} For any $S,O\in\cB$ and $\tau\in[\varepsilon,1]$, when $f_t$ is $\gamma$-weakly submodular, if we set $a_{\gamma}(\tau)\triangleq\frac{\gamma+(1-\gamma)\tau}
        {\gamma+2(1-\gamma)\tau}\in[\frac{1}{2},1]$ and $\mathbf{g}_{t,S,\tau}\triangleq(\1-\tau\1_S)\odotprod\nabla F_t(\tau\1_S)$, then, the following inequality holds:
   \begin{equation}\label{add_15_1}
\begin{aligned}
 &\sum_{i\in S}\sum_{j\in U}\left(a_{\gamma}(\tau)Q_S^{\x}(i,j)+(1-a_{\gamma}(\tau))Q_S^{\1_{S}}(i,j)\right)
    \nabla_jF_t\bigl(\tau\1_{S-i}\bigr)\\&\ge\frac{\gamma^{2}}{\gamma+2(1-\gamma)\tau}(f_t(O)-F_{t}(\tau 1_{S}))+\frac{\gamma}{\gamma+2(1-\gamma)\tau}\langle\x-\1_{O},\mathbf{g}_{t,S,\tau}\rangle.       
\end{aligned}
\end{equation} 
\end{lemma}Note that, compared
with Eq.\eqref{eq:gamma-drift1}, our Eq.\eqref{add_15_1} contains an 
additional linear residual $\bigl\langle \x-\1_O,\mathbf{g}_{t,S,\tau}\bigr\rangle$.
To prevent the cumulative contribution of this gap from becoming too large, a natural idea is to assign an online linear optimization
oracle $\mathcal{E}(\tau)$ for each Poisson event $\tau\in[\varepsilon,1]$. More specifically, at each round $t$, we use the fractional point $\x$
returned by $\mathcal{E}(\tau)$ to construct the exchange
distribution. Once $f_t$ is revealed, we feed
$\mathbf{g}_{t,S,\tau}$, or an unbiased stochastic estimator, back
to this oracle $\mathcal{E}(\tau)$. 

\begin{coltalgorithm}[t]
\caption{Online Poisson Exchange (\texttt{OPE})}\label{alg:ope}
\KwIn{Matroid \(\cM\), rank $k$, horizon \(T\) and  start \(\varepsilon\) }
Find an initial base \(S_0\in\cB\)\;
Sample event \(\varepsilon<\tau_1<\cdots<\tau_N\le1<\tau_{N+1}\) from a Poisson process with rate \(\lambda(\tau)=k/\tau\)\;
Initialize $N$ online maximization oracle $\{\mathcal{E}(1),\dots,\mathcal{E}(N)\}$ over matroid polytope $P(\cM)$\;
Initialize $N$ online weight oracle $\{\mathcal{A}(1),\dots,\mathcal{A}(N)\}$ over range $[1/2,1]$\;
\For{\(t=1,\ldots,T\)}{
  \(S_{t,0}\leftarrow S_0\)\;
 \tcp{\textcolor{teal}{Decision phase}}
 \For{\(r=1,\ldots,N\)}{
Obtain $\x_{t,r}\leftarrow\mathcal{E}(r)$ and $a_{t,r}\leftarrow\mathcal{A}(r)$\;
Compute $Q_{t,r}\leftarrow a_{t,r}Q_{S_{t,r-1}}^{\x_{t,r}}
       +(1-a_{t,r})Q_{S_{t,r-1}}^{\1_{S_{t,r-1}}}$ via Theorem~\ref{thm:balanced}\;
Draw $(i_{t,r},j_{t,r})\sim Q_{t,r}/k$ and set $S_{t,r}\leftarrow S_{t,r-1}-i_{t,r}+j_{t,r}$\;
  }
  Play \(S_t\leftarrow S_{t,N}\) and observe \(f_t\)\;
 \tcp{\textcolor{teal}{Feedback phase}}
  \For{\(r=1,\ldots,N\)}{
    
    Sample an exchange $(\tilde{i},\tilde{j})\sim Q_{S_{t,r-1}}^{\x_{t,r}}/k$ and a subset $\tilde{S}\sim\tau_r\1_{S_{t,r-1}}$\;
Compute $\tilde{\mathbf{g}}_{t,S_{t,r-1},\tau_r}$ with each $j$-th coordinate $\tilde{g}_{t,S_{t,r-1},\tau_r}(j)\triangleq f_{t}(j\mid \tilde{S})$\;
Compute the estimation  $\tilde{\Delta}_{t,r}\triangleq k(f_t(\tilde{S}-\tilde{i}+\tilde{j})-f_t(\tilde{S}))$\;
Feed $\tilde{\mathbf{g}}_{t,S_{t,r-1},\tau_r}$ and $\tilde{\Delta}_{t,r}$ back to \(\mathcal{E}(r)\) and $\mathcal{A}(r)$, respectively\;
  }
}
\end{coltalgorithm}

Furthermore, the exchange term appearing on the left-hand side of
\eqref{add_15_1} is a convex combination of $Q_S^{\x}$ and
$Q_S^{\1_S}$. Since both matrices are valid fractional exchange
matrices,  their convex combination $ a_\gamma(\tau)Q_S^{\x}
    +\bigl(1-a_\gamma(\tau)\bigr)Q_S^{\1_S}$ is naturally defines a feasible exchange distribution. However, the desired mixing coefficient
$a_\gamma(\tau)$ depends highly on the unknown parameter $\gamma$
and is generally unavailable to the learner. To address this issue, we
associate each event time $\tau$ with a one-dimensional online weight
oracle $\mathcal{A}(\tau)$. Specifically, rather than evaluating
$a_\gamma(\tau)$ directly, at each round $t$, the oracle
$\mathcal{A}(\tau)$ returns a weight $a_{t,\tau}$, which is used to
construct the mixed exchange matrix $a_{t,\tau}Q_S^{\x}
+\bigl(1-a_{t,\tau}\bigr)Q_S^{\1_S}$. The oracle subsequently updates its weights using the information of $f_t$.

Putting these ingredients together, we obtain the \texttt{OPE}
algorithm for online $\gamma$-weakly submodular maximization under a
matroid constraint, as presented in Algorithm~\ref{alg:ope}.

Notably, at Line~16 of Algorithm~\ref{alg:ope}, in order to enable each weight
oracle $\mathcal{A}(r)$ to adapt online to the unknown target
$a_\gamma(\tau_r)$,  we feed $\mathcal{A}(r)$ an
unbiased estimator of
\begin{equation}\label{add_delta}
    \Delta_{t,r}\triangleq\sum_{i\in S_{t,r-1}}\sum_{j\in U}
\left(
    Q_{S_{t,r-1}}^{\x_{t,r}}(i,j)
    -
    Q_{S_{t,r-1}}^{\1_{S_{t,r-1}}}(i,j)
\right)
\nabla_jF_t\bigl(\tau_r\1_{S_{t,r-1}-i}\bigr).
\end{equation}This quantity precisely governs how the left-hand side of
Eq.\eqref{add_15_1} varies with the weight $a_\gamma(\tau_r)$.

We next provide a regret analysis of our proposed \texttt{OPE} algorithm.
Before presenting the main result, we impose the following assumptions on the
linear and weight oracles used by Algorithm~\ref{alg:ope}.\begin{assumption}\label{ass:ope-oracles}
Let $n=|U|$ and $k$ be the rank of $\cM$. For each Poisson-event $\tau_r,\forall r\in[N]$, the linear oracle $\mathcal{E}(r)$ and
the weight oracle $\mathcal{A}(r)$ in Algorithm~\ref{alg:ope} satisfy the
following conditions:
\begin{enumerate}
    \item[(i)]At round $t$, the oracle $\mathcal{E}(r)$ outputs a \textbf{full-rank}
    $\x_{t,r}\in P(\cM)$ and subsequently receives the unbiased estimator $\tilde{\mathbf{g}}_{t,S_{t,r-1},\tau_r}\in[0,1]^n$. Then, for every horizon
    $T$ and any comparator $\mathbf u\in P(\cM)$, we have
    \[
        \E\left[
            \sum_{t=1}^T
            \left\langle
                \mathbf u-\x_{t,r},
                \mathbf{g}_{t,S_{t,r-1},\tau_r}
            \right\rangle
        \right]
        \leq O\left(
            k\sqrt{\log\left(n/k\right)T}\right).
    \]
   
    \item[(ii)]
    At round $t$, the oracle $\mathcal{A}(r)$ outputs
    $a_{t,r}\in[1/2,1]$ and subsequently receives the unbiased scalar estimator $\tilde{\Delta}_{t,r}\in[-k,k]$. Then, for every horizon
    $T$ and any comparator $a\in[1/2,1]$,
    \[
        \E\left[
            \sum_{t=1}^T
            \bigl(a-a_{t,r}\bigr)\Delta_{t,r}
        \right]
        \leq O(k\sqrt{T}),
    \] where the definition of $\Delta_{t,r}$ comes from Eq.\eqref{add_delta}.
\end{enumerate} 
\end{assumption}\begin{remark}
There exist randomized online linear algorithms satisfying
Assumption~\ref{ass:ope-oracles}. We provide their explicit constructions and regret analyses in
Appendix~\ref{sec:online-linear-oracles}.
\end{remark}Under Assumption~\ref{ass:ope-oracles}, we can obtain the following expected regret guarantee, i.e., 

\begin{theorem}\label{thm:ope-gamma-regret}
Fix $\varepsilon\in(0,1)$, matroid $\mathcal{M}$ and suppose that each
$f_t:2^U\to[0,1]$ is monotone 
$\gamma$-weakly submodular for any $t\in[T]$. Under Assumption~\ref{ass:ope-oracles}, then Algorithm~\ref{alg:ope} satisfies
\begin{equation}\label{eq:ope-gamma-regret}
\Reg_{\rho_\gamma-\varepsilon}(T)
    =
    O\left( k\sqrt{\log\left(n/k\right)T}
    \right),
\end{equation} where $\rho_\gamma\triangleq1-\left(\gamma/(2-\gamma)\right)^{\gamma^2/(2(1-\gamma))}$ when $\gamma\in(0,1)$ and $\rho_1=1-1/e$.
\end{theorem}

\subsection{Online \texorpdfstring{$(\gamma,\beta)$}{(gamma,beta)}-Weakly Submodular Maximization}\label{sec:online-two-sided}

We now turn to online maximization of monotone
$(\gamma,\beta)$-weakly submodular objectives.  Following the preceding analysis, we first establish a result analogous to
Lemma~\ref{add_important}, that is to say,

\begin{lemma}\label{add_important1}
For any $S,O\in\cB$, $\tau\in[\varepsilon,1]$, and any full-rank
fractional point $\x\in P(\cM)$, suppose that $f_t$ is monotone and
$(\gamma,\beta)$-weakly submodular. Let
$\zeta\coloneqq\max\{\gamma,2-\beta\}$, $a_\zeta(\tau)=
    \frac{\zeta+(1-\zeta)\tau}
         {\zeta+2(1-\zeta)\tau}
    \in(1/2,1]$, and define $\mathbf{g}_{t,S,\tau}=
    (\1-\tau\1_S)\odotprod\nabla F_t(\tau\1_S)$.
Then the following inequality holds:
\begin{equation}\label{eq:online-two-sided-fractional-drift}
\begin{aligned}
&\sum_{i\in S}\sum_{j\in U}
    \left(a_{\zeta}(\tau)Q_S^{\x}(i,j)
        +\bigl(1-a_{\zeta}(\tau)\bigr)Q_S^{\1_S}(i,j)
    \right)\nabla_jF_t(\tau\1_{S-i})                                    \\&\geq
    \frac{\gamma\zeta}
         {\zeta+2(1-\zeta)\tau}
    \bigl(f_t(O)-F_t(\tau\1_S)\bigr)
    +
    \frac{\zeta}
         {\zeta+2(1-\zeta)\tau}
    \left\langle
        \x-\1_O,\mathbf{g}_{t,S,\tau}
    \right\rangle .
\end{aligned}
\end{equation}
\end{lemma}Note that Eq~\eqref{eq:online-two-sided-fractional-drift} has the same
structural form as Eq.\eqref{add_15_1}. Therefore, by applying the same
oracle-based analysis, we obtain the following regret guarantee for online
$(\gamma,\beta)$-weakly submodular maximization.\begin{theorem}\label{thm:ope-two-sided-regret}
Fix $\varepsilon\in(0,1)$ and a rank-$k$ matroid
$\cM=(U,\cI)$. Suppose that, for every $t\in[T]$,
$f_t:2^U\to[0,1]$ is $(\gamma,\beta)$-weakly submodular. Under (i) of Assumption~\ref{ass:ope-oracles}, Algorithm~\ref{alg:ope} satisfies
\begin{equation}\label{eq:ope-two-sided-regret}
    Reg_{R_{\gamma,\beta}-\varepsilon}(T)
    =
    O\left( k\sqrt{\log\left(n/k\right)T}
    \right),
\end{equation} where $R_{\gamma,\beta}\triangleq1-\left(\zeta/(2-\zeta)\right)^{\gamma\zeta/(2(1-\zeta))}$ and $\zeta=\max\{\gamma,2-\beta\}$.
\end{theorem}

\subsection{Online \texorpdfstring{$\gamma$}{gamma}-weakly submodular maximization over Cardinality Constraints}\label{sec:online-cardinality}

\begin{coltalgorithm}[t]
\caption{Online Poisson Exchange over Cardinality  (\texttt{OPEC})}\label{alg:opec}
\KwIn{$k$-size uniform matroid \(\cM\), horizon \(T\) and  start \(\varepsilon\) }
Find an initial base \(S_0\in\cB\)\;
Sample event \(\varepsilon<\tau_1<\cdots<\tau_N\le1<\tau_{N+1}\) from a Poisson process with rate \(\lambda(\tau)=k/\tau\)\;
Initialize $N$ online maximization oracle $\{\mathcal{E}(1),\dots,\mathcal{E}(N)\}$ over matroid polytope $P(\cM)$\;
\For{\(t=1,\ldots,T\)}{
  \(S_{t,0}\leftarrow S_0\)\;
 \tcp{\textcolor{teal}{Decision phase}}
 \For{\(r=1,\ldots,N\)}{
Obtain $\x_{t,r}\leftarrow\mathcal{E}(r)$ and compute $\bar{Q}_{t,r}\leftarrow \bar{Q}_{S_{t,r-1}}^{\x_{t,r}}$ via Eq.~\eqref{eq:card-online-Q}\;
Draw $(i_{t,r},j_{t,r})\sim \bar{Q}_{t,r}/k$ and set $S_{t,r}\leftarrow S_{t,r-1}-i_{t,r}+j_{t,r}$\;
  }
  Play \(S_t\leftarrow S_{t,N}\) and observe \(f_t\)\;
 \tcp{\textcolor{teal}{Feedback phase}}
  \For{\(r=1,\ldots,N\)}{
Sample a subset $\tilde{S}\sim\tau_r\1_{S_{t,r-1}}$\;
\For{$j=1,\ldots,n$}{Sample $\tilde{i}_j\in S_{t,r-1}$ from the distribution $p(\x_{t,r},S_{t,r-1})$ of Eq.~\eqref{eq:card-online-Q}\;
Compute $j$-th estimation $\tilde{\ell}^{\mathrm{card}}_{
            t,S_{t,r-1},\tau_r,\x_{t,r}
        }(j)
        \leftarrow
        f_t\bigl(
            j\mid
            \widetilde{S}\setminus\{\tilde{i}_j,j\}
        \bigr)$\;

}Feed $\tilde{\ell}^{\mathrm{card}}_{
            t,S_{t,r-1},\tau_r,\x_{t,r}
        }$ back to \(\mathcal{E}(r)\)\;
  }
}
\end{coltalgorithm}

Recall that, when the matroid constraint reduces to cardinality constraint, our proposed \MGPE\ algorithm will recover the tight \((1-e^{-\gamma})\)-approximation guarantee for monotone \(\gamma\)-weakly submodular maximization. In this subsection, we extend this offline result to the online setting.

In view of the simple structure of cardinality constraint, in this subsection, we consider a specialized exchange matrix for every base \(S\) of size \(k\) and every  full-rank fractional point \(\x\), i.e., $\sum_{i=1}^{n}x_i=k$. Specifically, we define
\begin{equation}\label{eq:card-online-Q}
\begin{alignedat}{2}
\bar{Q}_S^{\x}(i,j)
&\triangleq
\begin{cases}
x_i,     & j=i,\\
0,       & j\in S\setminus\{i\},\\
p_i x_j, & j\notin S,
\end{cases}
&\qquad
p_i(\x,S)
&\triangleq
\begin{cases}\frac{1-x_i}{\sum_{h\in S}(1-x_h)},
& \sum_{h\in S}x_h>k,\\[1mm]\quad 1/k,
&\sum_{h\in S}x_h=k.
\end{cases}
\end{alignedat}
\end{equation}It is straightforward to verify  \(\bar{Q}_S^{\x}\) is a valid fractional exchange matrix for the rank-\(k\) uniform matroid In particular, every row of \(\bar{Q}_S^{\x}\) sums to one, its \(j\)-th column sums to \(x_j\), and every exchange in its support is feasible. With this special exchange matrix, we next establish a result analogous to Lemma~\ref{add_important1} for cardinality-constrained \(\gamma\)-weakly submodular maximization.
\begin{lemma}\label{lem:32} Fix any \(S,O\in\cB\), \(\tau\in[\varepsilon,1]\), and \(\x\in P(\cM)\), where \(\cM\) is the rank-\(k\) uniform matroid. Suppose that \(f_t\) is monotone \(\gamma\)-weakly submodular. For every \(j\in U\), define \begin{equation*}\ell^{\mathrm{card}}_{t,S,\tau,\x}(j)\triangleq\begin{cases}\sum_{i\in S}p_i(\x,S)\nabla_jF_t(\tau\1_{S-i}),& j\notin S\\[.5mm]\quad \nabla_jF_t(\tau\1_{S-j})& j\in S\end{cases}\end{equation*}Then, we can prove that

\begin{equation}\label{eq:online-card-drift}\sum_{i\in S}\sum_{j\in U}\bar Q_S^{\x}(i,j)\nabla_jF_t\bigl(\tau\1_{S-i}\bigr)\ge\gamma\bigl(f_t(O)-F_t(\tau\1_S)\bigr)+\langle\x-\1_O,\ell^{\mathrm{card}}_{t,S,\tau,\x}\rangle.\end{equation}\end{lemma}
Motivated by this Eq.\eqref{eq:online-card-drift}, we then design an \texttt{OPEC} algorithm for online $\gamma$-weakly DR-submodular maximization over cardinality constraints, as shown in Algorithm~\ref{alg:opec}. Next, we impose the following assumption on the
linear oracles used by Algorithm~\ref{alg:opec}.

\begin{assumption}\label{ass:ope-oracles2}
Let $n=|U|$ and $\cM$ be a $k$-size uniform matriod. For each Poisson-event $\tau_r,\forall r\in[N]$, the linear oracle $\mathcal{E}(r)$ in Algorithm~\ref{alg:opec} satisfy: At round $t$, the oracle $\mathcal{E}(r)$ outputs a \textbf{full-rank}
    $\x_{t,r}\in P(\cM)$ and subsequently receives the unbiased estimator $\tilde{\ell}^{\mathrm{card}}_{
            t,S_{t,r-1},\tau_r,\x_{t,r}
        }\in[0,1]^n$. Then, for every horizon
    $T$ and any comparator $\mathbf u\in P(\cM)$, we have
    \[
        \E\left[
            \sum_{t=1}^T
            \left\langle
                \mathbf u-\x_{t,r},
                \ell^{\mathrm{card}}_{
            t,S_{t,r-1},\tau_r,\x_{t,r}
        }
            \right\rangle
        \right]
        \leq O\left(
            k\sqrt{\log\left(n/k\right)T}\right).
    \]

\end{assumption}As a result, we can have\begin{theorem}\label{thm:ope-card-regret}
Fix $\varepsilon\in(0,1)$ and a rank-$k$ uniform matroid
$\cM=(U,\cI)$. Suppose that, for every $t\in[T]$,
$f_t:2^U\to[0,1]$ is $\gamma$-weakly submodular. Under 
Assumption~\ref{ass:ope-oracles2}, Algorithm~\ref{alg:opec} satisfies
\begin{equation}\label{eq:ope-card-regret}
    Reg_{(1-e^{-\gamma}-\varepsilon)}(T)
    =
    O\left( k\sqrt{\log\left(n/k\right)T}
    \right).
\end{equation}
\end{theorem}

\subsection{Online \texorpdfstring{$\alpha$}{alpha}-Weakly DR-Submodular Maximization}\label{sec:online-alpha}
\begin{coltalgorithm}[t]
\caption{Online Poisson Exchange for weak DR-submodularity (\texttt{OPE-DR})}\label{alg:ope-dr}
\KwIn{Matroid \(\cM\), rank $k$, point $\x^{o}$ horizon \(T\) and  start \(\varepsilon\) }
Find an initial base \(S_0\in\cB\)\;
Sample event \(\varepsilon<\tau_1<\cdots<\tau_N\le1<\tau_{N+1}\) from a Poisson process with rate \(\lambda(\tau)=k/\tau\)\;
Initialize $N$ online maximization oracle $\{\mathcal{E}(1),\dots,\mathcal{E}(N)\}$ over matroid polytope $P(\cM)$\;
\For{\(t=1,\ldots,T\)}{
  \(S_{t,0}\leftarrow S_0\)\;
 \tcp{\textcolor{teal}{Decision phase}}
 \For{\(r=1,\ldots,N\)}{
Obtain $\x_{t,r}\leftarrow\mathcal{E}(r)$\;
Compute $Q_{t,r}\leftarrow Q_{S_{t,r-1}}^{\x_{t,r}}$ via Theorem~\ref{thm:balanced}\;
Draw $(i_{t,r},j_{t,r})\sim Q_{t,r}/k$ and set $S_{t,r}\leftarrow S_{t,r-1}-i_{t,r}+j_{t,r}$\;
  }
  Play \(S_t\leftarrow S_{t,N}\) and observe \(f_t\)\;
 \tcp{\textcolor{teal}{Feedback phase}}
  \For{\(r=1,\ldots,N\)}{
Sample a subset $\tilde{S}\sim\tau_r\1_{S_{t,r-1}}$\;
\For{$j=1,\ldots,n$}{Sample $\tilde{i}_j\in S_{t,r-1}$ via the distribution $\frac{Q_{S_{t,r-1}}^{\x_{t,r}}(\cdot,j)}{y_{t,r,j}}$\;
Compute $j$-th estimation $\tilde{\ell}_{
            t,S_{t,r-1},\tau_r,\x_{t,r}
        }(j)
        \leftarrow
        f_t\bigl(
            j\mid
            \widetilde{S}\setminus\{\tilde{i}_j,j\}
        \bigr)$\;

}Feed $\tilde{\ell}_{
            t,S_{t,r-1},\tau_r,\x_{t,r}
        }$ back to \(\mathcal{E}(r)\)\;
  }
}
\end{coltalgorithm}

As in the previous analyses, we first establish a result analogous to
Lemma~\ref{add_important1} for weak DR-submodularity, namely,\begin{lemma}\label{lem:online-alpha-drift}
Fix any $S,O\in\cB$, $\tau\in[\varepsilon,1]$, and
any positive full-rank $\x\in P(\cM))$. Suppose that $f_t$ is monotone and
$\alpha$-weakly DR-submodular. For each $j\in U$, we define
\begin{equation}\label{equ:ell}
\begin{aligned}
\ell_{t,S,\tau,\x}(j)\triangleq
\sum_{i\in S}
\frac{Q_S^{\x}(i,j)}{x_j}
\nabla_j F_t\bigl(\tau\1_{S-i}\bigr).
\end{aligned}
\end{equation}Then, we can have that 
\begin{align}
\sum_{i\in S}\sum_{j\in U}
Q_S^{\x}(i,j)
\nabla_j F_t\bigl(\tau\1_{S-i}\bigr)
\ge\alpha\bigl(
f_t(O)-F_t(\tau\1_S)
\bigr)+
\left\langle
\x-\1_O,\ell_{t,S,\tau,\mathbf{y}}
\right\rangle.
\label{eq:online-alpha-fractional-drift}
\end{align}
\end{lemma}
Motivated by this Eq.\eqref{eq:online-alpha-fractional-drift}, we then design an \texttt{OPE-DR} algorithm for online $\alpha$-weakly DR-submodular maximization, as shown in Algorithm~\ref{alg:ope-dr}. Next, we impose the following assumption on the
linear oracles used by Algorithm~\ref{alg:ope-dr}.\begin{assumption}\label{ass:ope-oracles1}
Let $n=|U|$ and $k$ be the rank of $\cM$. For each Poisson-event $\tau_r,\forall r\in[N]$, the linear oracle $\mathcal{E}(r)$ in Algorithm~\ref{alg:ope-dr} satisfy: At round $t$, the oracle $\mathcal{E}(r)$ outputs a \textbf{full-rank} and \textbf{positive}
    $\x_{t,r}\in P(\cM)$ and subsequently receives the unbiased estimator $\tilde{\ell}_{
            t,S_{t,r-1},\tau_r,\x_{t,r}
        }\in[0,1]^n$. Then, for every horizon
    $T$ and any comparator $\mathbf u\in P(\cM)$, we have
    \[
        \E\left[
            \sum_{t=1}^T
            \left\langle
                \mathbf u-\x_{t,r},
                \ell_{
            t,S_{t,r-1},\tau_r,\x_{t,r}
        }
            \right\rangle
        \right]
        \leq O\left(
            k\sqrt{\log\left(n/k\right)T}\right).
    \]

\end{assumption}Thus, we can obtain the following regret bound for online
$\alpha$-weakly DR-submodular maximization.\begin{theorem}\label{thm:ope-DR-regret}
Fix $\varepsilon\in(0,1)$ and a rank-$k$ matroid
$\cM=(U,\cI)$. Suppose that, for every $t\in[T]$,
$f_t:2^U\to[0,1]$ is monotone
$\alpha$-weakly DR-submodular. Under 
Assumption~\ref{ass:ope-oracles1}, Algorithm~\ref{alg:ope-dr} satisfies
\begin{equation}\label{eq:ope-alpha-regret}
    Reg_{(1-e^{-\alpha}-\varepsilon)}(T)
    =
    O\left( k\sqrt{\log\left(n/k\right)T}
    \right).
\end{equation}
\end{theorem}

\section{Conclusion}
This paper introduced maximum-gain poisson exchange~(\MGPE) algorithm for maximizing
$\gamma$-weakly submodular functions subject to a matroid
constraint. By repeatedly performing maximum-gain local exchanges through a
non-homogeneous Poisson clock, \MGPE\ achieves an approximation ratio
arbitrarily close to $\rho_\gamma$, strictly improving upon the previous
$(1+1/\gamma)^{-2}$ guarantee and asymptotically recovering the optimal
$1-1/e$ ratio as $\gamma\to1$. For $(\gamma,\beta)$-weakly submodular
objectives, we further establish an improved approximation ratio
$R_{\gamma,\beta}$, which is strictly larger than $\rho_\gamma$ when $\beta<2-\gamma$. Furthermore, we also find that \MGPE\ will recover the tight approximation
ratios $1-e^{-\gamma}$ under a cardinality constraint and
$1-e^{-\alpha}$ for $\alpha$-weakly DR-submodular objectives. Finally, by
combining balanced exchange matrices with two specially designed online
oracles, we extend \MGPE\ to the online settings while preserving the
corresponding approximation factors of its offline counterpart.


\bibliography{references}
\fi
\clearpage
\appendix
\ifdefined\standaloneappendix
  \setcounter{theorem}{12}
  \setcounter{lemma}{10}
  \setcounter{assumption}{3}
  \setcounter{proposition}{2}
  \setcounter{definition}{4}
  \setcounter{remark}{3}
  \setcounter{equation}{32}
  \setcounter{figure}{1}
  \setcounter{table}{2}
  \setcounter{algocf}{4}
\fi
\section*{Detailed Appendix Contents}
\phantomsection
\label{app:guide}
\addcontentsline{toc}{section}{Detailed Appendix Contents}
\ifdefined\standaloneappendix
  \newcommand{\AppendixMainRefNumber}[1]{\getrefnumber{main-#1}}
\else
  \newcommand{\AppendixMainRefNumber}[1]{\getrefnumber{#1}}
\fi

\makeatletter
\newcommand{\AppendixGuideSection}[2]{%
  \par\addvspace{0.35\baselineskip}%
  \@dottedtocline{1}{0em}{0em}%
    {\hyperref[#1]{\textbf{Appendix~\getrefnumber{#1}: #2}}}%
    {\hyperref[#1]{\textbf{\getpagerefnumber{#1}}}}}
\newcommand{\AppendixGuideSubsection}[2]{%
  \@dottedtocline{2}{1.5em}{1.7em}%
    {\hyperref[#1]{\getrefnumber{#1}\quad #2}}%
    {\hyperref[#1]{\getpagerefnumber{#1}}}}
\makeatother

\begingroup
\small
\setcounter{tocdepth}{2}
\AppendixGuideSection{app:auxiliary-results}{Auxiliary Lemmas}
\AppendixGuideSubsection{app:local-to-global}{From Local Drift to Global Approximation Guarantees}

\AppendixGuideSection{app:conditioned-generator}{Proof of Lemma~\AppendixMainRefNumber{lem:conditioned-generator}}

\AppendixGuideSection{app:offline-weakly-submodular}{Proofs of Section~\AppendixMainRefNumber{sec:offline-weakly-submodular}}
\AppendixGuideSubsection{app:proof-repair}{Proof of Lemma~\AppendixMainRefNumber{lem:repair}}
\AppendixGuideSubsection{app:proof-gamma-drift}{Proof of Lemma~\AppendixMainRefNumber{lem:gamma-drift}}
\AppendixGuideSubsection{app:proof-offline-gamma}{Proof of Theorem~\AppendixMainRefNumber{thm:offline-gamma}}
\AppendixGuideSubsection{app:proof-offline-cardinality}{Proof of Theorem~\AppendixMainRefNumber{thm:offline-cardinality}}

\AppendixGuideSection{appendix:limitation_mgpe}{Limitations of \MGPE\ Algorithm}
\AppendixGuideSubsection{app:strict-barrier}{The Strict $(1-e^{-\gamma})$-Barrier}
\AppendixGuideSubsection{app:factor-revealing}{Factor-Revealing Linear Programming}

\AppendixGuideSection{app:beyond-barrier}{Proofs of Section~\AppendixMainRefNumber{sec:beyond-barrier}}
\AppendixGuideSubsection{app:proof-repair1}{Proof of Lemma~\AppendixMainRefNumber{lem:repair1}}
\AppendixGuideSubsection{app:proof-offline3}{Proof of Theorem~\AppendixMainRefNumber{thm:offline3}}

\AppendixGuideSection{app:fractional-exchange}{Further Details on Fractional Exchange}
\AppendixGuideSubsection{app:proof-balanced}{Proof of Theorem~\AppendixMainRefNumber{thm:balanced}}
\AppendixGuideSubsection{app:proof-fractional-exchange-properties}{Proof of Lemma~\AppendixMainRefNumber{thm:fractional-exchange-properties}}

\AppendixGuideSection{app:online-gamma}{Proofs of Section~\AppendixMainRefNumber{sec:online-gamma}}
\AppendixGuideSubsection{app:proof-conditioned-generator1}{Proof of Lemma~\AppendixMainRefNumber{lem:conditioned-generator1}}
\AppendixGuideSubsection{app:proof-online-gamma-drift}{Proof of Lemma~\AppendixMainRefNumber{add_important}}
\AppendixGuideSubsection{app:proof-ope-gamma-regret}{Proof of Theorem~\AppendixMainRefNumber{thm:ope-gamma-regret}}

\AppendixGuideSection{app:online-two-sided}{Proofs of Section~\AppendixMainRefNumber{sec:online-two-sided}}
\AppendixGuideSubsection{app:proof-online-two-sided-drift}{Proof of Lemma~\AppendixMainRefNumber{add_important1}}
\AppendixGuideSubsection{app:proof-ope-two-sided-regret}{Proof of Theorem~\AppendixMainRefNumber{thm:ope-two-sided-regret}}

\AppendixGuideSection{app:online-cardinality}{Proofs of Section~\AppendixMainRefNumber{sec:online-cardinality}}
\AppendixGuideSubsection{app:proof-online-cardinality-drift}{Proof of Lemma~\AppendixMainRefNumber{lem:32}}
\AppendixGuideSubsection{app:proof-ope-card-regret}{Proof of Theorem~\AppendixMainRefNumber{thm:ope-card-regret}}

\AppendixGuideSection{app:online-alpha}{Proofs of Section~\AppendixMainRefNumber{sec:online-alpha}}
\AppendixGuideSubsection{app:proof-online-alpha-drift}{Proof of Lemma~\AppendixMainRefNumber{lem:online-alpha-drift}}
\AppendixGuideSubsection{app:proof-ope-dr-regret}{Proof of Theorem~\AppendixMainRefNumber{thm:ope-DR-regret}}

\AppendixGuideSection{app:online-linear-oracles}{Online Linear Oracles}
\AppendixGuideSubsection{subsec:matroid-linear-oracle}{Stochastic Online Linear Optimization over a Matroid Polytope}
\AppendixGuideSubsection{subsec:one-dimensional-linear-oracle}{One-Dimensional Online Linear Oracle over $[1/2,1]$}
\endgroup
\clearpage

\section{Auxiliary Lemmas}
\label{app:auxiliary-results}
In this section, we present several auxiliary results used repeatedly in our 
analysis.
\subsection{From Local Drift to Global Approximation Guarantees}\label{app:local-to-global}
For a function \(q:[a,b]\to\mathbb R\), we define
\begin{equation}\label{eq:right-derivative-definition}
q'_+(\tau)
\triangleq
\lim_{h\downarrow0}
\frac{q(\tau+h)-q(\tau)}{h}
\end{equation}
when this limit exists. Next, we can show  
\begin{lemma}\label{lem:potential-comparison}
Let both \(V\in\mathbb R\) and \(\delta\ge0\) be constants, and let the function \(\eta:[a,b]\to\mathbb R_+\) be continuous. If the function \(q:[a,b]\to\mathbb R\) satisfies the following conditions:
\begin{enumerate}[label=(\roman*),leftmargin=18pt]
    \item \(q\) is continuous on \((a,b)\);
    \item \(q\) is right-continuous at \(a\) and left-continuous at \(b\);
    \item \(q'_+(\tau)\) exists for every \(\tau\in[a,b)\);
    \item for any $\tau\in[a,b)$, the following inequality holds, namely,
    \begin{equation}\label{eq:potential-local-drift}
q'_+(\tau)
\ge
\eta(\tau)\bigl(V-q(\tau)\bigr)-\delta
\end{equation}
\end{enumerate}, then  we can have, for every \(\tau\in[a,b]\),
\begin{align}
q(\tau)\ge
e^{-H(\tau)}q(a)
+\bigl(1-e^{-H(\tau)}\bigr)V-\delta\int_a^\tau
\exp\left(-\int_s^t\eta(u)\,du\right)ds,
\end{align} where  $H(\tau)\triangleq\int_a^\tau\eta(u)\,du$.
\end{lemma}

\begin{proof}
We first construct a function \(z:[a,b]\to\mathbb R\) and suppose that it is
be the solution of 
\begin{equation}\label{eq:comparison-trajectory-ode}
z'(\tau)
=
\eta(\tau)\bigl(V-z(\tau)\bigr)-\delta,
\qquad
z(a)=q(a).
\end{equation}Consequently, we can show that 
\begin{align}
z(\tau)=e^{-H(\tau)}q(a)
+\bigl(1-e^{-H(\tau)}\bigr)V-\delta\int_a^\tau
\exp\left(-\int_s^\tau\eta(u)\,du\right)ds.
\label{eq:comparison-trajectory-explicit}
\end{align}
Next, we prove that \(q(\tau)\ge z(\tau)\) for every \(\tau\in[a,b]\).

We first consider the case in which
Eq.\eqref{eq:potential-local-drift} holds strictly, namely,
\begin{equation}\label{eq:potential-strict-local-drift}
q'_+(\tau)
>
\eta(\tau)\bigl(V-q(\tau)\bigr)-\delta
\qquad\text{for every }\tau\in[a,b).
\end{equation}
Fix an arbitrary \(\bar \tau\in(a,b]\) and define
\[
\mathcal X_{\bar \tau}
\triangleq
\{\tau\in[a,\bar \tau]:q(\tau)\ge z(\tau)\}.
\]
Since \(q(a)=z(a)\), the set \(\mathcal X_{\bar t}\) is nonempty. Let $s\triangleq\sup\mathcal X_{\bar \tau}$. If \(s=\bar \tau\), there exists a sequence
\(\{\tau_m\}_{m\ge1}\subseteq\mathcal X_{\bar \tau}\) converging to
\(\bar \tau\) from the left. The continuity assumptions on \(q\), together
with the continuity of \(z\), then imply $q(\bar \tau)
=
\lim_{m\to\infty}q(\tau_m)
\ge
\lim_{m\to\infty}z(\tau_m)
=
z(\bar \tau)$ such that 
\[
q(\tau)\ge z(\tau),\quad\text{for every}\,\tau\in[a,\bar \tau].
\]

As for the scenario that \(s<\bar \tau\), we also infer that \(q(s)= z(s)\). Like the previous analysis, by continuity, we can show \(q(s)\ge z(s)\). If \(q(s)>z(s)\), and right continuity would give
\(q(t)>z(t)\) for some \(t>s\), contradicting the definition of \(s\).
Thus, $q(s)=z(s)$, when \(s<\bar \tau\). Furthermore, because \(s=\sup\mathcal X_{\bar t}\), we have $q(s+h)<z(s+h)$ for every $h\in(0,\bar t-s]$. Then, we can show  that
\begin{align}
q'_+(s)
&=
\lim_{h\downarrow0}
\frac{q(s+h)-q(s)}{h}
\nonumber\\
&\le
\lim_{h\downarrow0}
\frac{z(s+h)-z(s)}{h}
=
z'(s).
\label{eq:first-contact-comparison}
\end{align}
On the other hand, by
Eq.\eqref{eq:potential-strict-local-drift},
\(q(s)=z(s)\), and Eq.\eqref{eq:comparison-trajectory-ode},
\[
q'_+(s)
>
\eta(s)\bigl(V-q(s)\bigr)-\delta
=
\eta(s)\bigl(V-z(s)\bigr)-\delta
=
z'(s),
\]
which contradicts Eq.\eqref{eq:first-contact-comparison}. Hence the case
\(s<\bar t\) is impossible, and therefore, for any fixed \(\bar \tau\in(a,b]\), we conclude
that
\[
q(\tau)\ge z(\tau),\quad\text{for every}\, \tau\in[a,\bar \tau].
\]

We now remove the strictness assumption of Eq.\eqref{eq:potential-strict-local-drift}. For every \(\xi>0\), define
\[
q_\xi(\tau)\triangleq q(\tau)+\xi(\tau-a).
\]
The function \(q_\xi\) satisfies the same continuity assumptions as \(q\),
and \(q_\xi(a)=q(a)=z(a)\). Moreover, for every \(\tau\in[a,b)\),
\begin{align}
(q_\xi)'_+(\tau)
&=q'_+(\tau)+\xi\ge
\eta(\tau)\bigl(V-q(\tau)\bigr)-\delta+\xi\\
&=
\eta(\tau)\bigl(V-q_\xi(\tau)\bigr)-\delta
+\xi\bigl(1+\eta(t)(\tau-a)\bigr)\\&>
\eta(t)\bigl(V-q_\xi(\tau)\bigr)-\delta,
\end{align}
where the final inequality follows from \(\eta(t)\ge0\). 

Therefore, we can show  $q_\xi(\tau)\ge z(\tau)$ for every  $\tau\in[a,b]$. Letting \(\xi\downarrow0\) yields \(q(\tau)\ge z(\tau)\). Combining this
comparison with Eq.\eqref{eq:comparison-trajectory-explicit} proves Lemma~\ref{lem:potential-comparison}.
\end{proof}
\begin{lemma}\label{lem:finite-start-correction}
Let \(0<a\le\theta\le1\) and \(\varepsilon\in[0,1]\). Define
\begin{equation}\label{eq:finite-start-function}
\Psi_{\theta,a}(s)
\triangleq\left(\frac{a+2(1-a)s}{2-a}
\right)^{\frac{\theta a}{2(1-a)}}\quad\forall s\in[0,1].
\end{equation}Then
\begin{equation}\label{eq:finite-start-general-bound}
\Psi_{\theta,a}(\varepsilon)
\le\Psi_{\theta,a}(0)+\varepsilon.
\end{equation}
\end{lemma}

\begin{proof}For convenience, we define
\[
p_{\theta,a}
\triangleq
\frac{\theta a}{2(1-a)}
\qquad\text{and}\qquad
x_a(s)
\triangleq
\frac{a+2(1-a)s}{2-a}.
\]
Thus, we can rewrite $\Psi_{\theta,a}(s)=x_a(s)^{p_{\theta,a}}$. Differentiating \(\Psi_{\theta,a}\) gives
\begin{align}
\Psi_{\theta,a}'(s)=
p_{\theta,a}\cdot x_a(s)^{p_{\theta,a}-1}
\cdot \frac{2(1-a)}{2-a}=
\frac{\theta a}{2-a}
x_a(s)^{p_{\theta,a}-1}.
\label{eq:finite-start-derivative}
\end{align}Moreover, for every \(s\in[0,1]\), we know that  $\frac{a}{2-a}
\le
x_a(s)
\le1$.

Next, we prove that $\Psi_{\theta,a}'(s)\le1$ for any $s\in[0,1]$ by distinguishing two cases.

Firstly, we suppose that \(p_{\theta,a}\ge1\). Since \(x_a(s)\le1\),
we have \(x_a(s)^{p_{\theta,a}-1}\le1\). Hence,
\[
\Psi_{\theta,a}'(s)
\le
\frac{\theta a}{2-a}
\le1,
\]
where the final inequality follows from
\(0<\theta\le1\) and \(0<a\le1\).

As for \(p_{\theta,a}<1\), due to \(x_a(s)\ge a/(2-a)\), we ca  obtain
\begin{align*}
\Psi_{\theta,a}'(s)
&=\frac{\theta a}{2-a}
x_a(s)^{p_{\theta,a}-1}\le
\frac{\theta a}{2-a}
\left(\frac{a}{2-a}\right)^{p_{\theta,a}-1}\\
&=
\theta
\left(\frac{a}{2-a}\right)^{p_{\theta,a}}
\le1.
\end{align*} Therefore, we have $\Psi_{\theta,a}'(s)\le1$ for any $s\in[0,1]$. Integrating this derivative bound over \([0,\varepsilon]\) gives
\[
\Psi_{\theta,a}(\varepsilon)-\Psi_{\theta,a}(0)
=
\int_0^\varepsilon
\Psi_{\theta,a}'(s)\,ds
\le
\varepsilon.
\]
\end{proof}

\section{Proof of Lemma~\ref{lem:conditioned-generator}}\label{app:conditioned-generator}

In this section, we focus on verifying the core Lemma~\ref{lem:conditioned-generator} about the potential value function $Q(r)$ and its conditioned value process $V_{S,\tau}(r)$. Before going into the details, we first present a related result about the conditioned Transition probabilities $\Pr[A(r)=T\mid A(t)=S]$, that is to say, 
\begin{lemma}[Derivatives of conditioned transition probabilities]\label{lem:transition-derivative}
For any $\tau\in[\varepsilon,1)$ and any $S,T\in\cB$, we have 
\begin{equation}\label{eq:transition-derivative}
 \left.\frac{\partial_+\Pr[A(r)=T\mid A(\tau)=S]}{\partial r}\right|_{r=\tau}
 =
 \begin{cases}
 \lambda(\tau)p_{i,j}(\tau,S), & T=S-i+j,\ j\notin S,\\
 -\lambda(\tau)\left(1-\sum_{i\in S}p_{i,i}(\tau,S)\right), & T=S,\\
 0, & \text{otherwise},
 \end{cases}
\end{equation}where \(p_{i,j}(\tau,S)\) denotes the probability that we delete
\(i\) and insert \(j\), conditional on a Poisson event occurring at time
\(\tau\) while the current base is \(S\).
\end{lemma}
\begin{remark} Note that the proof of Lemma~\ref{lem:transition-derivative} follows the same argument
as that of Lemma~3.5 in \citet{rozenman2026poisson} and is therefore omitted.
\end{remark}
With the the help of Lemma~\ref{lem:transition-derivative}, we next can prove that 
\begin{lemma}[Conditioned Generator]
\label{lem:core_lemma+}
For every $\tau\in[\varepsilon,1)$ and $S\in\cB$,
\begin{align}
 \left.\frac{\partial^+V_{S,\tau}(r)}{\partial r}\right|_{r=\tau}
 &=\sum_{i\in S}\nabla_iF(\tau\1_S)
  +\lambda(t)\E_{(i,j)\sim p(\tau,S)}
    \left[F(\tau\1_{S-i+j})-F(\tau\1_S)\right]\label{eq:generator-general}\\
 &=\sum_{i\in S}\nabla_iF(\tau\1_S)
  +\tau\lambda(\tau)\E_{(i,j)\sim p(\tau,S)}
  \left[\nabla_jF(\tau\1_{S-i})-\nabla_iF(\tau\1_{S-i})\right].\notag
\end{align}
\end{lemma}
\begin{proof}
According to the definition of $V_{S,\tau}(r)$, we can show that
\begin{equation*}
    V_{S,\tau}(r)\triangleq\sum_{T\in\cB}F(r\1_{T})\cdot \Pr[A(r)=T\mid A(\tau)=S].
\end{equation*}
Then, we have that 
\begin{equation*}
    \begin{aligned}
    &\left.\frac{\partial^+V_{S,\tau}(r)}{\partial r}\right|_{r=\tau}\\&=\sum_{T\in\cB}\left.\frac{\partial^+F(r\1_{T})}{\partial r}\right|_{r=\tau}\cdot \Pr[A(\tau)=T\mid A(\tau)=S]+ \sum_{T\in\cB}\left.F(\tau\1_{T})\frac{\partial^+\Pr[A(r)=T\mid A(\tau)=S]}{\partial r}\right|_{r=\tau}\\
    &=\sum_{i\in S}\nabla_{j}F(\tau\1_{S})-\sum_{i\in S}\lambda(\tau)\left(1-\sum_{i\in S}p_{i,i}(\tau,S)\right)\cdot F(\tau \1_{S})+\sum_{i\in S}\sum_{j:S-i+j\in\cB,j\notin S}\lambda(\tau)p_{i,j}(\tau,S)\cdot F(\tau \1_{S-i+j}),\\
    &=\sum_{i\in S}\nabla_{j}F(\tau\1_{S})-\sum_{i\in S}\lambda(\tau)\left(1-\sum_{i\in S}p_{i,i}(\tau,S)\right)\cdot F(\tau \1_{S})+\sum_{i\in S}\sum_{j\in U\setminus S}\lambda(\tau)p_{i,j}(\tau,S)\cdot F(\tau \1_{S-i+j})\\
    &=\sum_{i\in S}\nabla_{j}F(\tau\1_{S})+\lambda(\tau)\sum_{i\in S}\sum_{j\in U}p_{i,j}(\tau,S)\cdot \big(F(\tau \1_{S-i+j})-F(\tau\1_{S})\big)\\&=\sum_{i\in S}\nabla_iF(\tau\1_S)
  +\lambda(t)\E_{(i,j)\sim p(\tau,S)}
    \left[F(\tau\1_{S-i+j})-F(\tau\1_S)\right],
    \end{aligned}
\end{equation*} where the second equality follows from  Lemma~\ref{lem:transition-derivative}; the third equality comes from $p_{i,j}(\tau,S)\equiv0$ if $S-i+j\notin\cB$ and the final  equality comes from $\sum_{i\in S}\sum_{j\in U}p_{i,j}=1$. 
\end{proof}

Note that our proposed \MGPE\ , we only select the maximum-gain exchange such that we have 
\begin{equation*}
 p_{i,j}(\tau,S)\triangleq 
 \begin{cases}
 1/k, & i\in S, j\in\max_{z\in C_S(i)}\nabla_{z}F(\tau\1_{S-i}),\\
 0, & \text{otherwise}.
 \end{cases}   
\end{equation*}

As a result, by applying this lemma~\ref{lem:core_lemma+} to
our proposed \MGPE, we immediately get 
\begin{equation*}
    \begin{aligned}
     \left.\frac{\partial^+V_{S,\tau}(r)}{\partial r}\right|_{r=\tau}&=\sum_{i\in S}\nabla_iF(\tau\1_S)
  +\tau\lambda(\tau)\E_{(i,j)\sim p^{\text{\MGPE}}(\tau,S)}
  \left[\nabla_jF(\tau\1_{S-i})-\nabla_iF(\tau\1_{S-i})\right]\\
&=\sum_{i\in S}\nabla_iF(\tau\1_S)
  +\frac{\tau\lambda(\tau)}{k}\sum_{i\in S}
  \left(\nabla_{j^{*}(S,\tau)}F(\tau\1_{S-i})-\nabla_iF(\tau\1_{S-i})\right)=\sum_{i\in S}\nabla_{j^{*}(S,\tau)}F(\tau\1_{S-i}),
    \end{aligned}
\end{equation*} where the final equality comes from $\lambda(\tau)=k/\tau$.

\section{Proofs of Section~\ref{sec:offline-weakly-submodular}}\label{app:offline-weakly-submodular}

\subsection{Proof of Lemma~\ref{lem:repair}}\label{app:proof-repair}
\begin{proof}
If $j=i$, we can show that 
\begin{equation*}
    \E _{R\sim \tau\1_{S}}\left[f(j\mid R)\right]=\E _{R\sim \tau\1_{S}}\left[f(i\mid R)\right]=(1-\tau)\E _{R\sim \tau\1_{S-i}}\left[f(i\mid R)\right]=(1-\tau)\nabla_iF(\tau\1_{S-i}).
\end{equation*}where the second equation comes from $f(i\mid R)=0$ when $i\in R$. Hence, if $j=i$, we have that
\begin{equation*}
   \begin{aligned}
\frac{\gamma}{C_\gamma(\tau)}\E _{R\sim \tau\1_{S}}\left[f(j\mid R)\right]
 -\frac{(1-\gamma)\tau}{C_\gamma(\tau)}\nabla_iF(\tau\1_{S-i})\le\frac{\gamma}{C_\gamma(\tau)}\E _{R\sim \tau\1_{S}}\left[f(j\mid R)\right]=\frac{\gamma(1-\tau)}{C_\gamma(\tau)}\nabla_jF(\tau\1_{S-i}).
   \end{aligned} 
\end{equation*}Due to that $\gamma(1-\tau)\le C_\gamma(\tau)=\gamma+(1-\gamma)\tau$, we thus have Eq.\eqref{eq:repair}.

Now suppose $j\notin S$ and set $R_0\sim \tau\1_{S-i}$. At first, we have
\begin{equation}\label{equ:add1}
    \E _{R\sim \tau\1_{S}}\left[f(j\mid R)\right]=(1-\tau)\E _{R_0}\left[f(j\mid R_0)\right]+\tau\E _{R_0}\left[f(j\mid R_0+i)\right].
\end{equation}Then, applying weak submodularity with
$A=R_0$ and $B=R_0+i+j$ gives
\begin{equation*}\label{eq:repair-two-point}
f(i\mid R_0)+f(j\mid R_0)\ge\gamma\left(f(R_0+i+j)-f(R_0)\right)=\gamma\left(f(j\mid R_0+i)+f(i\mid R_0)\right).
\end{equation*} 
As a result, we get the following inequality:
\begin{equation}\label{equ:add2}
\begin{aligned}
\nabla_{j} F(t\1_{S-i})=\E_{R_0}[f(j\mid R_0)]\ge\gamma\E_{R_0}[f(j\mid R_0+i)]-(1-\gamma)\E_{R_0}[f(i\mid R_0)]
\end{aligned}
\end{equation}
Multiplying the Eq.\eqref{equ:add2} by $t$, we then have that 
\begin{equation}\label{equ:add3}
    \begin{aligned}
        \tau\nabla_{j} F(\tau\1_{S-i})&=\tau\E_{R_0}[f(j\mid R_0)]\ge\gamma(\tau\E_{R_0}[f(j\mid R_0+i)])-(1-\gamma)\tau\E_{R_0}[f(i\mid R_0)]\\
        &=\gamma(E _{R\sim \tau\1_{S}}\left[f(j\mid R)\right]-(1-\tau)\E _{R_0}\left[f(j\mid R_0)\right])-(1-\gamma)\tau\E_{R_0}[f(i\mid R_0)]\\
        &=\gamma E _{R\sim \tau\1_{S}}\left[f(j\mid R)\right]-(1-\tau)\gamma\nabla_{j} F(\tau\1_{S-i})-(1-\gamma)\tau\E_{R_0}[f(i\mid R_0)],
    \end{aligned}
\end{equation} where the second equation comes from  Eq.\eqref{equ:add1}. Therefore, adding $(1-\tau)\gamma\nabla_{j} F(\tau\1_{S-i})$ to both sides of Eq.\eqref{equ:add3} and then dividing by  $C_{\gamma}(\tau)$, we obtain Eq.\eqref{eq:repair}.
\end{proof}

\subsection{Proof of Lemma~\ref{lem:gamma-drift}}\label{app:proof-gamma-drift}
\begin{proof}
Let $O$ be the optimal solution and $h:S\to O$ be a Brualdi map. Since $h(i)\in\mathcal C_S(i)$, Lemma~\ref{lem:repair} gives that
\begin{equation}\label{add:4}
    \begin{aligned}
        &\sum_{i\in S}\nabla_{j_i^{*}(S,\tau)}F(\tau\1_{S-i})\ge\sum_{i\in S}\nabla_{h(i)}F(\tau\1_{S-i})\\&\ge\frac{\gamma}{C_\gamma(\tau)}\E _{R\sim \tau\1_{S}}\left[\sum_{i\in S}f(h(i)\mid R)\right]
 -\frac{(1-\gamma)\tau}{C_\gamma(\tau)}
   \sum_{i\in S}\nabla_iF(\tau\1_{S-i})\\
   &\ge\frac{\gamma^{2}}{C_\gamma(\tau)}\big(f(O)-\E _{R\sim \tau\1_{S}}\left[f(R)\right]\big)
 -\frac{(1-\gamma)\tau}{C_\gamma(\tau)}
   \sum_{i\in S}\nabla_iF(\tau\1_{S-i}),
    \end{aligned}
\end{equation} where the second inequality comes from the $\gamma$-weak submodularity, that is, $\sum_{i\in S}f(h(i)\mid R)\ge\gamma(f(O)-f(R))$. Furthermore, due to $i\in\mathcal C_S(i)$, we have 
\begin{equation}\label{eq:self-absorption}
\sum_{i\in S}\nabla_{j_i^{*}(S,\tau)}F(\tau\1_{S-i})\ge \sum_{i\in S}\nabla_iF(\tau\1_{S-i}).
\end{equation}
Replacing $\sum_{i\in S}\nabla_iF(\tau\1_{S-i})$ in Eq.\eqref{add:4} by $\sum_{i\in S}\nabla_{j_i^{*}(S,\tau)}F(\tau\1_{S-i})$ and moving the resulting term to the left yields
\begin{equation*}
 \sum_{i\in S}\nabla_{j_i^{*}(S,\tau)}F(\tau\1_{S-i})
 \ge\frac{\gamma^2}{\gamma+2(1-\gamma)\tau}\bigl(f(O)-F(\tau\1_S)\bigr).
\end{equation*}
Merging this inequality into the result of Lemma~\ref{lem:conditioned-generator}, we get the Eq.\eqref{eq:gamma-drift1} and Eq.\eqref{eq:gamma-drift2}.
\end{proof}

\subsection{Proof of Theorem~\ref{thm:offline-gamma}}\label{app:proof-offline-gamma}
\begin{proof} According to Eq.\eqref{eq:gamma-drift2} of Lemma~\ref{lem:gamma-drift}, we can show that  \begin{equation*}
\left.\frac{\partial^+Q(r)}{\partial r}\right|_{r=\tau}\ge\frac{\gamma^2}{\gamma+2(1-\gamma)\tau}
 \bigl(f(O)-Q(\tau)\bigr).
 \end{equation*}Then, by substituting $\eta(\tau)=\frac{\gamma^2}{\gamma+2(1-\gamma)\tau}$ and $\delta=0$ into Lemma~\ref{lem:potential-comparison}, we get 
\begin{equation*}
   Q(1)\ge e^{-H(1)}Q(\epsilon)+(1-e^{-H(1)})f(O). 
\end{equation*}Here,
\begin{equation*}
\left.H(1)\triangleq\int_\epsilon^1\eta(u)\,du=\frac{\gamma^{2}}{2(1-\gamma)}\log(\gamma+2(1-\gamma)\tau)\right|_{\tau=\epsilon}^{1}=\frac{\gamma^{2}}{2(1-\gamma)}\log(\frac{2-\gamma}{\gamma+2(1-\gamma)\epsilon}).
\end{equation*}
Then, we have 
\begin{equation*}
    \begin{aligned}
         Q(1)&\ge e^{-H(1)}Q(\epsilon)+(1-e^{-H(1)})f(O)\\
         &\ge(1-e^{-H(1)})f(O)=\left(1-(\frac{\gamma+2(1-\gamma)\epsilon}{2-\gamma})^{\frac{\gamma^{2}}{2(1-\gamma)}}\right)f(O)\ge(1-\rho_{\gamma}-\epsilon)f(O),
    \end{aligned}
\end{equation*} where the second inequality comes from $Q(\epsilon)\ge0$ and the final inequality comes from Lemma~\ref{lem:finite-start-correction} as well as $\rho_{\gamma}\triangleq(\frac{\gamma}{2-\gamma})^{\frac{\gamma^{2}}{2(1-\gamma)}}$.
\end{proof}

\subsection{Proof of Theorem~\ref{thm:offline-cardinality}}\label{app:proof-offline-cardinality}
\begin{proof}
Fix a base $S$, an optimal base $O\in\argmax_{|S|\le k}f(S)$, and $t\in(\varepsilon,1]$, we define 
\begin{equation*}
 X=O\setminus S,\qquad Y=O\cap S,\qquad Z=S\setminus O,
 \qquad \ell=|X|=|Z|,
\end{equation*}
and abbreviate the maximum-gain sum by
\begin{equation*}
 G_S(\tau)=\sum_{i\in S}\max_{j:\,S-i+j\in\cB}
             \nabla_jF(\tau\1_{S-i}),
\end{equation*} where $\cB=\{S\subseteq U: |S|=k\}$.

At first, we suppose that $\ell>0$.  Moreover, from the definition of maximum gain, we can show that
\begin{equation}\label{eq:cardinality-row-average}
\begin{aligned}
   G_S(\tau)&=\sum_{i\in S}\max_{j:\,S-i+j\in\cB}
             \nabla_jF(\tau\1_{S-i})\\&=\sum_{i\in Y}\max_{j:\,S-i+j\in\cB}
             \nabla_jF(\tau\1_{S-i}) +\sum_{i\in Z}\max_{j:\,S-i+j\in\cB}
             \nabla_jF(\tau\1_{S-i})\\&\ge L_{S}(\tau)\triangleq
 \sum_{y\in Y}\nabla_yF(\tau\1_{S-y})
 +\frac1\ell\sum_{z\in Z}\sum_{o\in X}
       \nabla_oF(\tau\1_{S-z}),
\end{aligned}
\end{equation} where the final inequality comes from that $S=S-y+y\in\cB$ when $y\in Y$ and $S-z+x\in\cB$ when $z\in Z$ and $o\in X$. 

Secondly, for every $A\subsetneq S$,  we define
\begin{equation}\label{eq:cardinality-weights}
 w_A=\frac{|Z\setminus A|}{\ell}
      \tau^{|A|}(1-\tau)^{k-1-|A|},
 \qquad w_S=0.
\end{equation}
Then, from the definition of  multilinear extension, we can show that, for any $o\in X$, the following identity holds:
\begin{equation*}
    \frac1\ell\sum_{z\in Z}
       \nabla_oF(\tau\1_{S-z})=\sum_{A\subseteq S} w_A\cdot f(o\mid A)
\end{equation*} Similarly, we also can show that
\begin{equation*}
    \sum_{y\in Y}\nabla_yF(\tau\1_{S-y})=\sum_{A\subseteq S}\sum_{Y\setminus A}\left(\tau^{|A|}(1-\tau)^{k-1-|A|}\right)f(y\mid A).
\end{equation*}
As a result, we have that
\begin{equation}\label{eq:cardinality-1}
    \begin{aligned}
L_S(\tau)&= \sum_{y\in Y}\nabla_yF(\tau\1_{S-y})
 +\frac1\ell\sum_{z\in Z}\sum_{o\in X}
       \nabla_oF(\tau\1_{S-z})\\
&=\sum_{A\subseteq S}\sum_{Y\setminus A}\left(\tau^{|A|}(1-\tau)^{k-1-|A|}\right)f(y\mid A)+\sum_{A\subseteq S} w_A\sum_{o\in X}f(o\mid A)\\
 &\ge\sum_{A\subseteq S}w_A
       \sum_{o\in O\setminus A}f(o\mid A)\ge\gamma\sum_{A\subseteq S}w_A
       \bigl(f(A\cup O)-f(A)\bigr)\\&\ge\gamma\left(f(O)-\sum_{A\subseteq S}w_Af(A)\right),
    \end{aligned}
\end{equation} where the first inequality comes from that $\tau^{|A|}(1-\tau)^{k-1-|A|}\ge w_A$ and the second inequality follows from weak submodularity. Moreover, we also can show that
\begin{equation}\label{eq:cardinality-weighted-value}
 \sum_{A\subseteq S}w_Af(A)
 =\frac1\ell\sum_{z\in Z}F(\tau\1_{S-z})
 \le F(\tau\1_S).
\end{equation}
Substitution this Eq.\eqref{eq:cardinality-weighted-value} into
Eq.\eqref{eq:cardinality-1} yields the following result, when $\ell>0$,
\begin{equation}\label{eq:cardinality-statewise-drift}
 G_S(\tau)\ge\gamma\bigl(f(O)-F(\tau\1_S)\bigr).
\end{equation}

If $\ell=0$, then $S=O$.  For $R\sim t\1_S$, weak submodularity applied
from $R$ to $S$ gives
\begin{equation*}
\begin{aligned}
  &G_S(\tau)=\sum_{i\in S}\max_{j:\,S-i+j\in\cB}
             \nabla_jF(\tau\1_{S-i})\ge \sum_{i\in S}
             \nabla_iF(\tau\1_{S-i})\\&=\sum_{A\subsetneq S}\sum_{S\setminus A}\left(\tau^{|A|}(1-\tau)^{k-1-|A|}\right)f(y\mid A)=\sum_{A\subsetneq S}\left(\tau^{|A|}(1-\tau)^{k-1-|A|}\right)\sum_{S\setminus A}f(y\mid A)\\
             &\ge \sum_{A\subsetneq S}\left(\tau^{|A|}(1-\tau)^{k-1-|A|}\right)(f(S)-f(A))=\frac{\gamma}{1-\tau}\sum_{A\subseteq S}\left(\tau^{|A|}(1-\tau)^{k-|A|}\right)(f(S)-f(A))\\&\ge\gamma(f(S)-F(t\1_{S}))=\gamma(f(O)-F(t\1_{S})).  
\end{aligned}
\end{equation*}
Thus, according to Lemma~\ref{lem:conditioned-generator}, we now can infer that
\begin{equation}\label{eq:cardinality-Q-ode}
 Q'_+(\tau)\ge\gamma\bigl(f(O)-Q(\tau)\bigr).
\end{equation}
Solving Eq.\eqref{eq:cardinality-Q-ode} via Lemma~\ref{lem:potential-comparison} and using
$Q(\varepsilon)\ge0$ proves Eq.\eqref{eq:cardinality-finite-start}.
\end{proof}

\section{Limitations of \MGPE\ Algorithm}\label{appendix:limitation_mgpe}
\subsection{The Strict \texorpdfstring{$(1-e^{-\gamma})$}{(1-exp(-gamma))}-Barrier}\label{app:strict-barrier}

In this subsection, we consider a ground set $U$ consisting of two bad elements \(\{b_1,b_2\}\) and two good elements \(\{g_1,g_2\}\), and simultaneously set the partition matroid $\mathcal{M}$ with blocks
$\{b_1,g_1\}$ and $\{b_2,g_2\}$, each of capacity one.  As for the set objective function $f$, we suppose its value only depends on the numbers of blocks in the bad-only, good-only, and double-occupied states. More specifically, for any subset $T\subseteq U$, let the tuple $(p,q,d)$ count the bad-only, good-only, and
double-occupied blocks of $T$. Then, $f(T)=v_{p,q,d}$. 

Furthermore, we define \begin{equation}\label{eq:barrier-parameters}
 c_\gamma=\frac\gamma2,\qquad
 s_\gamma=\frac{\gamma(2+3\gamma-\gamma^2)}{2(\gamma+2)},\qquad
 r_\gamma=\frac{\gamma(6-\gamma)}{2(\gamma+2)},
\end{equation}
and set the count function $v_{p,q,d}$ as follows: \begin{center}
\small
\begin{tabular}{@{}lc@{}}
\toprule
State or states $(p,q,d)$ & $v_{p,q,d}$\\
\midrule
$(0,0,0)$ & $0$\\
$(1,0,0),(0,1,0)$ & $c_\gamma$\\
$(2,0,0),(1,1,0)$ & $s_\gamma$\\
$(0,0,1),(1,0,1)$ & $r_\gamma$\\
$(0,2,0),(0,1,1),(0,0,2)$ & $1$\\
\bottomrule
\end{tabular}
\end{center}

Then, when $\gamma\in(0,1]$, we can show 
\begin{equation}\label{eq:barrier-monotone-gaps}
 s_\gamma-c_\gamma
 =\frac{\gamma^2(2-\gamma)}{2(\gamma+2)}\ge0,\quad
 r_\gamma-s_\gamma
 =\frac{\gamma(2-\gamma)^2}{2(\gamma+2)}\ge0,\quad
 1-r_\gamma=\frac{(2-\gamma)^2}{2(\gamma+2)}\ge0,
\end{equation}so $f_\gamma$ is normalized and monotone. Next, we verify that $f$ is $\gamma$-weakly submodular.
For every pair of sets $A\subsetneq C\subseteq U$, we define $\sigma_\gamma(A,C)
\coloneqq
\sum_{e\in C\setminus A} f_\gamma(e\mid A)
-\gamma\bigl(f_\gamma(C)-f_\gamma(A)\bigr)$. Thus, it suffices to show that $\sigma_\gamma(A,C)\geq 0$ for every
$A\subsetneq C\subseteq U$. Since the ground set contains only four
elements, we can enumerate all feasible pairs $(A,C)$. Table~\ref{tab:weak-submodularity-slacks} lists all possible values of
$\sigma_\gamma(A,C)$ together with their counts. It is worth noting that, for every $\gamma\in(0,1]$, all the expressions reported in
Table~\ref{tab:weak-submodularity-slacks} are nonnegative, which proves
that $f$ is $\gamma$-weakly submodular.

\begin{table}[t]
\centering
\footnotesize
\begingroup
\newcommand{\slackcell}[2]{%
    \shortstack{
        $\displaystyle #1$\\[-1pt]
        {\scriptsize $(\times\,#2)$}
    }%
}
\setlength{\tabcolsep}{2pt}
\renewcommand{\arraystretch}{1.25}
\caption{Possible values of $\sigma_\gamma(A,C)$. The notation
$(\times m)$ indicates that the corresponding value occurs $m$ times and $D_\gamma\triangleq 2(\gamma+2)$.}
\label{tab:weak-submodularity-slacks}
\begin{tabular*}{\linewidth}{
    @{\extracolsep{\fill}}
    ccccc
    @{}
}
\toprule
\slackcell{\frac{(2-\gamma)(1-\gamma)}{2}}{2}
&
\slackcell{\frac{(2-\gamma)^2(1+\gamma)}{D_\gamma}}{2}
&
\slackcell{\frac{(2-\gamma)^2}{D_\gamma}}{2}
&
\slackcell{\frac{(2-\gamma)^2(1-\gamma)(1+\gamma)}
{D_\gamma}}{2}
&
\slackcell{\frac{(2-\gamma)^2(1-\gamma)}
{D_\gamma}}{6}
\\[9pt]

\slackcell{\frac{(2-\gamma)^2(1+\gamma-\gamma^2)}
{D_\gamma}}{2}
&
\slackcell{\frac{2-\gamma}{2}}{2}
&
\slackcell{\frac{\gamma(1-\gamma)}{2}}{4}
&
\slackcell{\frac{\gamma(2-\gamma)^2(1-\gamma)}
{D_\gamma}}{5}
&
\slackcell{\frac{\gamma^2(2-\gamma)}
{D_\gamma}}{2}
\\[9pt]

\slackcell{0}{12}
&
\slackcell{\gamma}{1}
&
\slackcell{\frac{2\gamma(2-\gamma)(1-\gamma)}
{D_\gamma}}{4}
&
\slackcell{\frac{\gamma(2-\gamma)^2(1+\gamma)}
{D_\gamma}}{3}
&
\slackcell{\frac{\gamma(2-\gamma)^2}
{D_\gamma}}{6}
\\[9pt]

&
\slackcell{\frac{\gamma(\gamma^2-3\gamma+6)}
{D_\gamma}}{2}
&
\slackcell{\frac{\gamma^2(2-\gamma)(1-\gamma)}
{D_\gamma}}{6}
&
\slackcell{\frac{\gamma}{2}}{2}
&
\\
\bottomrule
\end{tabular*}
\endgroup
\end{table}

With the help of this $\gamma$-weakly submodular function $f$, we next prove the results of Proposition~\ref{prop:strict-barrier}. Generally speaking, adding the nonnegative modular function
will preserves the weak submodularity. Thus, for
$0<\delta<(1-s_\gamma)/2$, we define
\begin{equation*}
 g_{\gamma,\delta}(T)=f_\gamma(T)+\delta|T\cap\{b_1,b_2\}|,
\end{equation*} and then reset  $ f(T)\triangleq\min\{g_{\gamma,\delta}(T),1\}$. It is worth noting that  Truncation at one also can preserve the weak submodularity of $f$: if
$g_{\gamma,\delta}(A)<1$, we define $T_0=1-g_{\gamma,\delta}(A)$,
$D=g_{\gamma,\delta}(C)-g_{\gamma,\delta}(A)$, and
$d_e=g_{\gamma,\delta}(e\mid A)$.  Then
\begin{align*}
 \sum_{e\in C\setminus A}f(e\mid A)=\sum_{e\in C\setminus A}\min\{d_e,T_0\}.
\end{align*}
When all $\sum_{e\in C\setminus A}d_e\le T_0$, we have $d_e\le T_0$ such that $\sum_{e\in C\setminus A}f(e\mid A)=\sum_{e\in C\setminus A}d_e$. As for the setting $\sum_{e\in C\setminus A}d_e>T_0$, if there exists a $d_e\ge T_0$, we have 
\begin{equation*}
  \sum_{e\in C\setminus A}f(e\mid A)\ge T_0\ge \min\!\left\{\sum_{e\in C\setminus A} d_e,T_0\right\}.
\end{equation*} Otherwise, all $d_e\le T_0$, we go back to $\sum_{e\in C\setminus A}f(e\mid A)=\sum_{e\in C\setminus A}d_e$. As a result, we have 
\begin{equation*}
\begin{aligned}
 &\sum_{e\in C\setminus A}f(e\mid A)=\sum_{e\in C\setminus A}\min\{d_e,T_0\}\ge\min\!\left\{\sum_e d_e,T_0\right\}\\
 &\ge\min\{\gamma D,T_0\}
 \ge\gamma\min\{D,T_0\}
 =\gamma(f(C)-f(A)).
\end{aligned}
\end{equation*}In the case $g_{\gamma,\delta}(A)\ge1$, we have $\sum_{e\in C\setminus A}f(e\mid A)=0=f(C)-f(A)=\gamma(f(C)-f(A))$. 


If our \MGPE\ starts from $B=\{b_1,b_2\}$ and deletes $b_i$, then we can show that, for any $\tau\in[\epsilon,1]$
\begin{equation*}
\begin{aligned}
\nabla_{b_i}F(\tau\1_{B-b_i})-
 \nabla_{g_i}F(\tau\1_{B-b_i})=\E_{R\sim\1_{B-b_i}}[f(b_i|R)-f(g_i|R)].
\end{aligned}
 \end{equation*}From the definition of multilinear extension, we know that either $R=\varnothing$ or $R=\{b_j\}$ with $j\ne i$, so $f(b_i\mid R)-f(g_i\mid R)=\delta$. So the process remains at $B$. In this scenario, we have 
\begin{equation*}
    \frac{f(B)}{\max_{S\in\cI}f(S)}\le\frac{f(B)}{f(\{g_1,g_2\})}=s_{\gamma}+2\delta.
\end{equation*}
 
Note that when $\log(1-s_\gamma)+\gamma>0$, we have  $s_\gamma<1-e^{-\gamma}$. So we next investigate the property of $\Phi(\gamma)=\log(1-s_\gamma)+\gamma$ over $\gamma\in(0,1)$. Specifically, the derivative
\begin{equation*}
 \Phi'(\gamma)=
 \frac{(1+\gamma)^3-3}{(\gamma-2)(\gamma+1)(\gamma+2)}
\end{equation*}
shows that $\Phi$ first increases and then decreases on $(0,1)$.
Together with $\Phi(0)=0$ and, when $\gamma_0=0.857633$, $\log(1-s_{\gamma_0})+\gamma_0=0$, we get
Proposition~\ref{prop:strict-barrier}.

\subsection{Factor-Revealing Linear Programming}\label{app:factor-revealing}
In this appendix, we aim to extend the previous rank-two bad case. Like the prior subsection, we suppose the ground set $U$ consisting of $k$ bad elements $\{b_1,\dots, b_k\}$ and $k$ good elements $\{g_1,\dots, g_k\}$. Then, we take the partition
matroid with blocks $\{b_i,g_i\}$, each of capacity one. Similarly, we suppose the obejective function $f$ only depends on the numbers of blocks in the bad-only, good-only, and double-occupied states. Specifically, let 
$(p,q,r)$ be the numbers of bad-only, good-only, and double-occupied blocks of $T$, and then $f(T)\triangleq v_{p,q,r}$.  

Then, we consider a linear programming minimizes $v_{k,0,0}$ subject to
\begin{equation}\label{eq:factor-lp}
\begin{array}{ll}
 v_{0,0,0}=0,\quad v_{0,k,0}=1,\quad v_{p,q,r}\ge0,
 &\text{normalization},\\
 \text{all single-element monotonicity inequalities},
 &\text{monotonicity},\\
 \text{all orbit forms of \eqref{eq:weak-submod}},
 &\text{weak submodularity},\\
 v_{m,k-m,0}\le1\quad(0\le m\le k),
 &\text{optimality of the all-good base},\\
 v_{m,1,0}\le v_{m+1,0,0}\quad(0\le m<k),
 &\text{trapping of the all-bad base}.
\end{array}
\end{equation} Here, the orbit system of weak submodularity is exhaustive over all $A\subsetneq C\subseteq U$. Note that the first three families of
constraints ensure that the induced set function $f$ is normalized,
monotone, and $\gamma$-weakly submodular. The final two families create
the desired trapping structure: the all-good base is optimal, whereas
the all-bad base admits no improving single-element exchange and
therefore traps the \MGPE\ algorithm. More specifically, we can show

\begin{lemma}[Meaning of a feasible LP point]
\label{lem:factor-lp-semantics}Every feasible point of \eqref{eq:factor-lp} defines a normalized
monotone $\gamma$-weakly submodular function for which the all-good base is
optimal.  If we start at the all-bad base, our proposed \MGPE\ remains at the all-bad base,
so $R_{\MGPE}(\gamma)\le v_{k,0,0}$, where $R_{\MGPE}(\gamma)$ denotes the worst-case approximation ratio of
\MGPE\ for matroid-constrained $\gamma$-weakly submodular maximization problems. 
\end{lemma}
\begin{proof}
Only the trapping claim needs detailed explanation.  After deleting $b_i$, we can show that
\begin{equation*}
\nabla_{b_i}F(\tau\1_{B-b_i})-\nabla_{g_i}F(\tau\1_{B-b_i})=\sum_{m=0}^{k-1}\binom{k-1}{m}\tau^m(1-\tau)^{k-1-m}
 \bigl(v_{m+1,0,0}-v_{m,1,0}\bigr)\ge0,
\end{equation*} where the final inequality comes from the last line of Eq.\eqref{eq:factor-lp}.  Consequently, when initialized at the all-bad base, \MGPE\ remains trapped there. As a result, $R_{\MGPE}(\gamma)\le\frac{v_{k,0,0}}{v_{0,k,0}}=v_{k,0,0}$.

\end{proof}

\section{Proofs of Section~\ref{sec:beyond-barrier}}\label{app:beyond-barrier}
In this section, we mainly present the proofs in of Section~\ref{sec:beyond-barrier}. Before that, we establish a result analogous to Lemma~\ref{lem:repair}.
\begin{lemma}[Two-sided marginal repair]
\label{lem:two-sided-repair}Let $\zeta\coloneqq\max\{\gamma,2-\beta\}$ and $C_\zeta(\tau)\coloneqq
\zeta+(1-\zeta)\tau$. When $f$ is monotone and $(\gamma,\beta)$-weakly
submodular, for any feasible exchange $(i,j)$, we have
\begin{equation}\label{eq:two-sided-repair}
\nabla_jF\bigl(\tau\1_{S-i}\bigr)
\geq
\frac{\zeta}{C_\zeta(\tau)}
\E_{R\sim\tau\1_S}\bigl[f(j\mid R)\bigr]
-
\frac{(1-\zeta)\tau}{C_\zeta(\tau)}
\nabla_iF\bigl(\tau\1_{S-i}\bigr),
\end{equation}
where $R$ is the random subset drawn from the point
$\tau\1_S$.
\end{lemma}
\begin{proof}
We first consider the case $j=i$. According to the definition of multilinear extension, we can show
\begin{align}
\E_{R\sim\tau\1_S}\bigl[f(i\mid R)\bigr]=
(1-\tau)
\E_{R_0\sim\tau\1_{S-i}}
\bigl[f(i\mid R_0)\bigr]=
(1-\tau)
\nabla_iF\bigl(\tau\1_{S-i}\bigr).
\label{eq:two-sided-self-marginal}
\end{align}
Therefore, we can infer that
\begin{align*}
&\frac{\zeta}{C_\zeta(\tau)}
\E_{R\sim\tau\1_S}\bigl[f(i\mid R)\bigr]
-
\frac{(1-\zeta)\tau}{C_\zeta(\tau)}
\nabla_iF\bigl(\tau\1_{S-i}\bigr)\\&\leq
\frac{\zeta}{C_\zeta(\tau)}
\E_{R\sim\tau\1_S}\bigl[f(i\mid R)\bigr]=
\frac{\zeta(1-\tau)}{C_\zeta(\tau)}
\nabla_iF\bigl(\tau\1_{S-i}\bigr)\leq
\nabla_iF\bigl(\tau\1_{S-i}\bigr),
\end{align*}
where the final inequality follows from $\zeta(1-\tau)
\leq
\zeta+(1-\zeta)\tau
=
C_\zeta(\tau)$. Thus, Eq.~\eqref{eq:two-sided-repair} holds when $j=i$. We next consider the case $j\notin S$. Draw $R_0\sim\tau\1_{S-i}$. According to the definition of multilinear extension, we can show
\begin{align}
\E_{R\sim\tau\1_S}\bigl[f(j\mid R)\bigr]
&=
(1-\tau)
\E_{R_0}\bigl[f(j\mid R_0)\bigr]
+
\tau
\E_{R_0}\bigl[f(j\mid R_0+i)\bigr].
\label{eq:two-sided-marginal-decomposition}
\end{align}

We now derive a pointwise comparison between
$f(j\mid R_0)$ and $f(j\mid R_0+i)$. Fix a realization of $R_0$
and define $a\coloneqq f(i\mid R_0),
\qquad
b\coloneqq f(j\mid R_0)$, and set 
\[
c\coloneqq
f(R_0+i+j)-f(R_0+i)-f(R_0+j)+f(R_0).
\]
These definitions imply
\begin{align}
f(R_0+i+j)-f(R_0)
&=a+b+c,                                      \label{eq:two-sided-gap}\\
f(j\mid R_0+i)
&=b+c,                                        \label{eq:two-sided-j-after-i}\\
f(i\mid R_0+j)
&=a+c.                                        \label{eq:two-sided-i-after-j}
\end{align}

Applying lower weak submodularity with
$A=R_0$ and $B=R_0+i+j$ yields
\begin{align}
f(i\mid R_0)+f(j\mid R_0)
&\geq
\gamma\bigl(f(R_0+i+j)-f(R_0)\bigr),
\end{align}
or equivalently, $a+b\geq\gamma(a+b+c).$

We next apply upper weak submodularity with the same sets
$A=R_0$ and $B=R_0+i+j$. This gives $f(i\mid R_0+j)+f(j\mid R_0+i)
\leq\beta\bigl(f(R_0+i+j)-f(R_0)\bigr)$. Using Eqs.~\eqref{eq:two-sided-gap}--\eqref{eq:two-sided-i-after-j},
we obtain
\[
(a+c)+(b+c)
\leq
\beta(a+b+c),
\]
and hence $a+b\geq(2-\beta)(a+b+c)$.

Combining the aformentioned two results, we get
\begin{equation}\label{eq:two-sided-zeta-pair}
a+b\geq\zeta(a+b+c),
\end{equation} where $\zeta\coloneqq\max\{\gamma,2-\beta\}$.
Rearranging Eq.~\eqref{eq:two-sided-zeta-pair}, we have $\zeta c\leq(1-\zeta)(a+b)$ such that
\begin{align*}
f(j\mid R_0+i)=b+c\leq
b+\frac{1-\zeta}{\zeta}(a+b)=
\frac{1}{\zeta}b
+
\frac{1-\zeta}{\zeta}a.
\end{align*}
Equivalently, $f(j\mid R_0)
\geq
\zeta f(j\mid R_0+i)
-
(1-\zeta)f(i\mid R_0)$.

Taking expectations over
$R_0\sim\tau\1_{S-i}$ yields
\begin{align}
\nabla_jF\bigl(\tau\1_{S-i}\bigr)
&=
\E_{R_0}\bigl[f(j\mid R_0)\bigr]\notag\\
&\geq
\zeta
\E_{R_0}\bigl[f(j\mid R_0+i)\bigr]
-
(1-\zeta)
\E_{R_0}\bigl[f(i\mid R_0)\bigr]\notag\\
&=
\zeta
\E_{R_0}\bigl[f(j\mid R_0+i)\bigr]
-
(1-\zeta)
\nabla_iF\bigl(\tau\1_{S-i}\bigr).
\label{eq:two-sided-expected-repair}
\end{align}

Multiplying Eq.~\eqref{eq:two-sided-expected-repair} by $\tau$
gives
\begin{align}
\tau\nabla_jF\bigl(\tau\1_{S-i}\bigr)\geq
\zeta\tau
\E_{R_0}\bigl[f(j\mid R_0+i)\bigr]-
(1-\zeta)\tau
\nabla_iF\bigl(\tau\1_{S-i}\bigr).
\label{eq:two-sided-expected-repair-scaled}
\end{align}
From Eq.~\eqref{eq:two-sided-marginal-decomposition}, we have
\begin{align*}
\tau
\E_{R_0}\bigl[f(j\mid R_0+i)\bigr]
&=
\E_{R\sim\tau\1_S}\bigl[f(j\mid R)\bigr]-
(1-\tau)
\E_{R_0}\bigl[f(j\mid R_0)\bigr]\\
&=
\E_{R\sim\tau\1_S}\bigl[f(j\mid R)\bigr]
-
(1-\tau)
\nabla_jF\bigl(\tau\1_{S-i}\bigr).
\end{align*}
Substituting this identity into
Eq.~\eqref{eq:two-sided-expected-repair-scaled}, we obtain
\begin{align*}
\tau\nabla_jF\bigl(\tau\1_{S-i}\bigr)\geq
\zeta
\E_{R\sim\tau\1_S}\bigl[f(j\mid R)\bigr]-
\zeta(1-\tau)
\nabla_jF\bigl(\tau\1_{S-i}\bigr)-
(1-\zeta)\tau
\nabla_iF\bigl(\tau\1_{S-i}\bigr).
\end{align*}So we prove Eq.\eqref{eq:two-sided-repair}
\end{proof}

\subsection{Proof of Lemma~\ref{lem:repair1}}\label{app:proof-repair1}
\begin{proof}
Let $O$ be a optimal base. At first, we consider the $\alpha$-weakly DR-submodular objectives. From Lemma~\ref{lem:conditioned-generator}, we can know that 
\begin{align}
  \left.\frac{\partial^+V_{S,\tau}(r)}{\partial r}\right|_{r=\tau}
 =\sum_{i\in S}\nabla_{j_i^*(S,\tau)}F(\tau\1_{S-i}),
\end{align} where $j_i^*(S,\tau)\in\argmax_{j\in\mathcal C_S(i)}
\nabla_jF(\tau\1_{S-i})$. Fix a Brualdi map $h:S\to O$ and draw $R\sim t\1_S$.  Moreover, we order $O\setminus R=\{o_1,\ldots,o_m\}$ and let $i_\ell=h^{-1}(o_\ell)$. Then, for every $\ell\in[m]$, $R\setminus\{i_\ell\}\subseteq R\cup\{o_1,\ldots,o_{\ell-1}\}$. Therefore, weak DR-submodularity gives
\begin{align}
& \sum_{i\in S}f(h(i)\mid R-i)
 \ge\sum_{\ell=1}^m f(o_\ell\mid R-i_\ell)\notag\\
 &\ge\alpha\sum_{\ell=1}^m
 f(o_\ell\mid R+o_1+\cdots+o_{\ell-1})=\alpha\bigl(f(R\cup O)-f(R)\bigr).
 \label{eq:alpha-telescope}
\end{align}
After that, by taking the expectation of $R$ and using monotonicity $f(R\cup O)\ge f(O)$,  we get
\begin{align*}
  \left.\frac{\partial^+V_{S,\tau}(r)}{\partial r}\right|_{r=\tau}&=\sum_{i\in S}\nabla_{j_i^*}F(\tau\1_{S-i})
 \ge\sum_{i\in S}\nabla_{h(i)}F(\tau\1_{S-i})\\&=\E_{R\sim \tau\1_S}\left[\sum_{i\in S}f(h(i)\mid R-i)\right]\ge\alpha\bigl(f(O)-F(\tau\1_S)\bigr).
\end{align*}
Thus, we get Eq.\eqref{eq:gamma-drift1++} and Eq.\eqref{eq:gamma-drift2++}. 

Next, we prove the results for $(\gamma,\beta)$-weakly submodular objectives. At first, we set $\zeta=\max\{\gamma,2-\beta\}\in(0,1]$. After that, we fix a Brualdi map $h:S\to O$. Then, from Lemma~\ref{lem:two-sided-repair}, we can show
\begin{equation*}
\begin{aligned}
 \nabla_{j_i^{*}(S,\tau)}F(\tau\1_{S-i})\ge\nabla_{h(i)}F(\tau\1_{S})&\geq
\frac{\zeta}{C_\zeta(\tau)}
\E_{R\sim\tau\1_S}\bigl[f(h(i)\mid R)\bigr]
-
\frac{(1-\zeta)\tau}{C_\zeta(\tau)}
\nabla_iF\bigl(\tau\1_{S-i}\bigr)\\
&\geq
\frac{\zeta}{C_\zeta(\tau)}
\E_{R\sim\tau\1_S}\bigl[f(h(i)\mid R)\bigr]
-
\frac{(1-\zeta)\tau}{C_\zeta(\tau)}
\nabla_{j_i^{*}(S,\tau)}F(\tau\1_{S-i}),
\end{aligned}   
\end{equation*} where the first inequality comes from Lemma~\ref{lem:two-sided-repair}. As a result, we have
\begin{equation*}
    \nabla_{j_i^{*}(S,\tau)}F(\tau\1_{S-i})\ge\frac{\zeta}{\zeta+2(1-\zeta)\tau}\E_{R\sim\tau\1_S}\bigl[f(h(i)\mid R)\bigr].
\end{equation*}Similarly, from Lemma~\ref{lem:conditioned-generator}, we have
\begin{equation*}
    \begin{aligned}
       \left.\frac{\partial^+V_{S,\tau}(r)}{\partial r}\right|_{r=\tau}&=\sum_{i\in S}\nabla_{j_i^*}F(\tau\1_{S-i})\\&\ge\frac{\zeta}{\zeta+2(1-\zeta)\tau}\sum_{i\in S}\E_{R\sim \tau\1_{S}}\left[f(h(i)|R)\right]=\frac{\gamma\zeta}{\zeta+2(1-\zeta)\tau}(f(O)-F(\tau\1_{S})).
    \end{aligned}
\end{equation*} Thus, we get Eq.\eqref{eq:gamma-drift1+} and Eq.\eqref{eq:gamma-drift2+}.
\end{proof}

\subsection{Poorf of Theorem~\ref{thm:offline3}}\label{app:proof-offline3}
According to Lemma~\ref{lem:repair1}, we can show that, when $f$ is monotone $(\gamma,\beta)$-weakly submodular,
\begin{align}
   \left.\frac{\partial^+V_{S,\tau}(r)}{\partial r}\right|_{r=\tau}
 \ge \frac{\gamma\zeta}{\zeta+2(1-\zeta)\tau}
    \bigl(f(O)-F(\tau\1_S)\bigr)
\end{align} where  $\zeta \coloneqq \max\{\gamma,2-\beta\}$. Furthermore, if $f$ is monotone $\alpha$-weakly DR-submodular, we also have
\begin{align}
   \left.\frac{\partial^+V_{S,\tau}(r)}{\partial r}\right|_{r=\tau}
 \ge \alpha
    \bigl(f(O)-F(\tau\1_S)\bigr).
\end{align}

Then, by substituting $\eta(\tau)=\frac{\gamma\zeta}{\zeta+2(1-\zeta)\tau}$ or $\eta(\tau)=\alpha$ and $\delta=0$ into Lemma~\ref{lem:potential-comparison}, we get 
\begin{equation*}
   Q(1)\ge e^{-H(1)}Q(\epsilon)+(1-e^{-H(1)})f(O). 
\end{equation*}

Here, when $f$ is monotone $(\gamma,\beta)$-weakly submodular
\begin{equation*}
\left.H(1)\triangleq\int_\epsilon^1\eta(u)\,du=\frac{\gamma\zeta}{2(1-\zeta)}\log(\zeta+2(1-\zeta)\tau)\right|_{\tau=\epsilon}^{1}=\frac{\gamma\zeta}{2(1-\zeta)}\log(\frac{2-\zeta}{\zeta+2(1-\zeta)\epsilon}).
\end{equation*}
Then, we have 
\begin{equation*}
    \begin{aligned}
         Q(1)&\ge e^{-H(1)}Q(\epsilon)+(1-e^{-H(1)})f(O)\\
         &\ge(1-e^{-H(1)})f(O)=\left(1-(\frac{\zeta+2(1-\zeta)\epsilon}{2-\zeta})^{\frac{\gamma\zeta}{2(1-\zeta)}}\right)f(O)\ge(1-R_{\gamma,\beta}-\epsilon)f(O),
    \end{aligned}
\end{equation*} where the second  inequality comes from $Q(\epsilon)\ge0$ and the final inequality comes from Lemma~\ref{lem:finite-start-correction} as well as $R_{\gamma,\beta}\triangleq(\frac{\zeta}{2-\zeta})^{\frac{\gamma\zeta}{2(1-\zeta)}}$. Also, when $f$ is monotone $\alpha$-weakly DR-submodular
\begin{equation*}
H(1)\triangleq\int_\epsilon^1\eta(u)\,du=\alpha(1-\epsilon)
\end{equation*}
Then, we have 
\begin{equation*}
    \begin{aligned}
         Q(1)&\ge e^{-H(1)}Q(\epsilon)+(1-e^{-H(1)})f(O)\ge(1-e^{-\alpha(1-\epsilon)})f(O)\ge(1-\alpha-\epsilon)f(O),
    \end{aligned}
\end{equation*} where the second inequality comes from $Q(\epsilon)\ge0$.

\section{Further Details on Fractional Exchange}\label{app:fractional-exchange}
In this section, we first establish the existence of a fractional exchange
matrix $Q_S^{\x}$ and then present an efficient algorithm to calculate it. 

\subsection{Proof of Theorem~\ref{thm:balanced}}\label{app:proof-balanced}
\begin{proof}At first, from the definition of matroid polytope $P(\cM)$, we know that, for any $\x\in P(\cM)$, there are coefficients $\lambda_B\ge0$ with
\begin{equation}\label{eq:base-decomposition}
 \sum_{B\in\cB}\lambda_B\le 1,
 \qquad \x=\sum_{B\in\cB}\lambda_B\1_B.
\end{equation}Then, for any base $S\in\cB$, choose a Brualdi bijection $h_{S,B}:S\to B$ and define
\begin{equation}\label{eq:balanced-mixture-construction}
 Q_S^{\x}(i,j)\triangleq\sum_{B\in\cB}\lambda_B\1\{h_{S,B}(i)=j\}.
\end{equation}After that, we verify the matrix $ Q_S^{\x}(i,j)$ in Eq.\eqref{eq:balanced-mixture-construction} satisfies the three conditions in Definition~\ref{def:balanced}. At first, if the element $Q_S^{\x}(i,j)$ is positive, then there exists a base $B$ such that  $j=h_{S,B}(i)$, so $S-i+j\in\cB$; For a fixed row $i$, every bijection chooses one image, so $\sum_jQ_S^{\x}(i,j)=\sum_{B\in\cB}\lambda_B\le 1$.  For a fixed column $j$, bijectivity implies that exactly one $i\in S$ maps to $j$ when $j\in B$.  Therefore $\sum_iQ_S^{\x}(i,j)=\sum_B\lambda_B\1\{j\in B\}=x_j$,
where the last equality is coordinate $j$ of the definition of $\x$ in Eq.\eqref{eq:base-decomposition}.  Thus, we prove the existence of fractional exchange matrix.

In the remainder of this subsection, for any fixed $\x\in P(\cM)$ and $S\in\cB$, we show how to compute its induced fractional exchange matrix $Q_S^{\x}(i,j)$. 

\begin{figure}[t]
    \centering
    \includegraphics[width=0.75\linewidth]
    {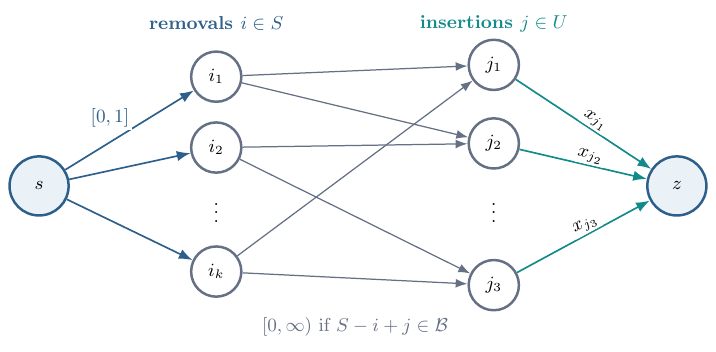}
    \caption{The directed network used to construct the fractional
    exchange matrix $Q_S^{\x}$.}
    \label{fig:fractional-exchange-flow}
\end{figure}

We first construct a directed network with source $s$ and sink $z$.
Between them, we introduce two layers of vertices: a left vertex for
each $i\in S$ and a right vertex for each $j\in U$. Them, for every
$i\in S$, we add an arc $s\to i$ with capacity at most one. For every pair
$(i,j)\in S\times U$ satisfying $S-i+j\in\cB$, we add an arc
$i\to j$ with unbounded capacity. Finally, for every $j\in U$, we add
an arc $j\to z$ with capacity $x_j$, as illustrated in
Figure~\ref{fig:fractional-exchange-flow}. We then seek a feasible
$s$-$z$ flow of value $\sum_{j\in U}x_j$, where $\x\triangleq(x_1,\ldots,x_n)$. Given such a feasible flow
$\varphi$, we define
\[
Q_S^{\x}(i,j)\triangleq
\begin{cases}
\varphi(i,j),
    & \text{if } S-i+j\in\cB,\\
0,
    & \text{otherwise}.
\end{cases}
\]
The source capacities ensure that the row sums of $Q_S^{\x}$ are at most
one, the capacity of the arcs entering $z$ gives the prescribed
column sums $x_j$, and the construction of the exchange arcs guarantees
feasible support. Hence, $Q_S^{\x}$ is a fractional exchange matrix.

Because the resulting network contains $k+n$ non-terminal vertices and at most
$kn$ exchange arcs, according to the maximum-flow algorithm of
\citet{orlin2013maxflows},  a feasible flow can therefore be computed in $O\bigl(kn(k+n)\bigr)$. 
\end{proof}
\subsection{Proof of Lemma~\ref{thm:fractional-exchange-properties}}\label{app:proof-fractional-exchange-properties}
\begin{proof}We first prove point (i). If there exist some $i\in S$ making $\sum_{j\in U}Q^{\x}_{S}(i,j)<1$, we then can show $\sum_{i\in S}\sum_{j\in U}Q^{\x}_{S}(i,j)<1$. But, from the definition of Insertion marginals, we have $\sum_{i\in S}Q^{\x}_{S}(i,j)=x_j$ such that $\sum_{j\in U}\sum_{i\in S}Q^{\x}_{S}(i,j)=\sum_{j\in U}x_j=k$. By contradiction, we have $\sum_{j\in U}Q^{\x}_{S}(i,j)=1$ for every $i\in $.

As for point (ii), for any $j\notin S$, we have $\sum_{i\in S}Q^{\1_{S}}_{S}(i,j)=0$ such that $Q^{\1_{S}}_{S}(i,j)=0$ for any $j\notin S$ and $i\in S$. Then, for any $j\in S$, from the feasible support condition, we know that, only when $i=j$, $S-i+j\in\cB$ such that $Q^{\1_{S}}_{S}(i,j)=0$ for any $j\in S\setminus\{i\}$. Moreover, according to the previous point (i), when $j\in S$, we have $Q^{\1_{S}}_{S}(i,j)=\sum_{j\in U}Q^{\x}_{S}(i,j)=1$. Finally, we have when $\x=\1_S$, then  $Q_S^{\1_S}(i,j)=\1\{i=j\}$ for every $i\in S$ and $j\in U$.\end{proof}
\section{Proofs of Section~\ref{sec:online-gamma}}\label{app:online-gamma}

\subsection{Proof of Lemma~\ref{lem:conditioned-generator1}}\label{app:proof-conditioned-generator1}
\begin{proof}
By applying Lemma~\ref{lem:core_lemma+} to objective function $f_t$, we can know that, for every $\tau\in[\varepsilon,1)$ and $S\in\cB$,
\begin{align}
 \left.\frac{\partial^+V^{t}_{S,\tau}(r)}{\partial r}\right|_{r=\tau}=\sum_{i\in S}\nabla_iF_t(\tau\1_S)
  +\tau\lambda(\tau)\E_{(i,j)\sim p(\tau,S)}
  \left[\nabla_jF_t(\tau\1_{S-i})-\nabla_iF_t(\tau\1_{S-i})\right], 
\end{align} where $p_{i,j}(\tau,S)$ denotes the probability that we delete
\(i\) and insert \(j\), conditional on a Poisson event occurring at time
\(\tau\) while the current base is \(S\)

Note that, in Lemma~\ref{lem:conditioned-generator1}, we suppose the rate $ \lambda(\tau)=\frac{k}{\tau}$ and we select exchange $(i,j)\sim\frac{Q^{\x}_{S}(i,j)}{k}$ such that 
\begin{equation*}
    \begin{aligned}
      \left.\frac{\partial^+V^{t}_{S,\tau}(r)}{\partial r}\right|_{r=\tau}&=\sum_{i\in S}\nabla_iF_t(\tau\1_S)
  +\tau\lambda(\tau)\E_{(i,j)\sim p(\tau,S)}
  \left[\nabla_jF_t(\tau\1_{S-i})-\nabla_iF_t(\tau\1_{S-i})\right]\\
  &=\sum_{i\in S}\nabla_iF_t(\tau\1_S)
  +k\cdot \E_{(i,j)\sim \frac{Q^{\x}_{S}}{k}}
  \left[\nabla_jF_t(\tau\1_{S-i})-\nabla_iF_t(\tau\1_{S-i})\right]\\
  &=\sum_{i\in S}\nabla_iF_t(\tau\1_S)
  +\sum_{i\in S}\sum_{j\in U}Q^{\x}_{S}(i,j)
  \left(\nabla_jF_t(\tau\1_{S-i})-\nabla_iF_t(\tau\1_{S-i})\right)\\
  &=\sum_{i\in S}\sum_{j\in U}Q^{\x}_{S}(i,j)\nabla_jF_t(\tau\1_{S-i}),
    \end{aligned}
\end{equation*} where the final identity comes from $\sum_{j\in U}Q^{\x}_{S}(i,j)=1$.
\end{proof}
\subsection{Proof of Lemma~\ref{add_important}}\label{app:proof-online-gamma-drift}
\begin{proof}
According to Lemma~\ref{lem:repair} and the definition of $Q_S^{\x}$,  we can have\begin{equation}\label{eq:fractional-generator_inequality}
    \begin{aligned}
&\sum_{i\in S}\sum_{j\in U}
    Q_S^{\x}(i,j)
    \nabla_jF_t\bigl(\tau\1_{S-i}\bigr)\\&\ge\frac{\gamma}{C_\gamma(\tau)}\sum_{j\in U}\sum_{i\in S}Q_S^{\x}(i,j)
       \E_{R\sim\tau\1_{S}}[f_t(j\mid R)]-\frac{(1-\gamma)\tau}{C_\gamma(\tau)}
       \sum_{i\in S}\sum_{j\in U}Q_S^{\x}(i,j)\nabla_iF_t(\tau\1_{S-i}),\\
&=\frac{\gamma}{C_\gamma(\tau)}\sum_{j\in U}x_i\cdot
       \E_{R\sim\tau\1_{S}}[f_t(j\mid R)]-\frac{(1-\gamma)\tau}{C_\gamma(\tau)}
       \sum_{i\in S}\nabla_iF_t(\tau\1_{S-i}),
    \end{aligned} 
\end{equation}where $C_\gamma(\tau)=\gamma+(1-\gamma)\tau$. Furthermore, from the point \textbf{(ii)} in Theorem~\ref{thm:fractional-exchange-properties}, we can rewrite 
\begin{equation*}
   \sum_{i\in S}\nabla_iF_t(\tau\1_{S-i})=\sum_{i\in S}\sum_{j\in U}Q^{\1_{S}}_{S}(i,j)\nabla_iF_t(\tau\1_{S-j}).
\end{equation*}As a result, if we set $a_{\gamma}(\tau)=\frac{\gamma+(1-\gamma)\tau}
        {\gamma+2(1-\gamma)\tau}$, we can infer that
\begin{equation*}
\begin{aligned}
 &\sum_{i\in S}\sum_{j\in U}\left(a_{\gamma}(\tau)Q_S^{\x}(i,j)+(1-a_{\gamma}(\tau))Q_S^{\1_{S}}(i,j)\right)
    \nabla_jF_t\bigl(\tau\1_{S-i}\bigr)\\&\ge\frac{\gamma}{\gamma+2(1-\gamma)\tau}\sum_{j\in U}x_i\cdot
       \E_{R\sim\tau\1_{S}}[f_t(j\mid R)]       
\end{aligned}
\end{equation*}From the definition of multilinear extension, we know that the $j$-th coordinate of $\textbf{g}_{t,S,\tau}$ satisfies $g_{t,S,\tau}(j)=\E_{R\sim\tau\1_{S}}[f_t(j\mid R)]$. So we have 
\begin{equation*}
\begin{aligned}
 &\sum_{i\in S}\sum_{j\in U}\left(a_{\gamma}(\tau)Q_S^{\x}(i,j)+(1-a_{\gamma}(\tau))Q_S^{\1_{S}}(i,j)\right)
    \nabla_jF_t\bigl(\tau\1_{S-i}\bigr)\\&\ge\frac{\gamma}{\gamma+2(1-\gamma)\tau}\sum_{j\in U}x_i\cdot
       \E_{R\sim\tau\1_{S}}[f_t(j\mid R)]=\frac{\gamma}{\gamma+2(1-\gamma)\tau}\langle\x,\textbf{g}_{t,S,\tau}\rangle\\
       &=\frac{\gamma}{\gamma+2(1-\gamma)\tau}\langle\1_{O},\textbf{g}_{t,S,\tau}\rangle+\frac{\gamma}{\gamma+2(1-\gamma)\tau}\langle\x-\1_{O},\textbf{g}_{t,S,\tau}\rangle\\
      &=\frac{\gamma}{\gamma+2(1-\gamma)\tau}\sum_{j\in O}\cdot
       \E_{R\sim\tau\1_{S}}[f_t(j\mid R)]+\frac{\gamma}{\gamma+2(1-\gamma)\tau}\langle\x-\1_{O},\textbf{g}_{t,S,\tau}\rangle\\
       &\ge\frac{\gamma^{2}}{\gamma+2(1-\gamma)\tau}(f_t(O)-F_{t}(\tau 1_{S}))+\frac{\gamma}{\gamma+2(1-\gamma)\tau}\langle\x-\1_{O},\mathbf{g}_{t,S,\tau}\rangle,
\end{aligned}
\end{equation*}where the final inequality comes from the $\gamma$-weak submodularity.
\end{proof}
\subsection{Proof of Theorem~\ref{thm:ope-gamma-regret}}\label{app:proof-ope-gamma-regret}
\begin{proof}
 In this subsection, we mainly focus on the verification of Theorem~\ref{thm:ope-gamma-regret}. Before going into the details, we define some useful notation. Firstly, we rewrite the evolution of fractional point $\x_{t,r},\forall r\in N$ in \texttt{OPE} algorithm as $\x_t(\tau),\forall\tau\in[\epsilon,1]$. Specifically, for any $\tau\in[\epsilon,1]$, we set $r^{max}(\tau)\triangleq\max_{0\le r\le N}\{r:\tau_r\le\tau\}$ with $\tau_0=\epsilon$. Then, we can infer that 
\begin{equation*}
\x_t(\tau)\triangleq\x_{t,r^{max}(\tau)},\forall\tau\in[\epsilon,1].
\end{equation*}Similarly, we also reset the the evolution of fractional point $S_{t,r},\forall r\in N$ in \texttt{OPE} algorithm as 
\begin{equation*}
    S_t(\tau)\triangleq S_{t,r^{max}(\tau)},\forall\tau\in[\epsilon,1].
\end{equation*}
Then, according to Lemma~\ref{lem:conditioned-generator1} and Line 9 in \texttt{OPE} algorithm, we can know that
\begin{equation*}
    \begin{aligned}
      \left.\frac{\partial^+V^{t}_{S_t(\tau),\tau}(r)}{\partial r}\right|_{r=\tau}=\sum_{i\in S_t(\tau)}\sum_{j\in U}\left(a_{t}(\tau)Q^{\x_t(\tau)}_{S_t(\tau)}(i,j)+(1-a_{t}(\tau)Q^{\1_{S_t(\tau)}}_{S_t(\tau)}(i,j))\right)\nabla_jF_t(\tau\1_{S_t(\tau)-i}).
    \end{aligned}
\end{equation*} Here, we also set the evolution of $a_{t,r},\forall r\in N$ in \texttt{OPE} algorithm as $ a_t(\tau)\triangleq a_{t,r^{max}(\tau)},\forall\tau\in[\epsilon,1]$.

Next, from Lemma~\ref{add_important}, we get 
\begin{equation*}
    \begin{aligned}
         \left.\frac{\partial^+V^{t}_{S_t(\tau),\tau}(r)}{\partial r}\right|_{r=\tau}&=\sum_{i\in S_t(\tau)}\sum_{j\in U}\left(a_{t}(\tau)Q^{\x_t(\tau)}_{S_t(\tau)}(i,j)+(1-a_{t}(\tau)Q^{\1_{S_t(\tau)}}_{S_t(\tau)}(i,j))\right)\nabla_jF_t(\tau\1_{S_t(\tau)-i})\\
&=\sum_{i\in S_t(\tau)}\sum_{j\in U}\left(a_{\gamma}(\tau)Q^{\x_t(\tau)}_{S_t(\tau)}(i,j)+(1-a_{\gamma}(\tau)Q^{\1_{S_t(\tau)}}_{S_t(\tau)}(i,j))\right)\nabla_jF_t(\tau\1_{S_t(\tau)-i})\\
&+\left(a_t(\tau)-a_{\gamma}(\tau)\right)\cdot\sum_{i\in S_t(\tau)}\sum_{j\in U}\left(Q^{\x_t(\tau)}_{S_t(\tau)}(i,j)-Q^{\1_{S_t(\tau)}}_{S_t(\tau)}(i,j))\right)\nabla_jF_t(\tau\1_{S_t(\tau)-i})
         \\&\ge\frac{\gamma^{2}}{\gamma+2(1-\gamma)\tau}(f_t(O)-F_{t}(\tau 1_{S_t(\tau)}))+\frac{\gamma}{\gamma+2(1-\gamma)\tau}\langle\x_t(\tau)-\1_{O},\mathbf{g}_{t,S_t(\tau),\tau}\rangle\\&+\left(a_t(\tau)-a_{\gamma}(\tau)\right)\cdot\sum_{i\in S_t(\tau)}\sum_{j\in U}\left(Q^{\x_t(\tau)}_{S_t(\tau)}(i,j)-Q^{\1_{S_t(\tau)}}_{S_t(\tau)}(i,j))\right)\nabla_jF_t(\tau\1_{S_t(\tau)-i}).
    \end{aligned}
\end{equation*}For notational convenience, we also define \begin{equation*}
    \Delta_{t,r}\triangleq\sum_{i\in S_{t,r-1}}\sum_{j\in U}
\left(
    Q_{S_{t,r-1}}^{\x_{t,r}}(i,j)
    -
    Q_{S_{t,r-1}}^{\1_{S_{t,r-1}}}(i,j)
\right)
\nabla_jF_t\bigl(\tau_r\1_{S_{t,r-1}-i}\bigr).
\end{equation*} and set $\Delta_t(\tau)\triangleq\Delta_{t,r^{max}(\tau)},\forall\tau\in[\epsilon,1]$.

As a result, we have 
\begin{equation}\label{add_sum}
    \begin{aligned}
\left.\frac{\partial^+\left(\sum_{t=1}^{T}V^{t}_{S_t(\tau),\tau}(r)\right)}{\partial r}\right|_{r=\tau}&\ge \frac{\gamma^{2}}{\gamma+2(1-\gamma)\tau}\left(\sum_{t=1}^{T}f_t(O)-\sum_{t=1}^{T}F_{t}(\tau 1_{S_t(\tau)})\right)\\&+\frac{\gamma}{\gamma+2(1-\gamma)\tau}\sum_{t=1}^{T}\langle\x_t(\tau)-\1_{O},\mathbf{g}_{t,S_t(\tau),\tau}\rangle+\sum_{t=1}^{T}\left(a_t(\tau)-a_{\gamma}(\tau)\right)\Delta_t(\tau).
    \end{aligned}
\end{equation}

If we set $Q^{(T)}(\tau)\triangleq\sum_{t=1}^{T}\E[F_{t}(\tau 1_{S_t(\tau)})]$, from Eq.\eqref{add_sum}, we have 
\begin{equation*}
    \begin{aligned}
        \left.\frac{\partial^+Q^{(T)}(r)}{\partial r}\right|_{r=\tau}&\ge \frac{\gamma^{2}}{\gamma+2(1-\gamma)\tau}\left(\sum_{t=1}^{T}f_t(O)-Q^{(T)}(\tau)\right)\\&+\frac{\gamma}{\gamma+2(1-\gamma)\tau}\sum_{t=1}^{T}\langle\x_t(\tau)-\1_{O},\mathbf{g}_{t,S_t(\tau),\tau}\rangle+\sum_{t=1}^{T}\left(a_t(\tau)-a_{\gamma}(\tau)\right)\Delta_t(\tau).
    \end{aligned}
\end{equation*}

Furthermore, under Assumption~\ref{ass:ope-oracles}, we can have 
\begin{equation}\label{add_regret}
\begin{aligned}
      \E\left[\sum_{t=1}^{T}\langle\1_{O}-\x_t(\tau),\mathbf{g}_{t,S_t(\tau),\tau}\rangle\right]&\le O\left(k\sqrt{\log\left(n/k\right)T}\right),\\
      \E\left[\sum_{t=1}^{T}\left(a_t(\tau)-a_{\gamma}(\tau)\right)\Delta_t(\tau)\right]&\le O\left(k\sqrt{T}\right).
\end{aligned}
\end{equation}

By applying the previous regret results Eq.\eqref{add_regret} about oracle $\mathcal{E}$ and weight oracle $\mathcal{A}$, we immediately obtain:
\begin{equation*}
    \left.\frac{\partial^+Q^{(T)}(r)}{\partial r}\right|_{r=\tau}\ge \frac{\gamma^{2}}{\gamma+2(1-\gamma)\tau}\left(\sum_{t=1}^{T}f_t(O)-Q^{(T)}(\tau)\right)-M_{0}\left(k\sqrt{\log\left(n/k\right)T}+k\sqrt{T}\right),
\end{equation*} where we assume $M_0$ is the hidden large constant in Eq.\eqref{add_regret}.

Finally, according to the Lemma~\ref{lem:potential-comparison} in Appendix~\ref{app:auxiliary-results}, we can show that
\begin{equation*}
    \begin{aligned}
       Q^{(T)}(1)&\ge e^{-H(1)}Q^{(T)}(\epsilon)+(1-e^{-H(1)}Q^{(T)})\cdot\sum_{t=1}^{T}f_t(O) \\
       &-M_{0}\left(k\sqrt{\log\left(n/k\right)T}+k\sqrt{T}\right)\int_{\epsilon}^{1}\exp{\left(-\left(H(1)-H(s)\right)\right)}\mathrm{d}s,
    \end{aligned}
\end{equation*} where $H(\tau)\triangleq\int_{\epsilon}^{\tau}\frac{\gamma^{2}}{\gamma+2(1-\gamma)u}\mathrm{d}u=\frac{\gamma^{2}}{2(1-\gamma)}\log(\frac{\gamma+2(1-\gamma)\tau}{\gamma+2(1-\gamma)\epsilon})$. Note that $H(\tau)=\frac{\gamma^{2}}{2(1-\gamma)}\log(\frac{\gamma+2(1-\gamma)\tau}{\gamma+2(1-\gamma)\epsilon})$ is a non-decreasing function about the variable $\tau$ such that  $\exp{\left(-\left(H(1)-H(s)\right)\right)}\le 1$.

As a result, we have 
\begin{equation*}
    \begin{aligned}
        Q^{(T)}(1)=\sum_{t=1}^{T}\E[f_{t}(S_t)]&\ge (1-e^{-H(1)}Q^{(T)})\cdot\sum_{t=1}^{T}f_t(O)-M_{0}\left(k\sqrt{\log\left(n/k\right)T}+k\sqrt{T}\right)\\
        &=\left(1-(\frac{\gamma+2(1-\gamma)\epsilon}{2-\gamma})^{\frac{\gamma^{2}}{2(1-\gamma)}}\right)\cdot\sum_{t=1}^{T}f_t(O)-M_{0}\left(k\sqrt{\log\left(n/k\right)T}+k\sqrt{T}\right)\\&\ge(1-\rho_{\gamma}-\epsilon)\sum_{t=1}^{T}f_t(O)-M_{0}\left(k\sqrt{\log\left(n/k\right)T}+k\sqrt{T}\right),
    \end{aligned}
\end{equation*} where the first inequality comes from $Q^{(T)}(\epsilon)\ge0$ and the final inequality comes from Lemma~\ref{lem:finite-start-correction} as well as $\rho_{\gamma}\triangleq(\frac{\gamma}{2-\gamma})^{\frac{\gamma^{2}}{2(1-\gamma)}}$. So we can get Eq.\eqref{eq:ope-gamma-regret} of Theorem~\ref{thm:ope-gamma-regret}.
\end{proof}
\section{Proofs of Section~\ref{sec:online-two-sided}}\label{app:online-two-sided}
\subsection{Proof of Lemma~\ref{add_important1}}\label{app:proof-online-two-sided-drift}
\begin{proof}
According to Lemma~\ref{lem:two-sided-repair} and the definition of $Q_S^{\x}$,  we can have\begin{equation}\label{eq:fractional-generator_inequality+}
    \begin{aligned}
&\sum_{i\in S}\sum_{j\in U}
    Q_S^{\x}(i,j)
    \nabla_jF_t\bigl(\tau\1_{S-i}\bigr)\\&\ge\frac{\zeta}{C_\zeta(\tau)}\sum_{j\in U}\sum_{i\in S}Q_S^{\x}(i,j)
       \E_{R\sim\tau\1_{S}}[f_t(j\mid R)]-\frac{(1-\zeta)\tau}{C_\zeta(\tau)}
       \sum_{i\in S}\sum_{j\in U}Q_S^{\x}(i,j)\nabla_iF_t(\tau\1_{S-i}),\\
&=\frac{\zeta}{C_\zeta(\tau)}\sum_{j\in U}x_i\cdot
       \E_{R\sim\tau\1_{S}}[f_t(j\mid R)]-\frac{(1-\zeta)\tau}{C_\zeta(\tau)}
       \sum_{i\in S}\nabla_iF_t(\tau\1_{S-i}),
    \end{aligned} 
\end{equation}where $C_\zeta(\tau)=\zeta+(1-\zeta)\tau$. Furthermore, from the point \textbf{(ii)} in Theorem~\ref{thm:fractional-exchange-properties}, we can rewrite 
\begin{equation*}
   \sum_{i\in S}\nabla_iF_t(\tau\1_{S-i})=\sum_{i\in S}\sum_{j\in U}Q^{\1_{S}}_{S}(i,j)\nabla_iF_t(\tau\1_{S-j}).
\end{equation*}As a result, if we set $a_{\zeta}(\tau)=\frac{\zeta+(1-\zeta)\tau}
        {\zeta+2(1-\zeta)\tau}$, we can infer that
\begin{equation*}
\begin{aligned}
 &\sum_{i\in S}\sum_{j\in U}\left(a_{\zeta}(\tau)Q_S^{\x}(i,j)+(1-a_{\zeta}(\tau))Q_S^{\1_{S}}(i,j)\right)
    \nabla_jF_t\bigl(\tau\1_{S-i}\bigr)\\&\ge\frac{\zeta}{\zeta+2(1-\zeta)\tau}\sum_{j\in U}x_i\cdot
       \E_{R\sim\tau\1_{S}}[f_t(j\mid R)]       
\end{aligned}
\end{equation*}From the definition of multilinear extension, we know that the $j$-th coordinate of $\textbf{g}_{t,S,\tau}$ satisfies $g_{t,S,\tau}(j)=\E_{R\sim\tau\1_{S}}[f_t(j\mid R)]$. So we have 
\begin{equation*}
\begin{aligned}
 &\sum_{i\in S}\sum_{j\in U}\left(a_{\zeta}(\tau)Q_S^{\x}(i,j)+(1-a_{\zeta}(\tau))Q_S^{\1_{S}}(i,j)\right)
    \nabla_jF_t\bigl(\tau\1_{S-i}\bigr)\\&\ge\frac{\zeta}{\zeta+2(1-\zeta)\tau}\sum_{j\in U}x_i\cdot
       \E_{R\sim\tau\1_{S}}[f_t(j\mid R)]=\frac{\zeta}{\zeta+2(1-\zeta)\tau}\langle\x,\textbf{g}_{t,S,\tau}\rangle\\
       &=\frac{\zeta}{\zeta+2(1-\zeta)\tau}\langle\1_{O},\textbf{g}_{t,S,\tau}\rangle+\frac{\zeta}{\zeta+2(1-\zeta)\tau}\langle\x-\1_{O},\textbf{g}_{t,S,\tau}\rangle\\
      &=\frac{\zeta}{\zeta+2(1-\zeta)\tau}\sum_{j\in O}\cdot
       \E_{R\sim\tau\1_{S}}[f_t(j\mid R)]+\frac{\zeta}{\zeta+2(1-\zeta)\tau}\langle\x-\1_{O},\textbf{g}_{t,S,\tau}\rangle\\
       &\ge\frac{\zeta\gamma}{\gamma+2(1-\zeta)\tau}(f_t(O)-F_{t}(\tau 1_{S}))+\frac{\zeta}{\zeta+2(1-\zeta)\tau}\langle\x-\1_{O},\mathbf{g}_{t,S,\tau}\rangle,
\end{aligned}
\end{equation*}where the final inequality comes from the $\gamma$-weak submodularity.
\end{proof}

\subsection{Proof of Theorem~\ref{thm:ope-two-sided-regret}}\label{app:proof-ope-two-sided-regret}
Under Lemma~\ref{lem:two-sided-repair} and Assumption~\ref{ass:ope-oracles}, the proof of
Theorem~\ref{thm:ope-two-sided-regret} follows the same argument as that of
Theorem~\ref{thm:ope-gamma-regret}. We therefore omit the details to avoid repetition.

\section{Proofs of Section~\ref{sec:online-cardinality}}\label{app:online-cardinality}

\subsection{Proof of Lemma~\ref{lem:32}}\label{app:proof-online-cardinality-drift}
\begin{proof} We write $X=O\setminus S$ and $Y=O\cap S$. Then, we can show that 
\begin{equation*}
\begin{aligned}
&\sum_{j\in U}\sum_{i\in S}\bar Q_S^{\x}(i,j)\nabla_jF_t\bigl(\tau\1_{S-i}\bigr)\\&=\sum_{j\in S}\sum_{i\in S}\bar Q_S^{\x}(i,j)\nabla_jF_t\bigl(\tau\1_{S-i}\bigr)+\sum_{j\in U\setminus S}\sum_{i\in S}\bar Q_S^{\x}(i,j)\nabla_jF_t\bigl(\tau\1_{S-i}\bigr)\\
&=\sum_{j\in S}x_j\nabla_j F_t(\tau\1_{S-j})
 +\sum_{j\in U\setminus S}\sum_{i\in S}p_i(\x,S)x_j
   \nabla_jF_t(\tau\1_{S-i})\\&=\langle\x,\ell^{\mathrm{card}}_{t,S,\tau,\x}\rangle=\langle\1_{O},\ell^{\mathrm{card}}_{t,S,\tau,\x}\rangle+\langle\x-\1_O,\ell^{\mathrm{card}}_{t,S,\tau,\x}\rangle.
\end{aligned}
\end{equation*}

Then, from the definition of $\ell^{\mathrm{card}}_{t,S,\tau,\x}$, we can show that
\begin{equation}\label{eq:card-online-comparator-expand}
    \begin{aligned}
\langle\1_{O},\ell^{\mathrm{card}}_{t,S,\tau,\x}\rangle   =\sum_{z\in Y}\nabla_z F_t(\tau\1_{S-z})
 +\sum_{o\in X }\sum_{i\in S}p_i(\x,S)
   \nabla_oF_t(\tau\1_{S-i})
    \end{aligned}
\end{equation}

For $A\subsetneq S$, we set
\begin{equation}\label{eq:card-online-weights}
 b_A=\tau^{|A|}(1-\tau)^{k-1-|A|},
 \qquad
 w_A=b_A\sum_{i\in S\setminus A}p_i(\x,S),
 \qquad w_S=0.
\end{equation}
According to the multilinear expansion of
Eq.\eqref{eq:card-online-comparator-expand}, we can know that the coefficient of
$f_t(o\mid A)$ is $w_A$ for every $o\in X$.  For every
$z\in Y\setminus A$, the coefficient of $f_t(z\mid A)$ is $b_A$. Thus, we can show that
\begin{align}\left\langle\1_O,\ell^{\mathrm{card}}_{t,S,\tau,y}\right\rangle
 &\ge\sum_{A\subseteq S}\sum_{z\in Y\setminus A}b_A\cdot f_t(z\mid A)+\sum_{A\subseteq S}\sum_{o\in X}w_A\cdot f_t(o\mid A)\\&\ge\sum_{A\subseteq S}w_A\sum_{o\in O\setminus A}f_t(o\mid A)\ge\gamma\sum_{A\subseteq S}w_A
       \bigl(f_t(O)-f_t(A)\bigr),
 \label{eq:card-online-weak-average}
\end{align}
where the second inequality comes from  and
$b_A\ge w_A$ due to that $\sum_{i\in S\setminus A}p_i(\x,S)\le1$. The weights form a probability distribution:
\begin{equation}\label{eq:card-online-weight-sum}
 \sum_{A\subseteq S}w_A
 =\sum_{i\in S}p_i(\x,S)\sum_{A\subseteq S-i}
   \tau^{|A|}(1-\tau)^{k-1-|A|}
 =\sum_{i\in S}p_i(\x,S)=1.
\end{equation}
Moreover,
\begin{equation}\label{eq:card-online-weighted-value}
 \sum_{A\subseteq S}w_Af_t(A)
 =\sum_{i\in S}p_i(\x,S)F_t(\tau\1_{S-i})
 \le F_t(\tau\1_S).
\end{equation}
Substituting
Eq.\eqref{eq:card-online-weight-sum} and Eq.\eqref{eq:card-online-weighted-value}
into Eq.\eqref{eq:card-online-weak-average} proves
\begin{equation*}
 \left\langle\1_O,\ell^{\mathrm{card}}_{t,S,\tau,y}\right\rangle\ge\gamma(f_{t}(O)-F_t(\tau\1_S)).   
\end{equation*} Thus, we get the result of Eq.\eqref{eq:online-card-drift}.
\end{proof}
\subsection{Proof of Theorem~\ref{thm:ope-card-regret}}\label{app:proof-ope-card-regret}
\begin{proof}
 In this subsection, we mainly focus on the verification of Theorem~\ref{thm:ope-card-regret}. Before going into the details, we define some useful notation. Firstly, we rewrite the evolution of fractional point $\x_{t,r},\forall r\in N$ in \texttt{OPEC} algorithm as $\x_t(\tau),\forall\tau\in[\epsilon,1]$. Specifically, for any $\tau\in[\epsilon,1]$, we set $r^{max}(\tau)\triangleq\max_{0\le r\le N}\{r:\tau_r\le\tau\}$ with $\tau_0=\epsilon$. Then, we can infer that 
\begin{equation*}
\x_t(\tau)\triangleq\x_{t,r^{max}(\tau)},\forall\tau\in[\epsilon,1].
\end{equation*}Similarly, we also reset the the evolution of fractional point $S_{t,r},\forall r\in N$ in \texttt{OPE} algorithm as 
\begin{equation*}
    S_t(\tau)\triangleq S_{t,r^{max}(\tau)},\forall\tau\in[\epsilon,1].
\end{equation*}
Then, according to Lemma~\ref{lem:32} and Line 8 in \texttt{OPEC} algorithm, we can know that
\begin{equation*}
    \begin{aligned}
      \left.\frac{\partial^+V^{t}_{S_t(\tau),\tau}(r)}{\partial r}\right|_{r=\tau}&=\sum_{i\in S_t(\tau)}\sum_{j\in U}\bar{Q}^{\x_t(\tau)}_{S_t(\tau)}(i,j)\nabla_jF_t(\tau\1_{S_t(\tau)-i})\\
      &\ge\gamma(f_t(O)-F_{t}(\tau 1_{S_t(\tau)}))+\langle\x_t(\tau)-\1_{O},\ell^{\mathrm{card}}_{t,S_t(\tau),\tau,\x_t(\tau)}\rangle.
    \end{aligned}
\end{equation*} 

If we set $Q^{(T)}(\tau)\triangleq\sum_{t=1}^{T}\E[F_{t}(\tau 1_{S_t(\tau)})]$,  we have 
\begin{equation*}
    \begin{aligned}
        \left.\frac{\partial^+Q^{(T)}(r)}{\partial r}\right|_{r=\tau}\ge \gamma\left(\sum_{t=1}^{T}f_t(O)-Q^{(T)}(\tau)\right)+\sum_{t=1}^{T}\langle\x_t(\tau)-\1_{O},\ell^{\mathrm{card}}_{t,S_t(\tau),\tau,\x_t(\tau)}\rangle
    \end{aligned}
\end{equation*}

Furthermore, under Assumption~\ref{ass:ope-oracles2}, we can have 
\begin{equation}\label{add_regret2}
\begin{aligned}
      \E\left[
            \sum_{t=1}^T
            \langle\1_{O}-\x_t(\tau),\ell^{\mathrm{card}}_{t,S_t(\tau),\tau,\x_t(\tau)}\rangle
        \right]
        \leq \mathcal{O}\left(
            k\sqrt{\log\left(n/k\right)T}\right).
\end{aligned}
\end{equation}

By applying the previous regret results Eq.\eqref{add_regret2} about oracle $\mathcal{E}$, we immediately obtain:
\begin{equation*}
    \left.\frac{\partial^+Q^{(T)}(r)}{\partial r}\right|_{r=\tau}\ge \gamma\left(\sum_{t=1}^{T}f_t(O)-Q^{(T)}(\tau)\right)-M_{0}k\sqrt{\log\left(n/k\right)T},
\end{equation*} where we assume $M_0$ is the hidden large constant in Eq.\eqref{add_regret2}.

Finally, according to the Lemma~\ref{lem:potential-comparison} in Appendix~\ref{app:auxiliary-results}, we can show that
\begin{equation*}
    \begin{aligned}
       Q^{(T)}(1)=\sum_{t=1}^{T}\E[f_{t}(S_t)]&\ge e^{-\gamma(1-\epsilon)}Q^{(T)}(\epsilon)+(1-e^{-\gamma(1-\epsilon)})\cdot\sum_{t=1}^{T}f_t(O) \\
     &-M_{0}\cdot k\sqrt{\log\left(n/k\right)T}\cdot\int_{\epsilon}^{1}\exp{\left(-\gamma(1-s)\right)}\mathrm{d}s\\
     &\ge (1-e^{-\gamma(1-\epsilon)})\cdot\sum_{t=1}^{T}f_t(O)-M_{0}\cdot k\sqrt{\log\left(n/k\right)T}\\
     &\ge(1-e^{-\gamma}-\epsilon)\cdot\sum_{t=1}^{T}f_t(O)-M_{0}\cdot k\sqrt{\log\left(n/k\right)T} ,
    \end{aligned}
\end{equation*} where the second inequality comes from $Q^{(T)}(\epsilon)\ge0$ and $\exp{\left(-\gamma(1-s)\right)}\le 1$ when $s\in[\epsilon,1]$, and the final inequality follows from $e^{-\gamma(1-\epsilon)})\le e^{-\gamma}-\epsilon$. So we can get Eq.\eqref{eq:ope-card-regret} of Theorem~\ref{thm:ope-card-regret}.
\end{proof}

\section{Proofs of Section~\ref{sec:online-alpha}}\label{app:online-alpha}
\subsection{Proof of Lemma~\ref{lem:online-alpha-drift}}\label{app:proof-online-alpha-drift}
\begin{proof} From the definition of Eq.\eqref{equ:ell}, we can show that
\begin{equation*}
    \begin{aligned}
      &\sum_{i\in S}\sum_{j\in U}
Q_S^{\x}(i,j)
\nabla_j F_t\bigl(\tau\1_{S-i}\bigr)=\sum_{j\in U}y_j\sum_{i\in S}
\frac{Q_S^{\x}(i,j)}{x_j}
\nabla_j F_t\bigl(\tau\1_{S-i}\bigr)\\
&=\langle\x,\ell_{t,S,\tau,\x}\rangle=\langle\1_{O},\ell_{t,S,\tau,\x}\rangle+\langle
\x-\1_{O},\ell_{t,S,\tau,\mathbf{y}}\rangle.
    \end{aligned}
\end{equation*}  According to the definition of multilinear extension, we have 
\begin{align*}
 \langle\1_{O},\ell_{t,S,\tau,\x}\rangle& =\sum_{o\in O}\sum_{i\in S}
 \frac{Q_S^{\x}(i,o)}{x_o}\nabla_{o} F_t(\tau\1_{S-i})\\&\ge\E_{R\sim\tau\1_S}\left[\sum_{o\in O\setminus R}\sum_{i\in S}
 \frac{Q_S^{\x}(i,o)}{x_o}f_t(o\mid R-i)\right]\\&\ge\alpha\cdot\E_{R\sim\tau\1_S}\left[\sum_{o\in O\setminus R}\sum_{i\in S}
 \frac{Q_S^{\x}(i,o)}{x_o}f_t(o_q\mid R+o_1+\cdots+o_{q-1})\right]\\&=\alpha\cdot\E_{R\sim\tau\1_S}\left[\sum_{o\in O\setminus R}f_t(o_q\mid R+o_1+\cdots+o_{q-1})\right]\\&\ge\alpha(f_t(O)-F_t(\tau\1_S)),
\end{align*} where the second inequality uses $R-i\subseteq R+o_1+\cdots+o_{q-1}$ and weak DR-submodularity.
\end{proof}
\subsection{Proof of Theorem~\ref{thm:ope-DR-regret}}\label{app:proof-ope-dr-regret}
\begin{proof}
 In this subsection, we mainly focus on the verification of Theorem~\ref{thm:ope-DR-regret}. Before going into the details, we define some useful notation. Firstly, we rewrite the evolution of fractional points $\x_{t,r},\mathbf{y}_{t,r},\forall r\in N$ in \texttt{OPE-DR} algorithm as $\x_t(\tau),\mathbf{y}_{t}(\tau),\forall\tau\in[\epsilon,1]$. Specifically, for any $\tau\in[\epsilon,1]$, we set $r^{max}(\tau)\triangleq\max_{0\le r\le N}\{r:\tau_r\le\tau\}$ with $\tau_0=\epsilon$. Then, we can infer that 
\begin{equation*}
\x_t(\tau)\triangleq\x_{t,r^{max}(\tau)},\quad\mathbf{y}_t(\tau)\triangleq\mathbf{y}_{t,r^{max}(\tau)},\quad\forall\tau\in[\epsilon,1].
\end{equation*}Similarly, we also reset the the evolution of fractional point $S_{t,r},\forall r\in N$ in \texttt{OPE} algorithm as 
\begin{equation*}
    S_t(\tau)\triangleq S_{t,r^{max}(\tau)},\forall\tau\in[\epsilon,1].
\end{equation*}
Then, according to Lemma~\ref{lem:online-alpha-drift} and Line 9 in \texttt{OPE-DR} algorithm, we can know that
\begin{equation*}
    \begin{aligned}
      \left.\frac{\partial^+V^{t}_{S_t(\tau),\tau}(r)}{\partial r}\right|_{r=\tau}&=\sum_{i\in S_t(\tau)}\sum_{j\in U}Q^{\x_t(\tau)}_{S_t(\tau)}(i,j)\nabla_jF_t(\tau\1_{S_t(\tau)-i})\\
      &\ge\alpha(f_t(O)-F_{t}(\tau 1_{S_t(\tau)}))+\langle\x_t(\tau)-\1_{O},\ell_{t,S_t(\tau),\tau,\x_t(\tau)}\rangle.
    \end{aligned}
\end{equation*} 

If we set $Q^{(T)}(\tau)\triangleq\sum_{t=1}^{T}\E[F_{t}(\tau 1_{S_t(\tau)})]$,  we have 
\begin{equation*}
    \begin{aligned}
        \left.\frac{\partial^+Q^{(T)}(r)}{\partial r}\right|_{r=\tau}\ge \alpha\left(\sum_{t=1}^{T}f_t(O)-Q^{(T)}(\tau)\right)+(1-\epsilon)\sum_{t=1}^{T}\langle\x_t(\tau)-\1_{O},\ell_{t,S_t(\tau),\tau,\x_t(\tau)}\rangle
    \end{aligned}
\end{equation*}

Furthermore, under Assumption~\ref{ass:ope-oracles1}, we can have 
\begin{equation}\label{add_regret1}
\begin{aligned}
      \E\left[
            \sum_{t=1}^T\langle\1_{O}-\x_t(\tau),\ell_{t,S_t(\tau),\tau,\x_t(\tau)}\rangle
        \right]
        \leq \mathcal{O}\left(
            k\sqrt{\log\left(n/k\right)T}\right).
\end{aligned}
\end{equation}

By applying the previous regret results Eq.\eqref{add_regret1} about oracle $\mathcal{E}$, we immediately obtain:
\begin{equation*}
    \left.\frac{\partial^+Q^{(T)}(r)}{\partial r}\right|_{r=\tau}\ge \alpha\left(\sum_{t=1}^{T}f_t(O)-Q^{(T)}(\tau)\right)-(1-\epsilon)M_{0}k\sqrt{\log\left(n/k\right)T},
\end{equation*} where we assume $M_0$ is the hidden large constant in Eq.\eqref{add_regret1}.

Finally, according to the Lemma~\ref{lem:potential-comparison} in Appendix~\ref{app:auxiliary-results}, we can show that
\begin{equation*}
    \begin{aligned}
       Q^{(T)}(1)=\sum_{t=1}^{T}\E[f_{t}(S_t)]&\ge e^{-\alpha(1-\epsilon)}Q^{(T)}(\epsilon)+(1-e^{-\alpha(1-\epsilon)})\cdot\sum_{t=1}^{T}f_t(O) \\
     &-M_{0}k\sqrt{\log\left(n/k\right)T}\cdot \int_{\epsilon}^{1}\exp{\left(-\alpha(1-s)\right)}\mathrm{d}s\\
     &\ge (1-e^{-\alpha(1-\epsilon)})\cdot\sum_{t=1}^{T}f_t(O)-M_{0}\cdot k\sqrt{\log\left(n/k\right)T}\\
     &\ge(1-e^{-\alpha}-\epsilon)\cdot\sum_{t=1}^{T}f_t(O)-M_{0}\cdot k\sqrt{\log\left(n/k\right)T},
    \end{aligned}
\end{equation*} where the second inequality comes from $Q^{(T)}(\epsilon)\ge0$ and $\exp{\left(-\alpha(1-s)\right)}\le 1$ when $s\in[\epsilon,1]$, and the final inequality comes from  $e^{-\alpha(1-\epsilon)}\le e^{-\alpha}-\epsilon$. So we can get Eq.\eqref{eq:ope-alpha-regret} of Theorem~\ref{thm:ope-DR-regret}.
\end{proof}

\section{Online Linear Oracles}\label{app:online-linear-oracles}
\label{sec:online-linear-oracles}
In this appendix, we mainly focus on constructing two 
stochastic online linear oracles used in our proposed online algorithms.

\subsection{Stochastic Online Linear Optimization over a Matroid Polytope}
\label{subsec:matroid-linear-oracle}
In this subsection, we suppose that $\cM=(U,\cI)$ is a matroid of rank $k$ with $n=|U|$ and $P(\cM)\triangleq\operatorname{conv}\{\1_I:I\in\cI\}
$ denotes the corresponding matroid polytope. 

Then, we consider a online optimization over $P(\cM)$. Specifically, at each round $t\in[T]$, the online leaner outputs a predictable
point $\x_t\in P(\cM)$ and subsequently receives a stochastic gradient
estimator $\tilde{\mathbf g}_t$ satisfying $\E_t[\tilde{\mathbf g}_t]=\mathbf g_t$ and $\tilde g_t(j)\in[0,G], \forall j\in U$, where $\E_t[\cdot]$ denotes conditional expectation given the history
available when $\x_t$ is selected and $\tilde g_t(j)$ is the $j$-th coordinate of $\tilde{\mathbf g}_t$.

Furthermore, for $\mathbf{x}\in\R_+^U$, we define the negative entropy as follows:
\[
\Psi(\mathbf{x})\coloneqq
\sum_{j\in U}\bigl(x_j\log x_j-x_j\bigr),
\]
with the convention $0\log 0=0$.  Its associated Bregman divergence is the KL divergence, namely,
\[
D_\Psi(\mathbf{x},\mathbf z)
\coloneqq
\sum_{j\in U}
\left(x_j\log\frac{x_j}{z_j}-x_j+z_j\right).
\]

With these preliminaries in place, we now recall the standard online entropic mirror-ascent algorithm~\citep{hazan2022introduction}, presented in Algorithm~\ref{alg:matroid-entropy-omd}. 

\begin{coltalgorithm}[t]
\caption{Stochastic Online entropic mirror ascent over $P(\cM)$}
\label{alg:matroid-entropy-omd}
\KwIn{Matroid polytope $P(\cM)$, horizon $T$, and feedback bound
$G$}
Set $\Lambda_{n,k}\leftarrow1+\log(n/k)$ and
$\eta\leftarrow
\min\{1,\sqrt{\Lambda_{n,k}/T}\}/G$\;
Compute the entropy center $\mathbf{x}_1\in\argmin_{\mathbf{x}\in P(\cM)}\Psi(\mathbf{x})$\;
\For{$t=1,\ldots,T$}{
    Output $\x_t$ and receive a stochastic feedback vector
    $\tilde{\mathbf g}_t\in[0,G]^U$\;
    Set
    $z_{t+1}(j)\leftarrow
    x_{t}(j)\exp\bigl(\eta\cdot
    \tilde g_t(j)\bigr)$ for every $j\in U$\;
    Set
    $\displaystyle
    \x_{t+1}\leftarrow
    \argmin_{\mathbf{x}\in P(\cM)}D_\Psi(\mathbf{x},\mathbf z_{t+1})$\;
}
\end{coltalgorithm}

Algorithm~\ref{alg:matroid-entropy-omd} first uses the entropy center $\mathbf{x}_1\in\argmin_{\mathbf{x}\in P(\cM)}\Psi(\mathbf{x})$ as its initial decision. Then, at Line~5, it performs a multiplicative mirror-ascent update using the
stochastic feedback vector $\widetilde{\mathbf g}_t\in[0,G]^U$. Finally, in Line~6, we maps the resulting point back to the feasible region
$P(\cM)$ through a Bregman projection. Next, we show that every point returned by
Algorithm~\ref{alg:matroid-entropy-omd} is full-rank.

\subsubsection{Full-Rank Property of Algorithm~\ref{alg:matroid-entropy-omd}}
We first recall a fractional augmentation property of matroid
polytopes. Before that, for any point $\x\in P(\cM)$, we define $\x(U)=\sum_{i\in U}x_i$. Then, we can show that
\begin{lemma}[Fractional augmentation]
\label{lem:E-fractional-augmentation}
Let $\mathbf a,\mathbf b\in P(\cM)$ satisfy
$a(U)<b(U)$. Then there exist an element $j\in U$ and a scalar
$\delta>0$ such that
\[
a_j<b_j
\qquad\text{and}\qquad
\mathbf a+\delta\mathbf e_j\in P(\cM).
\]
\end{lemma}

\begin{proof}
Let $J\coloneqq\{j\in U:a_j<b_j\}$. We suppose, toward a contradiction, that no coordinate in $J$ can be
increased while preserving feasibility. Then, for every $j\in J$,
there exists a set $A_j\subseteq U$ containing $j$ that is tight at
$\mathbf a$, i.e.,
\[
a(A_j)=r_{\cM}(A_j),
\] where $r_{\cM}(A)\triangleq\max_{I\in\cI,I\subseteq A}|I|$ is the rank function of matroid $\cM$. Note that the family of tight sets is closed under unions. More specifically, due to the submodularity of rank function $r_{\cM}$, for any two tight sets $C$ and $D$, we have
\begin{align*}
&r_{\cM}(C\cup D)+r_{\cM}(C\cap D)\ge
a(C\cup D)+a(C\cap D)
\\&=a(C)+a(D)=r_{\cM}(C)+r_{\cM}(D)\\
&\geq r_{\cM}(C\cup D)+r_{\cM}(C\cap D)\geq a(C\cup D)+a(C\cap D),
\end{align*} where the first equality comes from the linearity of $a$. As a result, if we set $A\coloneqq\bigcup_{j\in J}A_j$, we can show  $a(A)=r_{\cM}(A)$. 
For every $j\notin A$, we have $j\notin J$ and hence
$b_j\leq a_j$. Consequently,
\[
b(U)
=
b(A)+b(U\setminus A)
\leq
r_{\cM}(A)+a(U\setminus A)=a(A)+a(U\setminus A)
=
a(U),
\]
contradicting $a(U)<b(U)$.  
\end{proof}

With the fractional augmentation property, we next investigate the entropy center $\x_1$, that is, 
\begin{lemma}[Entropy-center properties]
\label{lem:E-entropy-center}
The initialization $\x_1$ is strictly positive and full-rank. Moreover, for every $\mathbf u\in P(\cM)$,
\begin{equation}\label{eq:E-entropy-diameter}
D_\Psi(\mathbf u,\x_1)
\leq
k\log\left(\frac{en}{k}\right).
\end{equation}
\end{lemma}

\begin{proof} Before going into the details, we fix a base $B^{(u)}\in\cB$ containing $u$ for every $u\in U$ and define
\begin{equation}\label{eq:E-positive-witness}
\overline{\x}
\coloneqq
\frac{1}{n}
\sum_{u\in U}\1_{B^{(u)}}.
\end{equation}

We first establish the positivity of $\x^{o}$. We suppose, for
contradiction, that $\x_1(j)=0$ for some $j\in U$, and set  $Z\coloneqq\{i\in U: \x_1(i)=0\}$. Then, for any $\varepsilon\in(0,1)$, let
\[
\x_\varepsilon
\coloneqq
(1-\varepsilon)\x_1
+\varepsilon\overline{\x}.
\]
By the convexity of $P(\cM)$, we have
$\x_\varepsilon\in P(\cM)$. Recall that $\Psi(\x)
=
\sum_{j\in U}\bigl(x_j\log x_j-x_j\bigr)$, where $0\log 0\coloneqq0$. For every $j\in Z$, we have
$\x_{\varepsilon}(j)=\varepsilon\overline{\x}(j)$, and hence
\begin{align*}
&\x_{\varepsilon}(j)\log \x_{\varepsilon}(j)
-\x_{\varepsilon}(j)
-\bigl(\x_{1}(j)\log \x_{1}(j)-\x_{1}(j)\bigr)\\
&\qquad
=
\varepsilon\overline{\x}(j)\log\varepsilon
+\varepsilon\overline{\x}(j)
\bigl(\log \overline{\x}(j)-1\bigr).
\end{align*}
On the other hand, for each $j\notin Z$, the function
$u\mapsto u\log u-u$ is continuously differentiable in a neighborhood
of $\x_1(j)>0$. Therefore,
\[
\x_{\varepsilon}(j)\log \x_{\varepsilon}(j)
-\x_{\varepsilon}(j)
-\bigl(\x_{1}(j)\log \x_{1}(j)-\x_{1}(j)\bigr)
=
\varepsilon
(\overline{\x}(j)-\x_{1}(j))\log \x_{1}(j)
+O(\varepsilon^2).
\]
Summing over all coordinates gives
\[
\Psi(\x_\varepsilon)-\Psi(\x_1)
=
\varepsilon
\left(\sum_{j\in Z}\overline{\x}(j)\right)\log\varepsilon
+O(\varepsilon).
\]
Since $\overline{\x}(j)>0$ for every $j\in U$ and $Z\neq\varnothing$, we have $\sum_{j\in Z}\overline{\x}(j)>0$. Consequently,
\[
\frac{\Psi(\x_\varepsilon)-\Psi(\x_1)}{\varepsilon}
=
\left(\sum_{j\in Z}\overline{\x}(j)\right)\log\varepsilon
+O(1)
\longrightarrow-\infty
\qquad\text{as }\varepsilon\downarrow0.
\]
Thus, $\Psi(\x_\varepsilon)<\Psi(\x_1)$ for all sufficiently small
$\varepsilon>0$, contradicting the optimality of $\x_1$. Hence,
$\x_1(j)>0$ for every $j\in U$.

Then, we suppose that $\x_1(U)<k=\overline{\x}(U)$. By
Lemma~\ref{lem:E-fractional-augmentation}, there exist $j\in U$ and
$\delta>0$ such that
\[
\x_1(j)<\overline{\x}(j)
\qquad\text{and}\qquad
\x_1+\delta\mathbf e_j\in P(\cM).
\]
Because $\x_1(j)<1$, $\frac{\partial\Psi(\x_1)}{\partial x_j}
=
\log(\x_1(j))<0$. Therefore, increasing coordinate $j$
slightly strictly will decrease $\Psi$, again contradicting the optimality of $\x_1$. Hence $\x_1(U)=k$. Finally, the first-order optimality gives $\left\langle
\nabla\Psi(\x_1),
\mathbf u-\x_1
\right\rangle
\geq0$, and hence
\[
D_\Psi(\mathbf u,\x_1)
\leq
\Psi(\mathbf u)-\Psi(\x_1).
\]
Since every coordinate of $\mathbf u$ lies in $[0,1]$,
$\Psi(\mathbf u)\leq0$. Moreover, we also know that 
\[
\min_{\x(U)=k,\x\ge0}\Psi(\x)=\Psi(\frac{k\cdot\1_{U}}{n})= k\log\left(\frac{k}{n}\right)-k.
\]
Since $\x_1(U)=k$, it follows that $\Psi(\x_1)
\geq
k\log\left(\frac{k}{n}\right)-k$. Consequently,
\[
D_\Psi(\mathbf u,\x_1)
\leq
k\log\left(\frac{n}{k}\right)+k
=
k\log\left(\frac{en}{k}\right).
\]
\end{proof}

Finally, we prove all iterates of Algorithm~\ref{alg:matroid-entropy-omd} is positive and full-rank, that is,
\begin{lemma}
\label{lem:E-full-rank}
Every iterate returned by Algorithm~\ref{alg:matroid-entropy-omd} is strictly
positive and full-rank.
\end{lemma}

\begin{proof}
The claim holds at initialization by
Lemma~\ref{lem:E-entropy-center}. Suppose that $\x_t$ is
strictly positive and satisfies $\x_t(U)=k$. Since
$\tilde{\mathbf g}_t\geq\mathbf0$, $z_{t+1}(j)
=
x_{t}(j)
\exp\bigl(
    \eta_{\mathcal E}\tilde g_t(j)
\bigr)
\geq
x_{t}(j)>0$.

Suppose, for contradiction, that $\x_{t+1}$ has at least one zero
coordinate, and let $Z\coloneqq\{j\in U:\x_{t+1}(j)=0\}$. Then, for any $\varepsilon\in(0,1)$, define
\[
\mathbf{y}_{\epsilon}
\coloneqq
(1-\varepsilon)\x_{t+1}
+\varepsilon\overline{\x},
\]
where $\overline{\x}\in P(\cM)$ is defined in Eq.\eqref{eq:E-positive-witness}. For each coordinate $j\in U$,  we define 
\begin{equation*}
    \phi_j(u)
\coloneqq
u\log\frac{u}{z_{t+1}(j)}-u+z_{t+1}(j),
\end{equation*} so that $D_{\Psi}(\x,\mathbf z_{t+1})
=
\sum_{j\in U}\phi_j(\x(j))$. For every $j\in Z$, we have
$\mathbf{y}_{\epsilon}(j)=\varepsilon\overline{x}(j)$, and therefore
\begin{align*}
\phi_j(\mathbf{y}_{\epsilon}(j))-\phi_j(\x_{t+1}(j))
&=
\varepsilon\overline{x}(j)
\log\frac{\epsilon\overline{x}(j)}{z_{t+1}(j)}
-\varepsilon\overline{x}(j)\\
&=
\varepsilon\overline{x}(j)\log\varepsilon
+\varepsilon\overline{x}(j)
\left(
\log\frac{\overline{x}(j)}{z_{t+1}(j)}-1
\right).
\end{align*}
For every $j\notin Z$, we have $\x_{t+1}(j)>0$. Since $\phi_j$ is continuously
differentiable in a neighborhood of $\x_{t+1}(j)$, with $\phi_j'(\x_{t+1}(j))=\log\frac{\x_{t+1}(j)}{z_{t+1}(j)}$,
a first-order expansion gives
\[
\phi_j(\mathbf{y}_{\epsilon}(j))-\phi_j(\x_{t+1}(j))
=
\varepsilon\cdot(\overline{x}(j)-\x_{t+1}(j))
\log\frac{\x_{t+1}(j)}{z_{t+1}(j)}
+O(\varepsilon^2).
\]
Summing over all coordinates, we obtain
\[
D_{\Psi}(\mathbf y_\varepsilon,\mathbf z_{t+1})
-
D_{\Psi}(\x_{t+1},\mathbf z_{t+1})
=
\varepsilon
\left(\sum_{j\in Z}\overline{x}(j)\right)
\log\varepsilon
+O(\varepsilon).
\]
Because $\overline{\x}$ is strictly positive and $Z\neq\varnothing$, we have 
\[
\frac{
D_{\Psi}(\mathbf y_\varepsilon,\mathbf z_{t+1})
-
D_{\Psi}(\x_{t+1},\mathbf z_{t+1})
}{\varepsilon}
=
\left(\sum_{j\in Z}\overline{x}(j)\right)\log\varepsilon
+O(1)
\longrightarrow-\infty
\]
as $\varepsilon\downarrow0$. Consequently, $D_{\Psi}(\mathbf y_\varepsilon,\mathbf z_{t+1})
<
D_{\Psi}(\x_{t+1},\mathbf z_{t+1})$ for all sufficiently small $\varepsilon>0$. This contradicts the
optimality of $\x_{t+1}$. Hence, every coordinate of
$\x_{t+1}$ is strictly positive.

We next claim that
\begin{equation}\label{eq:E-y-below-z}
\x_{t+1}(j)\leq z_{t+1}(j)
\qquad
\text{for every }j\in U.
\end{equation}
If $\x_{t+1}(j)>z_{t+1}(j)$, we set 
\begin{equation*}
    \mathbf{y}(b)=\x_{t+1}-b\cdot(\x_{t+1}(j)-z_{t+1}(j))\mathbf{e}_j.
\end{equation*}Due to the down-closedness of $P(\cM)$, we have $\mathbf{y}(b)\in P(\cM)$ for any $b\in[0,1]$.Similarly, we can show 
\[
\left.\frac{\mathrm{d}D_\Psi(\mathbf{y}(b),\mathbf z_{t+1})}{\mathrm{d}b}\right|_{b=0}=-(\x_{t+1}(j)-z_{t+1}(j))\log\frac{\x_{t+1}(j)}{z_{t+1}(j)}\le 0.
\] Thus,there exists $b\ge$ such that $D_\Psi(\mathbf{y}(b),\mathbf z_{t+1})\le D_\Psi(\x_{t+1},\mathbf z_{t+1})$.
This contradicts the optimality of $\x_{t+1}$. Suppose now that $\x_{t+1}(U)<k=\x_t(U)$. By Lemma~\ref{lem:E-fractional-augmentation}, there exist $j\in U$ and
$\delta>0$ such that
\[
\x_{t+1}(j)<\x_{t}(j)
\qquad\text{and}\qquad
\x_{t+1}+\delta\mathbf e_j\in P(\cM).
\]
According to Line 5 in Algorithm~\ref{alg:matroid-entropy-omd} and $\tilde{g}_t(t)\ge0$, we also have $x_{t+1}(j)<\x_{t}(j)\leq z_{t+1}(j)$. Hence, for all sufficiently small $\delta'>0$, $\x_{t+1}+\delta'\mathbf e_j\in P(\cM)$ and $\x_{t+1}(j)+\delta'<z_{t+1}(j)$ such that 

\[
\frac{\mathrm{d}D_\Psi(\x_{t+1}+\delta'\mathbf e_j,\mathbf z_{t+1})}{\mathrm{d}\delta'}=\log\frac{\x_{t+1}(j)+\delta'}{z_{t+1}(j)}< 0.
\] 
It follows that
\[
D_\Psi(\x_{t+1}+\delta'\mathbf e_j,\mathbf z_{t+1})
<
D_\Psi(\x_{t+1},\mathbf z_{t+1}),
\]
again contradicting the optimality of $\x_{t+1}$. Therefore, $\x_{t+1}(U)=k$.
\end{proof}

\subsubsection{Regret of Algorithm~\ref{alg:matroid-entropy-omd}}

The following result gives the regret guarantee.  Let
$\{\mathcal F_t\}_{t\geq0}$ denote the natural filtration, and suppose that
$\x_t$ is $\mathcal F_{t-1}$-measurable and
\begin{equation}
\label{eq:vector-unbiased-feedback}
\E[\tilde{\mathbf g}_t\mid\mathcal F_{t-1}]
=\mathbf g_t,
\qquad
\tilde{\mathbf g}_t\in[0,G]^U.
\end{equation}

\begin{theorem}
\label{thm:matroid-entropy-regret}
For every horizon $T\geq1$ and comparator $\mathbf u\in P(\cM)$,
Algorithm~\ref{alg:matroid-entropy-omd} satisfies
\begin{equation}
\label{eq:matroid-entropy-regret}
\E\!\left[
\sum_{t=1}^T
\langle\mathbf u-\x_t,\mathbf g_t\rangle
\right]
\leq
2G\cdot k
\sqrt{T\log\!\left(\frac{en}{k}\right)}.
\end{equation}
\end{theorem}
\begin{remark}
Both Theorem~\ref{thm:matroid-entropy-regret} and Lemma~\ref{lem:E-full-rank} immediately indicate the
validity of condition~(i) in Assumption~\ref{ass:ope-oracles}, as well
as Assumption~\ref{ass:ope-oracles1} and  Assumption~\ref{ass:ope-oracles2}, in which $G=1$.
\end{remark}\begin{proof}The Pythagorean inequality of Bregman divergence  gives
\[
D_\Psi(\mathbf u,\mathbf z_{t+1})
\geq
D_\Psi(\mathbf u,\x_{t+1})
+
D_\Psi(\x_{t+1},\mathbf z_{t+1}).
\]
Since Bregman divergences are nonnegative, it follows that $D_\Psi(\mathbf u,\x_{t+1})
\leq
D_\Psi(\mathbf u,\mathbf z_{t+1})$,

By the definition of the
generalized KL divergence, we also have 
\begin{align*}
&D_\Psi(\mathbf u,\mathbf z_{t+1})
-D_\Psi(\mathbf u,\x_t)\\
&=
\sum_{j\in U}
u_j\log\frac{x_{t,j}}{z_{t+1,j}}
+
\sum_{j\in U}(z_{t+1,j}-x_{t,j})\\
&=-\eta
\left\langle
\mathbf u,\tilde{\mathbf g}_t
\right\rangle
+\sum_{j\in U}x_{t,j}
\left(
e^{\eta\widetilde g_t(j)}-1\right),
\end{align*} where the second equality comes from $\log\frac{x_{t,j}}{z_{t+1,j}}
=
-\eta\widetilde g_t(j)$ and $z_{t+1,j}-x_{t,j}
=
x_{t,j}
\left(
e^{\eta\tilde g_t(j)}-1
\right)$.

Substituting these identities yields
\begin{align}
D_\Psi(\mathbf u,\x_{t+1})
&\leq D_\Psi(\mathbf u,\mathbf z_{t+1})\notag\\
&=
D_\Psi(\mathbf u,\x_t)
-\eta
\left\langle
\mathbf u,\tilde{\mathbf g}_t
\right\rangle
+\sum_{j\in U}x_{t,j}
\left(
e^{\eta\tilde g_t(j)}-1
\right).
\label{eq:kl-one-step-expansion}
\end{align}

Since
$0\leq\eta\tilde g_t(j)\leq1$ and
$e^v\leq1+v+v^2$ for every $v\in[0,1]$, we have
\[
e^{\eta\tilde g_t(j)}-1
\leq
\eta\tilde g_t(j)
+\eta^2\tilde g_t(j)^2.
\]
Applying this inequality to
Eq.~\eqref{eq:kl-one-step-expansion} gives
\begin{align*}
D_\Psi(\mathbf u,\x_{t+1})
&\leq
D_\Psi(\mathbf u,\x_t)
-\eta
\left\langle
\mathbf u,\tilde{\mathbf g}_t
\right\rangle\\
&\quad
+\eta
\left\langle
\x_t,\tilde{\mathbf g}_t
\right\rangle
+\eta^2
\sum_{j\in U}
x_{t,j}\tilde g_t(j)^2\\
&=
D_\Psi(\mathbf u,\x_t)
-\eta
\left\langle
\mathbf u-\x_t,\tilde{\mathbf g}_t
\right\rangle
+\eta^2
\sum_{j\in U}
x_{t,j}\tilde g_t(j)^2.
\end{align*}
Then, the following inequality holds:
\begin{equation}\label{eq:entropy-one-step}
\begin{aligned}
  \E[\left\langle
\mathbf u-\x_t,\mathbf g_t
\right\rangle]&=\E[\left\langle
\mathbf u-\x_t,\E_t[\widetilde{\mathbf g}_t]\right\rangle=\E[\left\langle
\mathbf u-\x_t,\widetilde{\mathbf g}_t
\right\rangle]
\\&\leq
\frac{
\E[D_\Psi(\mathbf u,\x_t)]
-
\E[D_\Psi(\mathbf u,\x_{t+1})]
}{\eta}
+
\eta
\E[\sum_{j\in U}
x_{t,j}\tilde g_t(j)^2].  
\end{aligned}
\end{equation}

Summing Eq.~\eqref{eq:entropy-one-step}, using
$\x_t(U)=k$, and applying $D_\Psi(\mathbf u,\x_1)
\leq k\frac{en}{k}$, we obtain
\begin{equation}
\label{eq:entropy-pathwise-regret}
\sum_{t=1}^T
\E[\left\langle
\mathbf u-\x_t,\mathbf g_t
\right\rangle]
\leq
\frac{k\log(\frac{en}{k})}{\eta}
+\eta G^{2}kT.
\end{equation}
As a result, we have  $\sum_{t=1}^T
\E[\left\langle
\mathbf u-\x_t,\mathbf g_t
\right\rangle]\leq kG\sqrt{T\log(\frac{en}{k}})$.\end{proof}\subsection{One-Dimensional Online Linear Oracle over \texorpdfstring{$[1/2,1]$}{[1/2,1]}}
\label{subsec:one-dimensional-linear-oracle}
In this subsection, we review the online gradient ascent, as shown in Algorithm~\ref{alg:interval-oga}, and set it as the weight oracle in our proposed \texttt{OPE} algorithm. More
precisely, in Algorithm~\ref{alg:interval-oga}, we suppose that $a_t$ is $\mathcal F_{t-1}$-measurable and set $\E[\widetilde h_t\mid\mathcal F_{t-1}]=h_t$ and $\tilde h_t\in[-G,G]$ for some $G>0$. Then, from the Section~2.4 of \citet{shalev2012online}, we can have

\begin{proposition}
\label{prop:interval-oga-regret}
For every horizon $T\geq1$ and comparator $a\in[1/2,1]$,
Algorithm~\ref{alg:interval-oga} satisfies
\begin{equation}
\label{eq:interval-oga-regret}
\E\!\left[
\sum_{t=1}^T(a-a_t)h_t
\right]
\leq G\sqrt{T}/4.
\end{equation}
\end{proposition}
\begin{remark}
Proposition~\ref{prop:interval-oga-regret} immediately implies the
validity of condition~(ii) in Assumption~\ref{ass:ope-oracles}, in which the parameter $G=k$.
\end{remark}

\begin{coltalgorithm}[t]
\caption{Stochastic online gradient ascent over $[1/2,1]$}
\label{alg:interval-oga}
\KwIn{Horizon $T$ and feedback bound $G>0$}
Set $a_1\leftarrow3/4$ and
$\eta_{\mathcal A}\leftarrow1/(4G\sqrt{T})$\;
\For{$t=1,\ldots,T$}{
    Output $a_t\in[1/2,1]$ and receive stochastic scalar feedback
    $\widetilde h_t$\;
    Set
    $\displaystyle
    a_{t+1}\leftarrow
    \Pi_{[1/2,1]}\left(a_t+\eta\cdot\tilde h_t\right)$\;
}
\end{coltalgorithm}
\ifdefined\standaloneappendix
  \bibliography{references}
\fi
\end{document}